\documentclass[11pt]{article}
\usepackage[letterpaper, left=1in, top=1in, right=1in, bottom=1in,nohead,includefoot]{geometry}
\usepackage{color,charter,graphicx,here,setspace,multirow,amsmath,amssymb,amsfonts,enumitem,placeins,amsthm,cancel,mathrsfs,bm,pdflscape,subcaption}
\usepackage[utf8]{inputenc}
\usepackage[authoryear]{natbib}
\usepackage{adjustbox}
\usepackage{indentfirst}
\usepackage{mathrsfs}
\usepackage{amssymb}
\usepackage{amsmath}
\usepackage{bbm}
\usepackage{mathtools}
\usepackage{ascmac}
\usepackage{amsthm}
\usepackage{cmll}
\usepackage{stmaryrd}
\usepackage{caption}
\usepackage{lscape}
\usepackage{subcaption}
\usepackage{natbib}
\usepackage{setspace}
\usepackage{bm}
\usepackage{url}
\usepackage{booktabs}
\usepackage{color}
\usepackage{enumitem}

\usepackage{times}

\usepackage[dvipsnames]{xcolor}
	\usepackage[pdftex]{hyperref}
	\hypersetup{colorlinks,
	citecolor=NavyBlue  
	}

\usepackage{subfiles}

\usepackage[colorinlistoftodos,prependcaption,textsize=tiny]{todonotes}

\usepackage{mathrsfs}
\usepackage{amssymb}
\usepackage{amsmath}
\usepackage{mathtools}
\usepackage{ascmac}
\usepackage{amsthm}
\usepackage{multirow}
\usepackage{proba}
\usepackage{subcaption}
\usepackage{adjustbox}

\usepackage{graphicx}
\usepackage{natbib}
\usepackage{setspace}
\usepackage{bm}
\usepackage{url}

\usepackage{color}

\usepackage{times}

\theoremstyle{definition}
\newtheorem{definition}{Definition}[section]
\newtheorem{assumption}{Assumption}[section]

\newtheorem{lem}{Lemma}[section]
\newtheorem{thm}{Theorem}[section]

\def\C{\mbox{\boldmath$C$}}

\def\zero{\mbox{\boldmath$0$}}

\def\x{{\bm{x}}}

\def\btheta{{\bm{\theta}}}

\def\X{\mbox{\boldmath$X$}}
\def\W{\mbox{\boldmath$W$}}

\def\m{\mbox{\boldmath$m$}}

\def\eqno#1{eqn.~(\ref{eq:#1})}
\def\eqsno#1{eqns.~(\ref{eq:#1})}
\def\Eqno#1{Eqn.~(\ref{eq:#1})}

\def\treat{t_c}

\usepackage{subfiles}

\usepackage{tikz}
\usetikzlibrary{arrows.meta,positioning,shapes,calc}

\begin{document}

\title{Dynamic causal inference with time series data}
\author{
Tanique Schaffe-Odeleye\thanks{Department of Statistical Science, Fox School of Business, Temple University, Philadelphia, PA 19122. {\scriptsize Email: \texttt{tanique.hudson@temple.edu}}},\,
K\={o}saku Takanashi\thanks{Riken AIP, Tokyo, Japan. {\scriptsize Email: \texttt{kosaku.takanashi@riken.jp}}},\\
Vishesh Karwa, Edoardo M.\ Airoldi, \& Kenichiro McAlinn\thanks{Department of Statistical Science, Fox School of Business, Temple University, Philadelphia, PA 19122. {\scriptsize Emails: \texttt{\{vishesh, airoldi, kenichiro.mcalinn\}@temple.edu}}}
}

\maketitle
\thispagestyle{empty}\setcounter{page}{0}
\begin{abstract}
We generalize the potential outcome framework to time series with an intervention by defining causal effects on stochastic processes. Interventions in dynamic systems alter not only outcome levels but also evolutionary dynamics-- changing persistence, volatility, and transition laws. Our framework treats potential outcomes as entire trajectories, enabling causal estimands, identification conditions, and estimators to be formulated directly on path space. The resulting Dynamic Average Treatment Effect (DATE) characterizes how causal effects evolve through time and reduces to the classical average treatment effect under one period of time. For observational data, we derive a dynamic inverse-probability weighting estimator that is unbiased under dynamic ignorability and positivity. When treated units are scarce, we show that conditional mean trajectories underlying the DATE admit a linear state-space representation, yielding a dynamic linear model implementation. Simulations demonstrate that modeling time as intrinsic to the causal mechanism exposes dynamic effects that static methods systematically misestimate. An empirical study of COVID-19 lockdowns illustrates the framework's practical value for estimating and decomposing treatment effects.

\bigskip{}
\end{abstract}
\newpage{}

\section{Introduction}
Traditional potential outcome causal inference \citep{rubin1974estimating,rubin2005causal}
treats the treated and control outcomes as random variables realized at (or immediately after)
treatment.
In dynamic systems, interventions alter not only outcome levels but also change the evolutionary dynamics, where interventions may fundamentally alter how a series progresses.
Standard static frameworks, therefore, can fail to capture how causal effects unfold through time, and may be biased when interventions alter dynamics.
Given the rise of interest in causal estimation of time series interventions in macroeconomics, epidemiology, and environmental sciences, a causal framework that directly addresses these issues is critical.

We generalize the potential outcome framework by defining causal effects as properties of entire trajectories rather than single-period observations.  
We consider interventions that alter the dynamics of evolution itself, including changing persistence and transition laws, not merely shifting levels or slopes.
Each unit's outcome is treated as a stochastic process, and potential outcomes are defined over full paths under treatment and control.  
The resulting \emph{Dynamic Average Treatment Effect} (DATE) describes how the expected difference between treated and control processes evolves through time, reducing to the classic ATE under a single period.

Conceptually, our framework is guided by an ideal dynamic randomized experiment that randomizes the intervention assignment at the time of intervention, conditional on the pre-treatment history (here, we focus on a single, persistent intervention).
Each unit’s evolution under treatment and control, post-treatment, could then be observed repeatedly, and the contrast between these two stochastic processes would reveal the true causal process.  
Randomization at the intervention time  would ensure assignment is independent of the post-intervention potential outcome processes, conditional on the pre-intervention history, providing identification by design.  
Observational studies violate this ideal because policies and behaviors adapt to evolving outcomes; treatment decisions depend on evolving histories, policies adapt dynamically, and outcomes influence subsequent exposure.  
Our framework formalizes what the causal estimand, DATE, represents under perfect dynamic randomization and derives identification and estimation strategies for the observational case, paralleling the logic of the classical potential outcome framework.

Time plays a dual role: it indexes dependence among observations and acts as a potential confounder, as each unit's history influences both treatment assignment and outcomes.
Addressing this dual role requires conditioning on the entire information history rather than treating time as an exogenous covariate.  
This perspective clarifies why methods, such as difference-in-differences, synthetic control, or interrupted time series regression, can yield biased estimates when interventions modify the evolution dynamics and not just the mean level or slope.
When interventions alter the dynamics, these methods conflate causal and temporal dependence.

Building on this foundation, we derive estimation methods suited to different empirical settings.  
When multiple comparable units are available, we extend inverse-probability weighting to dynamic settings by defining treatment propensities relative to each unit’s evolving history.  
The resulting \emph{Dynamic Inverse-Probability Weighting} (DIPW) estimator is unbiased and consistent under time-dependent confounding.  
When few, or even a single, treated units are available, a \emph{Dynamic Linear Model} (DLM) admits a state-space representation for the conditional mean dynamics of the DATE, linking the theoretical framework to practical state-space modeling.

This paper makes three contributions:
\begin{itemize}
\itemsep-0.5em
    \item[\textbf{(i)}] \textbf{Foundation:} extend the Rubin causal model to stochastic processes, establishing a coherent basis for causal inference in dynamic systems;
    \item[\textbf{(ii)}] \textbf{Identification:} define conditions for dynamic identification and derive the DIPW estimator, valid under time-varying confounding;
    \item[\textbf{(iii)}] \textbf{Representation:} show that the DLM admits the state-space representation for the conditional mean dynamics underlying the DATE when data are limited.
\end{itemize}

Section~\ref{sec:prelim} contrasts independent-sample and time series experiments and shows how uncertainty arises from both unit selection and temporal evolution.  
Section~\ref{sec:date} formalizes stochastic potential outcomes and the DATE.  
Section~\ref{sec:dipw} develops the DIPW estimator for observational data, and Section~\ref{sec:dlm} shows that the DLM admits a state-space representation for the conditional mean dynamics underlying the DATE.  
Sections~\ref{sec:sim}--\ref{sec:emp} present simulations and an empirical study of the economic effects of the COVID-19 lockdown, and Section~\ref{sec:summ} concludes with implications for dynamic causal discovery and continuous-time extensions.

\paragraph{Related Work.}
Methods for estimating treatment effects in time-dependent settings fall into several lines.  
Difference-in-differences and event-study designs \citep{card1993minimum,angrist2009mostly,lechner2011estimation} exploit repeated cross-sections but assume parallel trends and time-invariant dynamics.  
Synthetic-control approaches \citep{abadie2010synthetic,abadie2015comparative,abadie2021using} construct counterfactual trajectories as weighted combinations (typically treated as fixed once estimated) of controls.  
Interrupted time series analyses \citep{wagner2002segmented,bernal2017interrupted,mcdowall2019interrupted,ewusie2020methods} model pre- and post-intervention levels or slopes, providing descriptive contrasts rather than process-level causal estimands.  
These frameworks typically do not treat time as the carrier of the causal mechanism, in the sense of defining estimands and identification on stochastic path space; instead they impose structure, like parallel trends or deterministic counterfactual trajectories.
Many applied estimators target the average treatment of the treated (ATT); operationally they contrast observed treated outcomes with an estimated control trajectory.

Structural time series methods, including the widely-used CausalImpact framework \citep{brodersen2015inferring}, construct counterfactual predictions by fitting state-space models to pre-intervention data and projecting forward. While computationally convenient and practically influential, these approaches define causal effects implicitly through model-based counterfactual forecasts rather than through explicit potential outcome processes, and often requires control series to produce stable counterfactual forecasts. Our framework provides a formal foundation that clarifies when such model-based approaches target well-defined causal estimands. It also distinguishes two inferential modes that are often conflated in practice: forward-looking counterfactual prediction based only on pre-treatment data, and retrospective causal attribution under an intervention-augmented state-space model fit to the full series. The latter, which is what we propose and apply in our analyses, uses the full dataset to sharpen the decomposition between baseline dynamics (e.g., trend/persistence/seasonality) and the intervention component, yielding DATE directly from the posterior of the intervention effect rather than from an extrapolated pre-period forecast. Both modes are accommodated within our potential-outcomes formulation, though the focus of this paper is the latter.

Recent contributions have begun to recast these ideas in explicitly causal terms.
\cite{bojinov2019time,bojinov2021panel} formalize sequential potential outcomes for repeated experiments (treating autoregression as network interference, and treatment being randomly assigned per time period). \cite{tierney2023multivariate,west2024dynamic,KevinLiEtAlCausalMVTS2024} explore Bayesian, multivariate state-space modeling for counterfactual predictions.
Our contribution provides a unifying probabilistic foundation that defines potential outcomes, identification, and estimation directly on stochastic processes, encompassing existing methods as special cases.

\paragraph{One framework, two regimes.}
We distinguish two empirical regimes that share the same estimand (the DATE) but rely on different identifying assumptions.
\emph{(i) Multi-unit/design-based regime:} when multiple comparable units are available, identification is based on
Dynamic Ignorability (Assumption~\ref{ass:dynamicignore}), i.e., independence conditional on the \emph{pre-intervention} history,
and estimation proceeds via the DIPW reweighting approach in Section~\ref{sec:dipw}.
\emph{(ii) Scarce-treated/model-based regime:} when few (or a single) treated units are available,
DATE estimation is necessarily model-based, and we posit a state-space model for the conditional mean and infer the intervention
component from the full series under the maintained DLM structure (Section~\ref{sec:dlm} onward).


\section{Conceptual and Theoretical Preliminaries}
\label{sec:prelim}

Classical causal inference is formulated for cross-sectional or repeated-sample data, where each unit realizes a single outcome under treatment or control.  
In time series, uncertainty originates both from unit assignment and from the stochastic evolution of each trajectory.
In time series, outcomes evolve, and treatment itself may depend on past realizations.  
The corresponding causal analysis must, therefore, simultaneously take into account two sources of randomness: variation across units and stochastic evolution within each unit.

Throughout this paper, the intervention is a single assignment at time $t_c$,
represented by $Z\in\{0,1\}$, with assignment modeled conditional on the
pre-intervention information $\mathscr{F}^{X}_{t_c}$. The ``dynamic'' aspect lies
in the outcome: post-intervention potential outcomes are trajectories whose
evolution may differ under $Z=1$ versus $Z=0$, rather than a sequence of treatment decisions.

\subsection{Uncertainty in time series}
In the static, i.i.d. setting, sampling is defined as selecting a unit, $\omega_i$ for $i=1,...,N$, from some population $\Omega$.
The randomness comes from possibly selecting another unit instead of the one chosen.
In the time series setting, however, sampling, and thus randomness, do not pertain to just the choice of unit, but also each unit will have its own uncertainty around its evolution.
Thus, from a population, $\Omega$, a unit is selected, $\Omega_i$ for $i=1,...,N$, from which each observation path, $\omega_i$, is drawn.
This uncertainty occurs because a time series is not simply one outcome, but a path of outcomes that affect each other in the direction of time.
This is illustrated in Figure~\ref{fig:sample}.

\begin{figure}[t!]
\centering
\includegraphics[width=0.75\textwidth,clip]{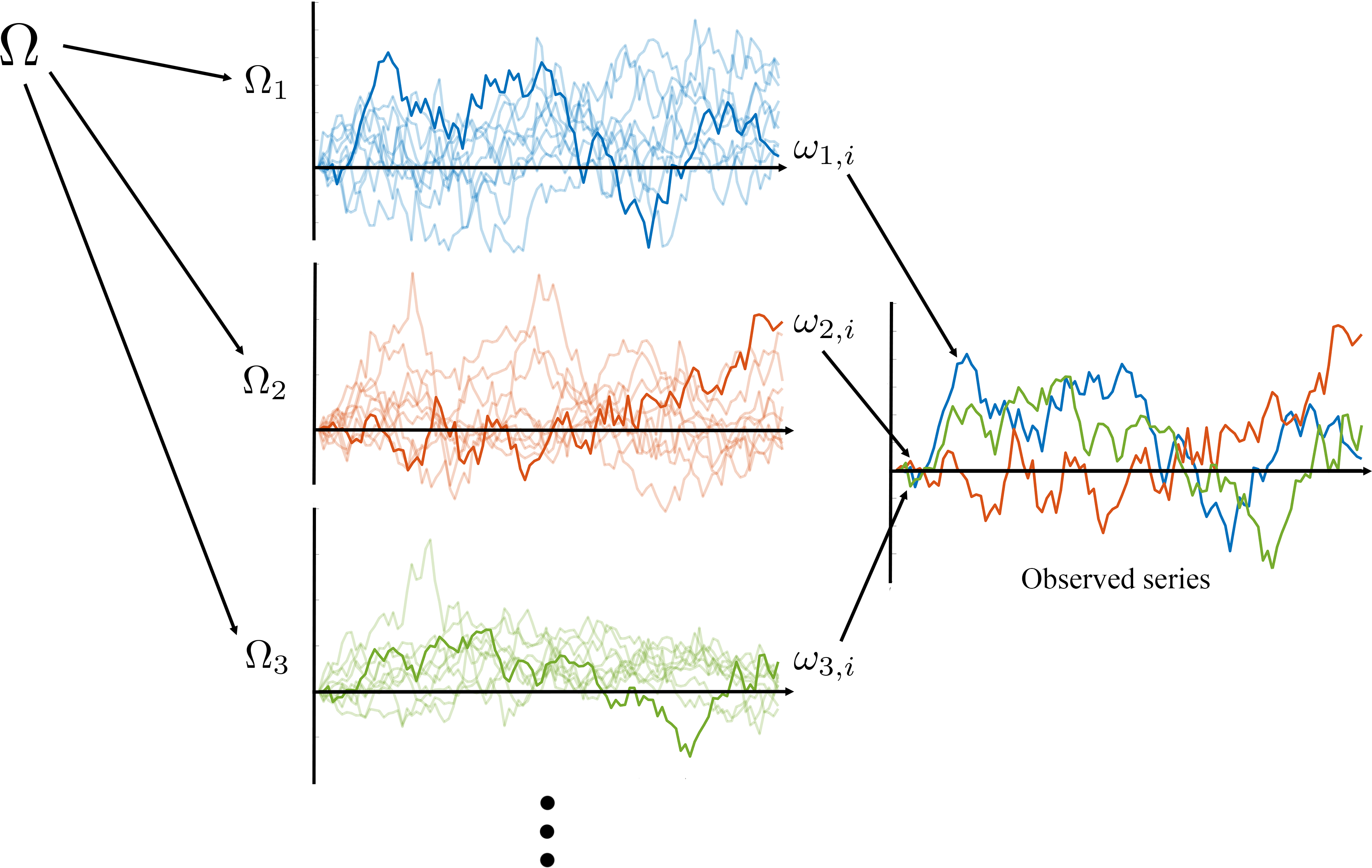}
\caption{Illustration of the sampling mechanism in time series. Unlike the i.i.d. case, even after each unit is selected, each unit will have numerous paths it can take, from which one path is sampled or observed. The lighter color paths are unobserved potential paths, whereas the bold color paths are the observed paths. The collection of the latter is what is observed.
\label{fig:sample}
}
\end{figure}

\subsection{Time as covariate and confounder}

It is often considered that time series data are characterized by trends, autocorrelation, heteroskedastic error, and seasonality, but these characteristics are secondary to the defining property that time itself is a covariate, meaning that the dynamical changes are a function of time itself.
Parameters and variances evolve as functions of time or unobserved latent states, making static models inadequate. 
Any statistical model that has fixed parameters will be unlikely to capture evolving temporal dynamics, even if it incorporates, e.g., autocorrelation terms.

For time series causal inference, time not only indexes dependence among successive observations but also carries information that influences both treatment assignment and outcomes: time is a confounder.  
At time $t$, treatment is rarely assigned at random; its probability depends on the accumulated history of the process.  
Causal inference should, therefore, condition on the evolving information set rather than on static covariates.  
In stochastic process language, this information structure is represented by a \emph{filtration} $\{\mathscr{F}_t\}$, where each $\mathscr{F}_t$ contains all data observed up to time $t$.  
Conditioning on $\mathscr{F}_t$ defines the temporal analogue of unconfoundedness-- \emph{dynamic ignorability} (defined in Assumption~\ref{ass:dynamicignore})-- under which treatment at time $\treat$ is independent of future potential outcomes given the past.  
This filtration captures the progressive revelation of information that drives process-level uncertainty.
If a causal estimation model cannot capture this dual role of time, it can cause bias.
Thus, a coherent framework should treat time as part of the causal mechanism.

Our formulation differs fundamentally from the sequential potential outcome framework of
\citet{bojinov2019time, bojinov2021panel}, where temporal dependence is treated as network interference and causal effects are defined by contrasts of treatment paths rather than stochastic processes. Their estimands do not characterize how interventions alter the probabilistic law of the outcome process (persistence, dynamics, etc.). In contrast, our framework treats each potential outcome as a full stochastic trajectory, allowing interventions to modify the entire dynamic evolution. A detailed comparison is in Supplementary Material~\ref{sm:comp}.

\subsection{Dynamic uncertainty and implications}

Even under perfect randomization, outcomes remain uncertain because future realizations of a stochastic process are inherently unpredictable.  
This recognition clarifies why static regression adjustments or trend corrections can be unsuccessful; they treat realized paths as deterministic counterfactuals, ignoring intrinsic processes.  
Defining potential outcomes as stochastic processes reframes causal inference as the study of how interventions alter transition laws rather than outcome levels.  
Identification depends on conditioning with respect to histories, and estimation necessitates respecting the temporal order of cause and effect under our estimand. 

\subsection{Illustrating example}

Consider a simple, concrete example to illuminate the above discussion. 
Suppose we run two different experiments. 
In experiment 1, we take 1,000 coins and toss them at the same time (presumably by 1,000 different people) and record the outcome of each coin. 
This is an experiment where we have 1,000 different realizations at the same time, $t=0$. 
The average here corresponds to the average at time $t=0$. 
On the other hand, consider picking one coin and tossing this particular coin 1,000 \emph{times}. 
The average of these outcomes is the time-varying average, at a fixed path (which corresponds to the fixed coin). 
This illustrates that averaging across individuals at one point in time differs fundamentally from averaging across time for one individual, underscoring why a dynamic estimand is required.

More generally, consider $N$ people $i = 1, \ldots, N$ each with their own coin, tossing it $T$ times $t = 1, \ldots, T$.  
$\bar{Y}_{i} = \frac{1}{T}\sum_{t=1}^T Y_{t,i}$ denotes the time-averaged proportion of heads of person $i$'s coin. $\bar{Y}_{t} = \frac{1}{N} \sum_{i=1}^N Y_{t,i}$ denotes the person-averaged proportion of heads at time $t$. Now consider an intervention where at time $\treat$, we give a drug to each person that is supposed to make them drowsy. 
The goal is to measure the effect of the drug on the outcome of the coin toss, over time. 
From time $\treat$, each unit now has two potential outcomes $Y_{t,i}(Z_i=1)$ and $Y_{t,i}(Z_i=0)$. 
The causal effect of $Z$ for person $i$ at time $t$ is defined as $\tau_{t,i} = Y_{t,i}(Z_i=1) - Y_{t,i}(Z_i=0)$. 
The time-averaged causal effect of person $i$ is $\bar{\tau}_{i} = \frac{1}{T} \sum_t\tau_{t,i}$. 
The person-averaged causal effect at time $t$ is $\bar{\tau}_{t} = \frac{1}{N} \sum_i\tau_{t,i}$. The overall average causal effect is ${\tau}_{\cdot \cdot} = \frac{1}{TN} \sum_{t,i} \tau_{t,i}$. 

In general, $\bar{\tau}_{i}$ is not equal to $\bar{\tau}_{t}$. 
Moreover, the overall causal effect $\tau$ does not give a complete picture of the time-varying causal effect of the drug on the ability of a person $i$ to toss a coin.

These distinctions motivate the need for a stochastic process definition of potential outcomes, developed in the next section, where each unit's outcome path replaces a single realization in the i.i.d.\ framework.

\section{Potential outcomes and causal estimand in time series}
\label{sec:date}

This section defines causal effects for time-evolving outcomes by extending the potential outcome framework to stochastic processes.
Each unit’s outcome is now a random path indexed by time and treatment path.
We begin with clear definitions of potential outcomes, ignorability, and the dynamic average treatment effect (DATE), followed by an intuitive explanation after each formal statement.

To extend the potential outcomes framework to time series, we view each outcome process as a stochastic trajectory that evolves through time. For every unit, there exist two potential outcome paths: one under treatment, $\left\{ Y_{t}(Z=1)\right\} _{t\in\left[0,T\right]}$ and one under control, $\left\{ Y_{t}(Z=0)\right\} _{t\in\left[0,T\right]}$. 
Randomness, here, arises both across sampled units and within each unit’s stochastic trajectory through time.
Before treatment, these paths coincide, though  they may diverge after treatment, as the intervention may alter the process’s evolution. Our goal is to characterize and estimate the dynamics of this divergence.
This section formalizes that idea by defining time series potential outcomes, the corresponding causal estimand, and the assumptions necessary for identification.

We first state the stable unit treatment value assumption (SUTVA) and consistency, which we maintain throughout.
\begin{assumption}[SUTVA and Consistency]\label{ass:sutva}
(i) \emph{No interference}: For each unit $i$, the potential outcome process depends only on its own treatment assignment: $Y_{t,i}(Z_1, \ldots, Z_N) = Y_{t,i}(Z_i)$ for all $t \in [0,T]$. 
(ii) \emph{No hidden versions of treatment}: There is a single version of each treatment level.
(iii) \emph{Consistency}: The observed outcome equals the potential outcome under the received treatment: $Y_{t,i} = Y_{t,i}(Z_i)$ for all $t \in [0,T]$.
\end{assumption}
Throughout, we maintain Assumption~\ref{ass:sutva}, positivity, and the no-anticipation assumption (i.e., the intervention cannot be anticipated).
Technical details about measurability, cylinder measures, and filtration are provided in the Supplementary Material~\ref{sm:comcepts}.

Let the outcome be a measurable stochastic process, $\left\{ Y_{t}\right\} _{t\in\left[0,T\right]}$.
For the stochastic process, $\left\{ Y_{t}\right\} _{t\in\left[0,T\right]}$, we call the function, $t\mapsto Y_{t}\left(\omega\right)$, when $\omega$ is fixed a {\it sample path} of $\omega$.
The sample path of the observed outcome for each subject is given as $t\mapsto Y_{t}\left(\omega\right)$.

We can define the Rubin causality model as follows.
\begin{definition}[Dynamic Potential Outcomes]
{\it Denote the intervention time by $t_c$ with $0<t_c<T$. The potential outcome processes
$\{Y_t(Z_i=1,\omega)\}_{t\in[0,T]}$ and $\{Y_t(Z_i=0,\omega)\}_{t\in[0,T]}$ satisfy
\[
Y_t(Z_i=1,\omega)=Y_t(Z_i=0,\omega),
\quad \forall\,\omega\in\Omega_i,\ \forall\, t<t_c .
\]}
\end{definition}

The treatment effect of a stochastic process is defined by different statistics regarding the cylinder measure of the potential outcome processes, $\left\{ Y_{t}(Z=1)\right\} ,\left\{ Y_{t}(Z=0)\right\} $; i.e., causal quantity of interest is the expected divergence between the treated and untreated outcome processes at each time $t$ after intervention.
Let $y_0$ denote a unit's initial value. Under observational assignment, units with different initial values may have different probabilities of treatment, so we write $y_0^1$  and $y_0^0$  to denote initial values for treated and control units, respectively. Under randomization, these share a common marginal $p(y_0)$.
The dynamic average treatment effect (DATE) process is a function of $t$ that is defined as a temporal function of the difference in expectations of paths for treated and control.
\begin{definition}[Dynamic Average Treatment Effect]\label{def:date}
{\it The DATE at time $t$ is defined as}
\begin{align}\label{eq:date}
\textrm{DATE}\left(t\right)\triangleq\int\left(\mathbb{E}_{y_{0}^{1}}\left[Y_{t}(Z=1)\right]-\mathbb{E}_{y_{0}^{0}}\left[Y_{t}(Z=0)\right]\right)p\left(y_{0}\right)dy_{0},
 \end{align}
\end{definition}

Here, $\mathbb{E}_{y_{0}^{1}}\left[Y_{t}(Z=1)\right]$ denotes the expected trajectory under treatment $z$ for a unit starting at $y_0$. Under randomization (and sampling over the population), initial values are balanced across treatment groups, so we integrate over the common marginal distribution $p(y_0)$.
Thus, the DATE compares the treated and control processes starting from initial states $y_0$.
The integral averages this contrast over the distribution of possible starting states, so the estimand reflects the population-level causal effect of the intervention on future trajectories.

If all units share the same baseline value, the DATE reduces to the simple difference 
$\mathbb{E}\left[Y_{t}(Z=1)\right]-\mathbb{E}\left[Y_{t}(Z=0)\right]$.
When outcomes are i.i.d. at a single time point, the DATE collapses to the classical average treatment
effect, $\mathbb{E}\left[Y(Z=1)\right]-\mathbb{E}\left[Y(Z=0)\right]$.
Thus, our framework nests the classical Rubin causal model as a special case.

The expectation $\mathbb{E}_{y_{0}^{Z}}\left[\cdot\right]$ averages over the stochastic evolution of paths conditional on initial value $y_{0}$ and treatment status $Z$.
Here, $y_{0}^{1}$ and $y_{0}^{0}$ denote the initial values for treated and control units, respectively, which may differ in observational settings where treatment assignment depends on baseline characteristics; under randomization, these distributions coincide.
The DATE conditions on the initial state $y_0$ and then averages over its marginal distribution, targeting a population-level causal effect. This formulation does not require a Markov assumption: the post-treatment evolution $Y_t(z)$ may depend on the entire pre-treatment path, with those dependencies absorbed into the identification conditions in Section~\ref{sec:dipw}.

Our proposed DATE is illustrated in Figure~\ref{fig:manymany}.
Each light path (for control and treated) is the path a stochastic process can take.
The expectations of the paths are the bolder colored paths.
The DATE is defined as the difference, per $t$, over the entire post-treatment period.

\begin{figure}[t!]
\centering
\includegraphics[width=0.75\textwidth,clip]{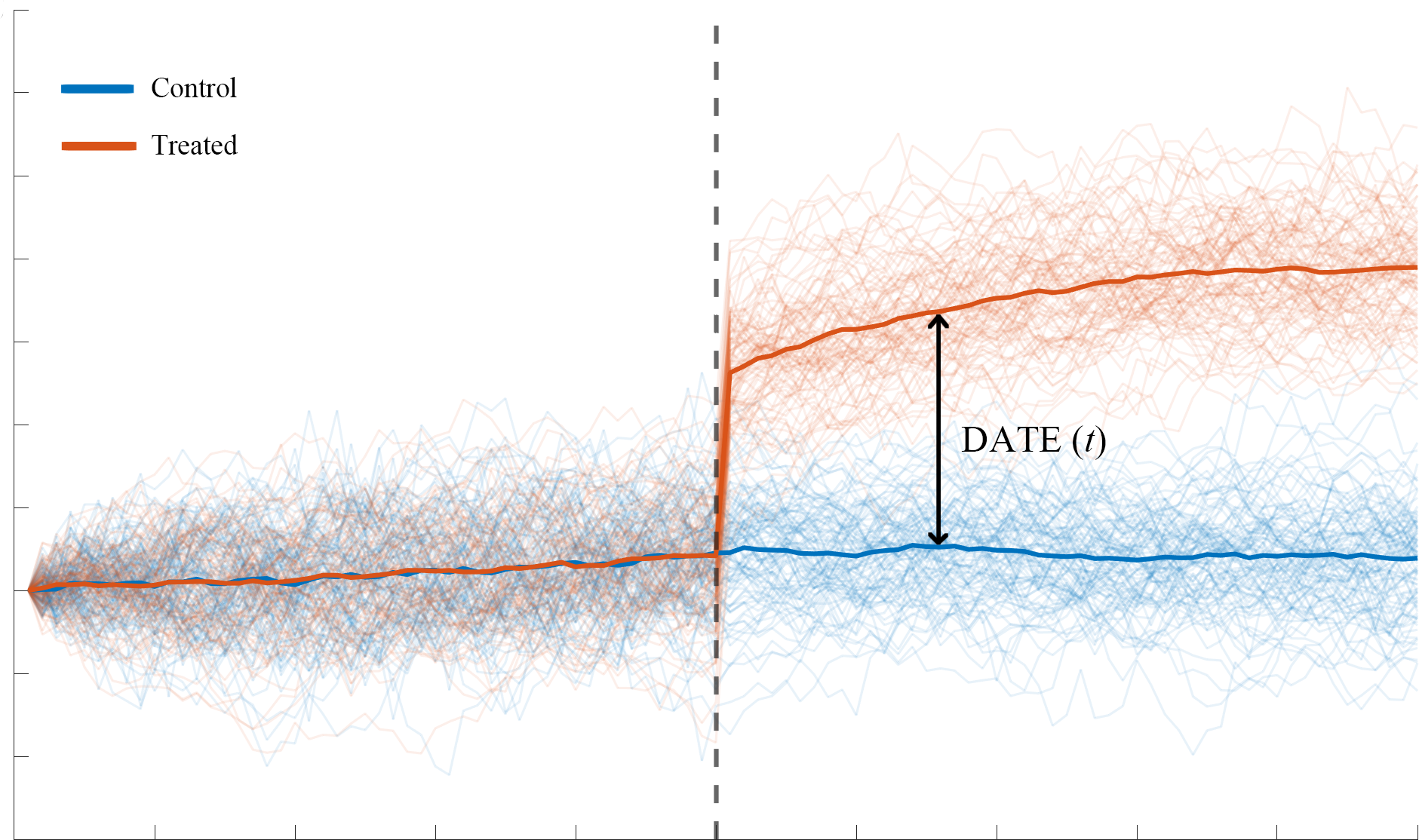}
\caption{Illustration of the proposed DATE. The light paths are the paths the stochastic process can take, and the bold color paths are their expectations. The DATE is defined as the difference between these expectations, per $t$, over the entire post-treatment period.
\label{fig:manymany}
}
\end{figure}

Let $i$ index each subject, and correspond the sample space, $\Omega$, to its subspace, $\Omega_{i}$.
The potential outcome for $i$ is the realization of the sample path, $t\mapsto\left[Y_{t}\left(Z=1,\omega_{i}\right),Y_{t}\left(Z=0,\omega_{i}\right)\right]$, of the measurable stochastic processes, $\left\{ Y_{t}(Z=1),Y_{t}(Z=0)\right\} _{t\in\left[0,T\right]}$, of $\omega_{i}$ drawn from $\Omega_{i}$.
If we let $Z$ be the binary random variable indicating treatment assignment, for each $i$, we have $Z\left(\omega_{i}\right)$.
If $Z\left(\omega_{i}\right)=1$, then $i$ is treated, and if $Z\left(\omega_{i}\right)=0$, then $i$ is not treated (control).
The observed process is 
\[
\left\{ Y_{t}\left(\omega_{i}\right)\right\} =\left\{ Z\left(\omega_{i}\right)Y_{t}\left(Z=1,\omega_{i}\right)+\left(1-Z\left(\omega_{i}\right)\right)Y_{t}\left(Z=0,\omega_{i}\right)\right\} ,
\]
therefore,
\[
Y_{t}\left(\omega_{i}\right)=\begin{cases}
Y_{t}\left(Z=1,\omega_{i}\right)=Y_{t}\left(Z=0,\omega_{i}\right) & \left(t<\treat\right),\\
Z\left(\omega_{i}\right)Y_{t}\left(Z=1,\omega_{i}\right)+\left(1-Z\left(\omega_{i}\right)\right)Y_{t}\left(Z=0,\omega_{i}\right) & \left(t\geqq \treat\right).
\end{cases}
\]

To interpret DATE causally, we require that the treatment assignment process $Z$ be independent of the potential outcome processes:
\begin{equation}
\left\{Y_{t}(Z=1),Y_{t}(Z=0)\right\}_{t\in\left[0,T\right]}\Perp Z.\label{eq:RCT}
 \end{equation}
The assignment, $\{0,1\}$, is determined by which element is drawn from $\Omega_{i}$, and there are as many sample paths regarding the outcome as the number of elements in $\Omega_{i}$.

This ``dynamic randomization'' assumption is the time series analogue of the classical unconfoundedness condition. 
This condition is closely related to the sequential ignorability assumption of, e.g., \cite{bojinov2019time}, but differs in interpretation.
In those frameworks, potential outcomes are deterministic functions of treatment histories,
and all randomness arises from treatment assignment.
Here, $\{Y_t(z)\}_{t\in[0,T]}$ are themselves stochastic processes,
so independence in \eqno{RCT} refers to the joint law of the potential outcome paths and the treatment variable.
This distinction allows us to model process uncertainty and to define causal estimands over entire trajectories rather than fixed sequences.
Further, this condition ensures that, at each time $t$, differences between treated and untreated trajectories can be attributed to the intervention rather than to pre-existing differences in their stochastic evolution.

Our proposed potential outcomes and estimand generalizes the Rubin causality framework by including dynamics of the pre- and post-treatment potential outcomes.
As a generalization, the i.i.d. potential outcomes setting can be interpreted as selecting a single moment in time of the stochastic process.
Thus, this stochastic process foundation allows weighting and modeling procedures (Sections~\ref{sec:dipw}-\ref{sec:dlm}) to operate on evolving paths rather than single-period outcomes.

\section{Covariates and the dynamic IPW estimator}\label{sec:dipw}

In observational time series, treatment assignment is rarely random. The probability of intervention at some time may depend on the unit’s history, trends, or covariates. Such time-dependent confounding violates the independence assumption of the previous section. To adjust for this, we generalize inverse probability weighting \citep{rosenbaum1983central} to dynamic settings, constructing weights that re-balance treatment and control trajectories conditional on the observed confounder processes.
This section formalizes that idea and derives a dynamic weighting estimator that is unbiased and consistent under mild conditions.

Definition~\ref{def:date} defines the DATE as a baseline-standardized mean contrast obtained by averaging
$\mathbb{E}[Y_t(1)| Y_0=y_0]-\mathbb{E}[Y_t(0)| Y_0=y_0]$ over a reference baseline distribution $p(y_0)$.
Throughout Section~4 we target this baseline-standardized estimand.
In the common-baseline case (e.g., by design or after standardization), this reduces to the simpler notation
$\tau_t := \mathbb{E}[Y_t(1)]-\mathbb{E}[Y_t(0)]$.

Denote the confounders as a $J$-dimensional stochastic process, $\left\{ \left(X_{1,t},\cdots,X_{J,t}\right)\right\} _{t\in\left[0,T\right]}$, or in ($J$-dimensional) vector notation, $\boldsymbol{X}_{t}=\left(X_{1,t},\cdots,X_{J,t}\right)^{\top}$.
Assume $\boldsymbol{X}_{t}$ is $\left\{ \mathscr{F}_{t}\right\} $-adapted.
Let the initial value, $\boldsymbol{X}_{0}$, be a random variable determined at $t=0$.
If the confounding process, $\left\{ \boldsymbol{X}_{t}\right\} _{t\in\left[0,T\right]}$, is determined up to even $\treat\leqq T$, the subset of the sample space, $\Omega$, can be determined; $\left\{ \omega\left|\left\{ \boldsymbol{X}_{t}\right\} _{t\leqq \treat}=\left\{ \boldsymbol{x}_{t}\right\} _{t\leqq \treat},\omega\in\Omega\right.\right\} $.
Denote this as $\Omega\left(\boldsymbol{x}_{t};t\leqq \treat\right)$.
This can be interpreted as a set of units, $\omega_{i}$, with $\left\{ \boldsymbol{x}_{t}\right\} _{t\leqq \treat}$ as its confounders.
If we let $\treat\uparrow T$, the set of units will be smaller: $\Omega\left(\boldsymbol{x}_{t};t\leqq \treat\right)\subseteqq\Omega\left(\boldsymbol{x}_{t};t\leqq u\right)$, if $\treat<u$.
These sets formalize the notion of conditioning on histories: all units sharing the same observed confounder path up to $\treat$ belong to the same conditioning set. 
This structure allows the weighting function to be defined over histories rather than discrete covariate values.

Assume the treatment assignment is determined at $\treat$.
Then, the assignment covariate, $Z$, is a $\mathscr{F}_{\treat}$-measurable random variable.
Here, we assume that the treatment assignment depends on the observed pre-treatment information set (e.g., $\mathscr{F}_{\treat}^{X}$, and potentially other observed pre-treatment variables); i.e., strongly ignorable treatment assignment.
For stochastic processes, we propose the following definition for dynamic strong ignorability:
If we condition the outcome with the filtration, $\mathscr{F}_{\treat}^{X}=\sigma\left\{ X_{\treat};t\leqq\treat\right\} $, at $\treat$, the potential outcome processes, $\left\{ Y_{t}(Z=1),Y_{t}(Z=0)\right\} _{t\in\left[\treat,T\right]}$, and treatment assignment, $Z$, after $\treat$, are independent:
\begin{assumption}[Dynamic Ignorability]\label{ass:dynamicignore}
    \begin{equation}
\left\{Y_{t}(Z=1),Y_{t}(Z=0)\right\}_{t\in\left[\treat,T\right]}\Perp \left.Z\right|{\mathscr{F}_{\treat}^{X}}.
 \end{equation}
\end{assumption}
The filtration notation $\mathscr{F}_{\treat}^{X}$ emphasizes that the assignment depends on the information set generated by the confounder process up to time $\treat$, encompassing both stochastic and deterministic components. This generalizes the standard conditional ignorability to continuous-time settings.
In other words, the conditional expectation of $Z$, given $\mathscr{F}_{\treat}^{X}$, is the propensity score:
\begin{equation}
\mathbb{E}_{P}\left[Z\left|\mathscr{F}_{\treat}^{X}\right.\right]=p\left(\left.Z=1\right|\mathscr{F}_{\treat}^{X}\right).\label{eq:Propensity}
 \end{equation}
Note that $\mathbb{E}\left[\cdot\right]$ the expectation operator regarding the probability measure, $P$ (of the probability space, $\left(\Omega,\mathscr{F},P\right)$).
Further discussion of this assumption is given in Supplementary Material~\ref{sm:assumption}.

\begin{assumption}[Regularity conditions for DIPW]\label{ass:regularity}
(i) \emph{Positivity}: $0 < p(Z_i = 1 | \mathscr{F}_{t_c}^X) < 1$ almost surely.
(ii) \emph{Propensity score}: Either $p(Z_i = 1 | \mathscr{F}_{t_c}^X)$ is known, or a consistent estimator $\hat{p}(Z_i = 1 | \mathscr{F}_{t_c}^X)$ is available.
(iii) \emph{Moments}: $\sup_{t \in [0,T]} \mathbb{E}[Y_t^2] < \infty$ and $\mathbb{E}[1/p(Z=1|\mathscr{F}_{t_c}^X)^2] < \infty$.
\end{assumption}

The dynamic inverse probability weighting (DIPW) estimator reweights observed trajectories by the inverse of their
treatment assignment probabilities conditional on the observed pre-treatment history.
Let $\mathscr F^{X}_{t_c}$ denote the $\sigma$-field generated by all observed pre-treatment covariate information used for assignment,
which may include lagged outcomes and any measurable summaries of past observed data.
Define the propensity score
\begin{equation}\label{eq:Propensity}
p_i \equiv p\!\left(Z_i=1 | \mathscr F^{X}_{t_c}\right),
\end{equation}
and assume positivity $0< p_i <1$ a.s.

For each time $t$, define the (unnormalized) DIPW estimators of the mean potential outcomes
\begin{align}
\widehat{\mu}^{\,1}_{t,\mathrm{DIPW}}
    &= \frac{1}{n}\sum_{i=1}^n \frac{Z_i}{p_i}\,Y_{t,i}, \label{eq:DIPW}\\
\widehat{\mu}^{\,0}_{t,\mathrm{DIPW}}
    &= \frac{1}{n}\sum_{i=1}^n \frac{1-Z_i}{1-p_i}\,Y_{t,i}. \nonumber
\end{align}
Under Assumptions~\ref{ass:sutva}--\ref{ass:regularity}, we identify the dynamic average treatment effect (DATE) at time $t$ by
\[
\tau_t \equiv \mathbb E\!\left[Y_t(Z=1)\right]-\mathbb E\!\left[Y_t(Z=0)\right],
\qquad
\widehat{\tau}_{t,\mathrm{DIPW}} \equiv \widehat{\mu}^{\,1}_{t,\mathrm{DIPW}}-\widehat{\mu}^{\,0}_{t,\mathrm{DIPW}}.
\]

\begin{thm}\label{thm:dipw}
{\it Under Assumptions~\ref{ass:sutva}--\ref{ass:regularity} and positivity, the DIPW estimators in \eqref{eq:DIPW} are unbiased:
\[
\mathbb E\!\left[\widehat{\mu}^{\,1}_{t,\mathrm{DIPW}}\right] = \mathbb E\!\left[Y_t(Z=1)\right],
\qquad
\mathbb E\!\left[\widehat{\mu}^{\,0}_{t,\mathrm{DIPW}}\right] = \mathbb E\!\left[Y_t(Z=0)\right],
\]
and hence $\mathbb E[\widehat{\tau}_{t,\mathrm{DIPW}}]=\tau_t$.
Moreover, if $\mathbb E\!\left[\left|\frac{Z}{p}Y_t\right|\right]<\infty$ and $\mathbb E\!\left[\left|\frac{1-Z}{1-p}Y_t\right|\right]<\infty$,
then for each fixed $t$, $\widehat{\mu}^{\,z}_{t,\mathrm{DIPW}} \to \mathbb E[Y_t(z)]$ almost surely as $n\to\infty$ for $z\in\{0,1\}$.
If time is discrete with $t\in\{0,1,\dots,T\}$ and $T<\infty$, then the convergence is uniform:
\[
\max_{0\le t\le T}\left|\widehat{\mu}^{\,z}_{t,\mathrm{DIPW}}-\mathbb E\!\left[Y_t(z)\right]\right|\to 0 \quad\text{a.s.}, \qquad z\in\{0,1\}.
\]
Consequently, $\max_{0\le t\le T}|\widehat{\tau}_{t,\mathrm{DIPW}}-\tau_t|\to 0$ almost surely.}
\end{thm}

The proof (Supplementary Material~\ref{sm:dipw}) adapts standard IPW arguments:
unbiasedness follows by iterated expectation under ignorability given $\mathscr F^{X}_{t_c}$,
and consistency follows from the strong law of large numbers.

While \eqno{DIPW} is unbiased, its variance can be large when propensities are near 0 or 1.
A common variance-reduction modification is the stabilized (self-normalized) version
\[
\widehat{\mu}^{\,1}_{t,\mathrm{stab}}
=\frac{\sum_{i=1}^n \frac{Z_i}{p_i}Y_{t,i}}{\sum_{i=1}^n \frac{Z_i}{p_i}},
\qquad
\widehat{\mu}^{\,0}_{t,\mathrm{stab}}
=\frac{\sum_{i=1}^n \frac{1-Z_i}{1-p_i}Y_{t,i}}{\sum_{i=1}^n \frac{1-Z_i}{1-p_i}},
\]
with $\widehat{\tau}_{t,\mathrm{stab}}=\widehat{\mu}^{\,1}_{t,\mathrm{stab}}-\widehat{\mu}^{\,0}_{t,\mathrm{stab}}$.
These stabilized estimators are generally not exactly unbiased in finite samples, but are consistent under the same conditions,
and often yield improved finite-sample behavior.

In practice, $p_i$ is estimated using summary features or basis expansions of the pre-treatment history.
For example, one may fit a logistic regression using lagged outcomes, cumulative sums, or principal components of $\{Y_{s,i}\}_{s<t_c}$
(and other pre-treatment covariates) as predictors; model selection can proceed via cross-validation or information criteria.

When treated units are scarce (or when no comparable untreated units are observed), DATE cannot be identified solely from
cross-sectional reweighting arguments. In this regime, causal estimation necessarily relies on a model for the
untreated counterfactual trajectory of the treated unit. Our state-space specification, therefore, plays the role of a
{structural counterfactual model}: it imputes the latent evolution
and DATE is identified conditional on the maintained state-space assumptions. We emphasize this is not a limitation of our approach
but a feature of the data regime. With one treated trajectory, some form of structural restriction is unavoidable.

\section{Estimating the dynamic average treatment effect}\label{sec:dlm}
Under random assignment, the DATE can be estimated nonparametrically by averaging treated and control paths at each time point.
However, when only one or few treated units are observed, the single observed path does not suffice-- we must estimate the mean process.
Figure~\ref{fig:oneone} illustrates that the DATE is the difference between estimated expectations, not the observed series themselves.

\begin{figure}[t!]
\centering
\includegraphics[width=0.7\textwidth,clip]{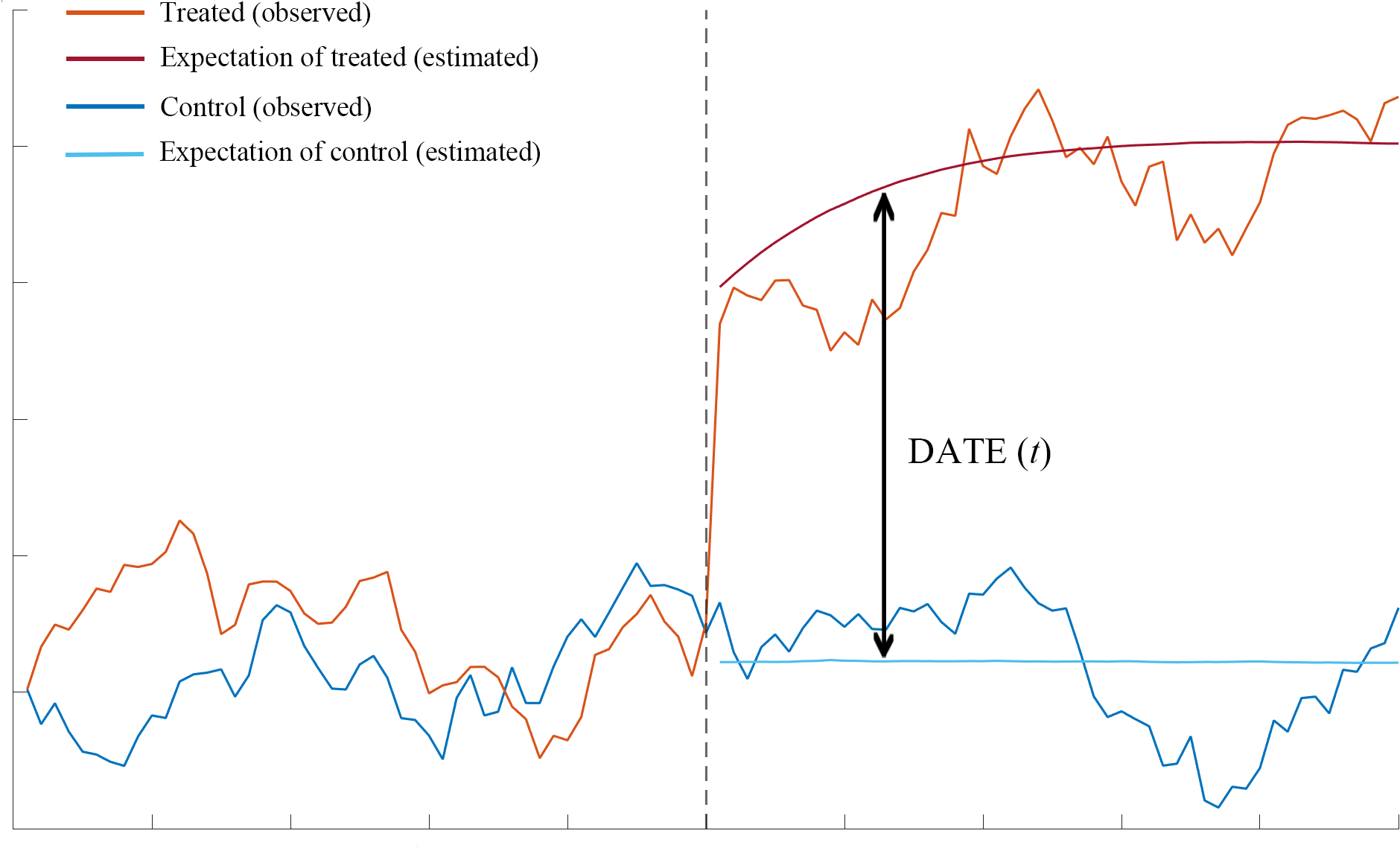}
\caption{Illustration of the DATE when there is only one treated and one control series. Although both series are observed, since there is only one path, the expectations of the paths are estimated. The difference between the two estimated expectations of paths is the DATE.
\label{fig:oneone}
}
\end{figure}

In the following, we show that, under some assumptions, the true DATE can be estimated using a state-space representation, of which the DLM is a special case.

\paragraph{Assumptions for representation.}
The estimation of the causal effect can be defined using covariates, similar to the linear regression in i.i.d. settings.
However, in contrast to i.i.d. regressions, the confounders are assumed to be stochastic processes, $\left\{ \boldsymbol{X}_{t}\right\} $, which makes conditioning on $X$ require special attention.
If strong ignorability holds, then conditional on $\mathscr{F}_{\treat}^{X}$, we have
\[
\mathbb{E}\left[\left.Y_{t}(Z=1)-Y_{t}(Z=0)\right|\mathscr{F}_{\treat}^{X}\right]=\mathbb{E}\left[\left.Y_{t}\right|Z=1,\mathscr{F}_{\treat}^{X}\right]-\mathbb{E}\left[\left.Y_{t}\right|Z=0,\mathscr{F}_{\treat}^{X}\right],
\]
thus, the causal effect can be expressed using a regression model using only the observed variables, i.e., the projection of the stochastic process, $\left.Y_{t}(Z=1)\right|_{\Omega_{i}^{1}},\left.Y_{t}(Z=0)\right|_{\Omega_{i}^{0}}$, to $\mathscr{F}_{\treat}^{X}$.

The strongly ignorable treatment
assignment assumption is not directly applicable because it has to be stated relative to filtration, and static conditioning can condition on post-treatment information.
We therefore introduce a new assumption we call intertemporal strongly ignorable treatment assignment, which is a time series extension of strong ignorability:
\begin{assumption}[Intertemporal Strong Ignorability]
\begin{equation}
\left.\left\{ Y_{t}(Z=1),Y_{t}(Z=0)\right\} _{t\in\left[\treat,T\right]}\Perp Z\right|\mathscr{F}_{t}^{X},\ \left(\mathscr{F}_{\treat}^{X}\subseteqq\mathscr{F}_{t}^{X}\right).\label{eq:TSC-USttong}
 \end{equation}
\end{assumption}

The sufficient (but not necessary) condition for \eqno{TSC-USttong} to hold is a condition we call ``no feedback to confounders;" 
\begin{equation}
Z\Perp\left\{ \mathscr{F}_{t}^{X}\right\} _{t>\treat}.\label{eq:TSC-nonFeedBack}
 \end{equation}
More specifically, for $Z\Perp\left\{ \mathscr{F}_{t}^{X}\right\} _{t>\treat}$, we have,
\[
P\left(Z=z,A\right)=P\left(Z=z\right)P\left(A\right),\textrm{for all }A\in\mathscr{F}_{t}^{X}.
\]
This condition states that the treatment assignment at $\treat$ does not affect the confounders, $\left\{ X_{t}\right\} _{t\geqq\treat}$, after $\treat$.
In other words, a mediator variable-- a variable that is affected by $Z$ and affects $Y$-- either does not exist as a confounder or is removed after $\treat$.
For example, if the interest is in estimating the effect of a mask mandate on infectious cases, weather can be a covariate that affects the outcome and timing of intervention, but is not affected by the mandate itself.
In contrast, mobility data will be affected by such a mandate, and thus receives feedback from the treatment, and thus does not satisfy this condition.
Such feedback can be handled through explicit structural modeling.
We note that \eqno{TSC-nonFeedBack} is sufficient but not necessary for intertemporal strong ignorability. Characterizing the full set of necessary conditions remains an open problem; in practice, researchers should assess whether feedback from treatment to confounders is plausible in their specific application and consider sensitivity analyses when this assumption is questionable.
Further discussion of this assumption is given in Supplementary Material~\ref{sm:assumption}.

The no feedback assumption rules out treatment-induced changes in confounders that would then affect outcomes.
If such feedback is plausible, post-treatment covariates should be treated as \emph{mediators} rather than confounders and should not be
conditioned on for identification. In practice, one can (i) define $\mathscr F^{X}_{t_c}$ using only pre-treatment covariates and history,
(ii) report sensitivity analyses in which candidate feedback pathways are perturbed (e.g., bounded shifts to post-treatment covariates or
propensities), and/or (iii) move to an explicit joint structural model for $(X_t,Y_t)$ in which treatment enters both equations.
Our empirical analyses adopt (i) by construction and use placebo checks as diagnostic evidence against major misspecification.

\subsection{Representation theorem for DATE}
When few or one treated unit is observed,  the mean process must be modeled, since there are not enough paths to approximate well by simply taking the mean.
The choice is not whether to make structural assumptions, but which assumptions to make. Common alternatives are linear models (assume flat effects) and ARIMAX (static parameters). 
We show that, under mild regularity, the conditional mean trajectories admit a state-space (DLM) form, which motivates DLM-based estimation when treated units are scarce.

\paragraph{State-space models as the representation of DATE.}
The fact that the outcomes, $\left\{ Y_{t}(Z=1),Y_{t}(Z=0)\right\} $, and confounders, $\left\{ \boldsymbol{X}_{t}\right\} $, follow a general stochastic process, is not sufficient to derive a formulation for the expected outcomes, $\mathbb{E}\left[\left.Y_{t}(Z=1)\right|\mathscr{F}_{t}^{X}\right]$, $\mathbb{E}\left[\left.Y_{t}(Z=0)\right|\mathscr{F}_{t}^{X}\right]$, and propensity score, $\mathbb{E}_{P}\left[Z\left|\mathscr{F}_{\tau}^{X}\right.\right]=p\left(\left.Z=1\right|\mathscr{F}_{\tau}^{X}\right)$.
A statistical model, as a function of $\left\{ X_{t}\right\} $, is thus required to estimate the above with data.

We propose a model that captures both the temporal dynamics and history-dependent causal responses. 
Our proposed model can be formulated from continuous semimartingales, with mild assumptions on the observed stochastic processes.
Specifically, these are i) $\{Y_t(z)\}$ are square-integrable semimartingales adapted to $\{\mathcal F_t\}$; and ii) covariate processes $\{X_t\}$ have finite variation and are $\mathcal F_t$-predictable.
Here `predictable' just means the coefficient at time $t$ is a function of information available up to time $t-1$; i.e., a time-varying regression coefficient updated from past history.
We show that this model is a state-space model, of which a dynamic linear model is a special case, and show that this yields a state-space characterization of the conditional mean components entering the DATE; the DLM is a convenient special case for estimation.

The introduction of continuous semimartingales to represent $\left\{ Y_{t}(Z=1),Y_{t}(Z=0)\right\} $, $\left\{ \boldsymbol{X}_{t}\right\} $ is done for several benefits.
The first is that the statistical model for the outcomes and confounders can be specified using semimartingales without strong assumptions.
The second is that using continuous semimartingales allows us to formulate the required representation, even if the correct model is inestimable.

The formal statement below provides one sufficient set of conditions under which the DATE can be expressed as a state-space model for the discrete case:
\begin{thm}\label{thm:approx}
{\it Let $\left\{ Y_{t}\left(=ZY_{t}(Z=1)+\left(1-Z\right)Y_{t}(Z=0)\right)\right\} _{t\in\left[0,T\right]}$ and $\left\{ \boldsymbol{X}_{t}\right\} _{t\in\left[0,T\right]}$ be square integrable discrete time stochastic processes.
Then, $\mathbb{E}\left[\left.Y_{t}\right|Z=1,\mathscr{F}_{t}^{X}\right]$ and $\mathbb{E}\left[\left.Y_{t}\right|Z=0,\mathscr{F}_{t}^{X}\right]$ can be expressed as
\begin{equation}
\begin{aligned}\mathbb{E}\left[\left.Y_{t}\right|Z=1,\mathscr{F}_{t}^{X}\right] & =\mathbb{E}\left[\left.Y_{t-1}\right|Z=1,\mathscr{F}_{t-1}^{X}\right]+\sum_{j=1}^{d}\beta_{j,t}^{1}\left(X_{j,t}-X_{j,t-1}\right),\ \left(\treat\leqq t\leqq T\right),\\
\textrm{with initial value } & \mathbb{E}\left[\left.Y_{\treat}\right|Z=1,\mathscr{F}_{\treat}^{X}\right]=\mathbb{E}\left[\left.Y_{\treat}\right|Z=0,\mathscr{F}_{\treat}^{X}\right]+\alpha_{\treat},\\
\mathbb{E}\left[\left.Y_{t}\right|Z=0,\mathscr{F}_{t}^{X}\right] & =\mathbb{E}\left[\left.Y_{t-1}\right|Z=0,\mathscr{F}_{t-1}^{X}\right]+\sum_{j=1}^{d}\beta_{j,t}^{0}\left(X_{j,t}-X_{j,t-1}\right).\ \left(0\leqq t\leqq T\right),\\
\textrm{with initial value } & \mathbb{E}\left[\left.Y_{0}\right|Z=0,\mathscr{F}_{0}^{X}\right]=Y_{0}.
\end{aligned}
\label{eq:TSC-VCM-1}
 \end{equation}
Here, $\left\{ \beta_{j,t}^{z}\right\} _{1\leqq j\leqq d}^{z=0,1}$ is a $\mathscr{F}_{t-1}$-measurable stochastic process.}
{\it Assume $\left\{ \boldsymbol{\beta}_{t}^{z}\right\} $ is $\mathscr{F}_{t}$-predictable, square-integrable, and continuous path.
The time series model for discrete observations of the continuous time stochastic process, $\left\{ \boldsymbol{\beta}_{t}^{z}\right\} $, is given as
\begin{equation}
\boldsymbol{\beta}_{t}^{z}=G_{t}\boldsymbol{\beta}_{t-1}^{z}+\boldsymbol{\omega}_{t},\ \boldsymbol{\omega}_{t}\sim N\left(0,\sigma_{t}^{2}\boldsymbol{W}_{t}\right),\ \ \ \left(t=1,2,\cdots,T\right),\left(z=0,1\right).\label{eq:TimeSeriesCausal_DLMtransition}
 \end{equation}}
\end{thm}
The proof (as well as its continuous time counterpart) is in the Supplementary Material~\ref{sm:approx}.
The above theorem shows that any estimator targeting the conditional mean process via the state/observation \eqsno{TSC-VCM-1}-\ref{eq:TimeSeriesCausal_DLMtransition} are aligned with the implied conditional-mean dynamics.
These two formulations together form a DLM, thus showing that the DLM form arises as the projection of the outcome process onto the space spanned by the covariate process. 
In practice, this result justifies using DLMs for causal modeling when data are limited.

Theorem~\ref{thm:approx} is an existence result: under mild regularity, the conditional means
$\mathbb{E}[Y_t| Z=z,\mathscr{F}_t^X]$ admit a state-space representation with predictable coefficients on covariate increments.
To operationalize this in scarce-treated settings, we adopt a linear-Gaussian DLM with a structured intervention design
matrix. This design matrix is a {modeling choice} that provides an interpretable basis (e.g., spot/persistent/trend)
for the post-intervention mean response; alternative bases can be substituted within the same state-space framework.
Accordingly, component-level conclusions are model-based, conditional on the maintained DLM specification.

Thus, we have derived the following DLM representation with specific modeling choice:
\begin{subequations}
\label{eq:DLM} 
\begin{align}
y_{t}^{z} & =y_{t-1}^{z}+\boldsymbol{F}_{t}^{\top}\boldsymbol{\theta}_{t}+\varepsilon_{t},\quad\varepsilon_{t}\sim N(0,\sigma_{t}^{2}),\label{eq:DLMa}\\
\boldsymbol{\theta}_{t} & =\boldsymbol{G}_{t}\boldsymbol{\theta}_{t-1}+\boldsymbol{\omega}_{t},\quad\boldsymbol{\omega}_{t}\sim N\left(\left[\begin{array}{c}
0\\
\boldsymbol{0}
\end{array}\right],\left[\begin{array}{cc}
\sigma_{t}^{2} & \boldsymbol{0}^{\top}\\
\boldsymbol{0} & \sigma_{t}^{2}\W_{t}
\end{array}\right]\right),\label{eq:DLMb}\\
\boldsymbol{F}_{t} & =\left[\begin{array}{c}
-1\\
\boldsymbol{x}_{t}
\end{array}\right],\ \boldsymbol{\theta}_{t}=\left[\begin{array}{c}
\beta_{0,t}^{z}\\
\boldsymbol{\beta}_{t}^{z}
\end{array}\right]\label{eq:DLMc}\\
\X & =[\x_{1},...,\x_{\treat},...,\x_{T}]\label{eq:DLMd}\\
 & =\begin{bmatrix}y_{0} & ... & y_{\treat-1} & y_{\treat} & y_{\treat+1} & ... & y_{T-1}\\
0 & ... & 0 & \mathbbm{1}_{Z_{i}=1} & 0 & ... & 0\\
0 & ... & 0 & \mathbbm{1}_{Z_{i}=1} & \mathbbm{1}_{Z_{i}=1} & ... & \mathbbm{1}_{Z_{i}=1}\\
0 & ... & 0 & \mathbbm{1}_{Z_{i}=1} & 2\mathbbm{1}_{Z_{i}=1} & ... & (T-{\treat}+1)\mathbbm{1}_{Z_{i}=1}
\end{bmatrix},\label{eq:DLMe}
 \end{align}
\end{subequations}
where the functional form of \eqsno{DLMa} and (\ref{eq:DLMb}), corresponds to the above theoretical result.

Building on the representation result in Theorem~\ref{thm:approx}, the following regularity conditions ensure the DLM estimator is unbiased and consistent.
First, the underlying mean causal process must evolve smoothly over time 
so that a first-order representation captures its local dynamics. 
Second, the observation and state equations must be correctly specified up to 
additive, mean-zero Gaussian innovations $(\nu_t, \omega_t)$. 
Third, the state variance $W_t$ should remain bounded, ensuring that 
the filtering and smoothing distributions concentrate around the true latent paths as $T\to\infty$.
Under these conditions, the DLM posterior mean of $Y_t(1)-Y_t(0)$ 
converges to the true DATE trajectory in mean-square loss.

The other \eqsno{DLMc}-(\ref{eq:DLMe}) will depend on the specifications and assumptions of the data and effects of treatment.
The design matrix in \eqno{DLMe} represents one specification that decomposes the treatment effect into spot (immediate impact), persistent (level shift), and trend (slope change) components. Alternative specifications-- such as seasonal indicators, additional lags, or nonlinear basis expansions-- are equally compatible with Theorem~\ref{thm:approx}.

For the treatment indicator vectors, we include three different types of vectors to decompose and capture the treatment effect.
The first, where it is one at, and only at, the point of treatment measures the ``spot'' effect-- the immediate impact as a reaction to the intervention, which may or may not have long-term effects.
The second is the persistent effect, representing a long-term level shift (sometimes called a regime change).
The third is the change in trends, or slope, in the outcomes.
For example, consider a drug with an immediate positive impact that dissipates over time but has a long-term negative effect.
The immediate impact is measured by the spot effect, while the dissipation and ultimate negative impact are captured by the trend and persistent effects, respectively.

Mixed cases (e.g., one treated unit and many controls) can combine strategies: the treated mean path is estimated via DLM while the control mean path uses the sample average, or vice versa.
Details are in Supplementary Material~\ref{sm:estimation}.

When only one series receives treatment, the counterfactual control is itself, pre-treatment.
The theory shows that the counterfactual mean path is the DLM expectation with treatment indicators set to zero; the DATE is the difference between treated and counterfactual paths.
Additional discussion of retrospective and forward-looking causality is in Supplementary Material~\ref{sm: date}.

Operationally, estimation proceeds as follows: (i) when many comparable units are observed, estimate propensities from pre-treatment
history and compute the (unnormalized) DIPW mean paths for treated and control units, then difference to obtain $\widehat{\mathrm{DATE}}(t)$;
(ii) when treated units are scarce, fit the maintained state-space model on the full data and propagate a posterior predictive
factual and counterfactual trajectory for $t\ge t_c$, by switching on/off the treatment indicators, then compute DATE as the factual-minus-counterfactual difference with uncertainty.

\section{Simulation Study}\label{sec:sim}

We now consider simulation studies to assess our proposed framework against existing strategies using four different scenarios with varying numbers of units.
The first, and canonical, scenario is the time series RCT case, where there are multiple units for treated and control (Many-Many).
The other three scenarios are a variation of the RCT case where either the treated, or both have only one unit, and one case where there is only one treated and no control.
We denote this as One-Many, for one treated, many control (a scenario where SCMs have been used), One-One, for one treated, one control (a scenario where DiD has been used), and One-None, for only one treated and no control (an ITS scenario).
Finally, the One-None case is the ITS scenario, where only one treated series is observed.
The objective of this simulation study is to highlight the difference in estimands and analyze how the different estimation strategies produce varying performances with regard to estimating the true treatment effect.

\paragraph{Data generating process.}
For the DGP, we consider the following process:
\begin{align*}
    y_t&=
  \begin{cases}
    \theta y_{t-1}+b_1+\epsilon_t, \quad \epsilon_t\sim N(0,\sigma^2_{t}), & \text{if $t<\treat$,} \\
    \mathbbm{1}_{Z_i=1}(y_{t-1}+b_2+b_3)+(1-\mathbbm{1}_{Z_i=1})(\theta y_{t-1}+b_1)+\epsilon_t, \quad \epsilon_t\sim N(0,\sigma^2_{t}),                & \text{if $t=\treat$,} \\
    \theta y_{t-1}+\mathbbm{1}_{Z_i=1}b_3+(1-\mathbbm{1}_{Z_i=1})b_1+\epsilon_t, \quad \epsilon_t\sim N(0,\sigma^2_{t}),        & \text{if $t>\treat$.}
  \end{cases}\\
  \sigma^2_{t}&=\sigma^2_{t-1}(\beta/\eta_t), \quad \eta_t\sim Be(\frac{\beta (t+T/2)}{2},\frac{(1-\beta) (t+T/2)}{2}),
 \end{align*}
for $t=1:T$.
Here, $\theta$ is the AR(1) coefficient, $b_1$ is the initial intercept, $b_2$ is the spot effect, and $b_3$ is the persistent, long-term effect.
If the unit is in control, $y_t=\theta y_{t-1}+b_1+\epsilon_t$ is the DGP throughout.
The initial value of $y_0$ is set to 0.05, the stationary mean.
The time-varying variance follows a Beta-Gamma random walk with some discount factor, $\beta$, i.e., discount stochastic volatility.
Because of the AR(1) effect, the goal of the simulation is not to estimate each parameter, since they compound over time, but to estimate the time-varying difference between the mean treated and control series.

For the simulation study, we consider three lengths of time, $T=72,120,240$, to emulate six, ten, and twenty years of monthly data.
The parameter values for our DGP are set as
$b_1=0.01$, $b_2=0.5$, $b_3=-0.03$, $\beta=0.95$.
For $\theta$, we vary it from $0.75,0.8,0.9$, depending on the length of the series.
Since $\theta$ controls how the effect carries over, we set $\theta$ to be larger for longer series so the effect is observed over a longer period.
In the many cases, we generate 100 units.
For each scenario, we generate 1,000 paths, applying the treatment in the middle of each time frame $\treat = 37,61,121$.
The true DATE is plotted in Figure~\ref{fig:ATE}.
As shown, the treatment has an initial positive impact, which decreases over time, until it steadies into a persistent negative effect.

\begin{figure}[t!]
\centering
\includegraphics[width=0.7\textwidth,clip]{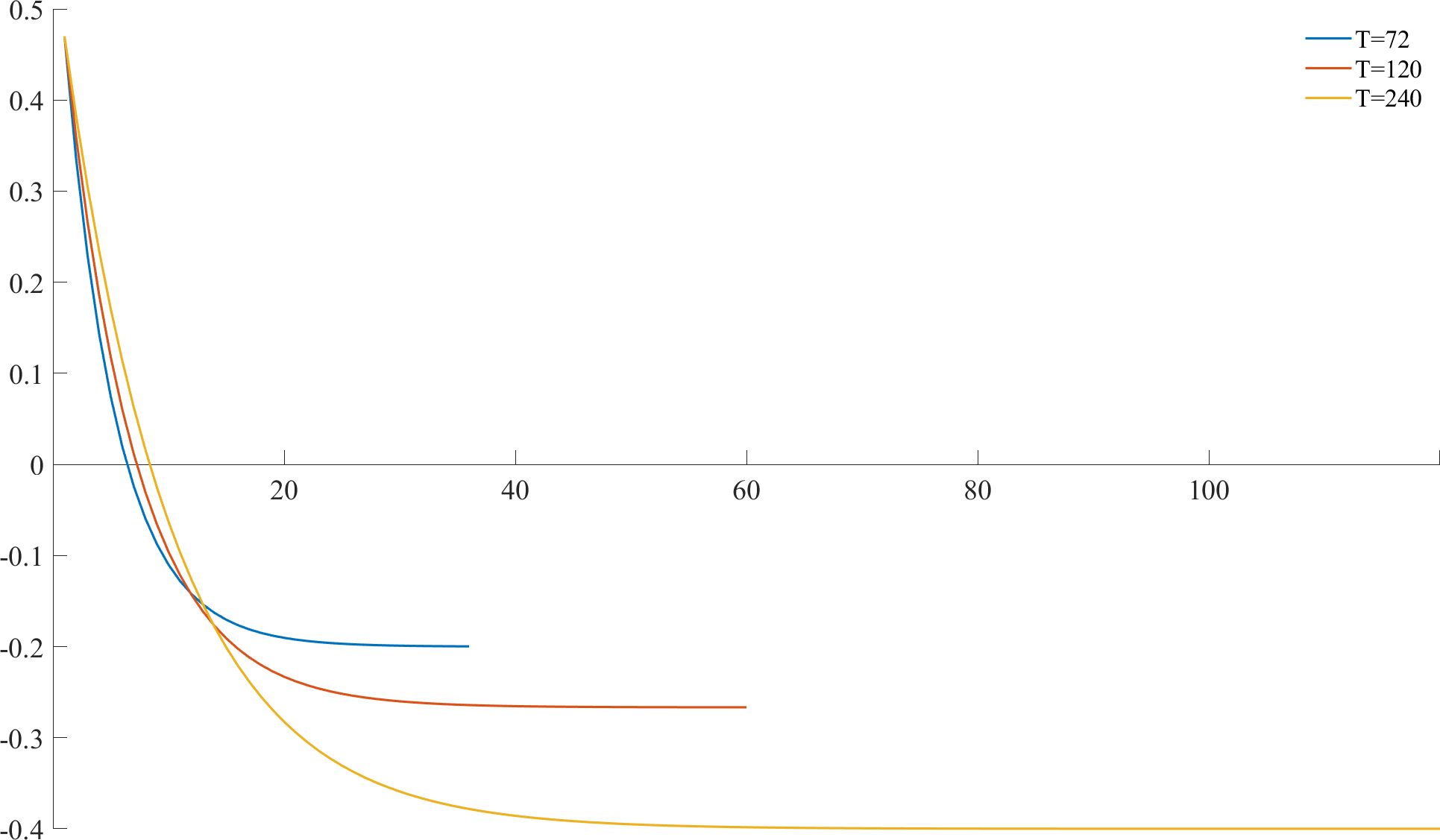}
\caption{True DATE for $T=72,120,240$.
\label{fig:ATE}
}
\end{figure}

\paragraph{Estimators.}
We compare DLM, linear model (LM), LM with AR(1) term, ARIMAX, difference in observed series (Y), and for specific scenarios, synthetic control (SCM), difference-in-differences (DiD), and causal ARIMA (C-ARIMA).
DLM priors and estimation details are in Supplementary Material~\ref{sm: model desc}.

\paragraph{Evaluation metrics.}
We compare mean squared error (MSE), standardized to the Many-Many $T=72$ case, and coverage probability (CP) at nominal 95\%.

\paragraph{Results.}

The results from the simulation study are collated in Table~\ref{table:msecp}. 
The results for MSE are standardized to the Many-Many case with $T=72$, which is our ideal RCT case, to assess how well each method approximates the ideal condition.

For the Many-Many case, the MSE increases slightly as $T$ increases, reflecting both the small baseline errors (order of $10^{-6}$) and the difficulty of estimating DATE with longer horizons when AR(1) decay is slower.
Across all other scenarios, the DLM achieves the smallest MSE-- often half that of the second-best model (LM)-- demonstrating its ability to capture temporal dynamics.
Methods that use the observed series directly (Y, SCM, DiD, C-ARIMA) perform poorly because variability in the observed series directly affects estimation accuracy.

The coverage probability results mirror the MSE findings: DLM achieves accurate 95\% intervals, while LM is overly conservative (CP near 100\%) and other methods undercover.
Figure~\ref{fig:CPi} shows quantile coverage across levels from 5\% to 95\%; the DLM remains well-calibrated throughout, while other methods show varying accuracy across quantiles, highlighting that uncertainty evaluation cannot rely on a single quantile.

Figure~\ref{fig:ON120} displays performance over time for the One-None case with $T=120$.
The DLM accurately tracks the true DATE (black line), while other models fail to capture the time-varying effect: LM overfits to the long tail, AR(1)-LM and ARIMAX estimate a flat average, and methods using observed series (Y, C-ARIMA) show high variability.
Both MSE and CP over time confirm that DLM maintains superior performance throughout the post-treatment period.


\begin{figure}[t!]
\centering
\includegraphics[width=0.7\textwidth,clip]{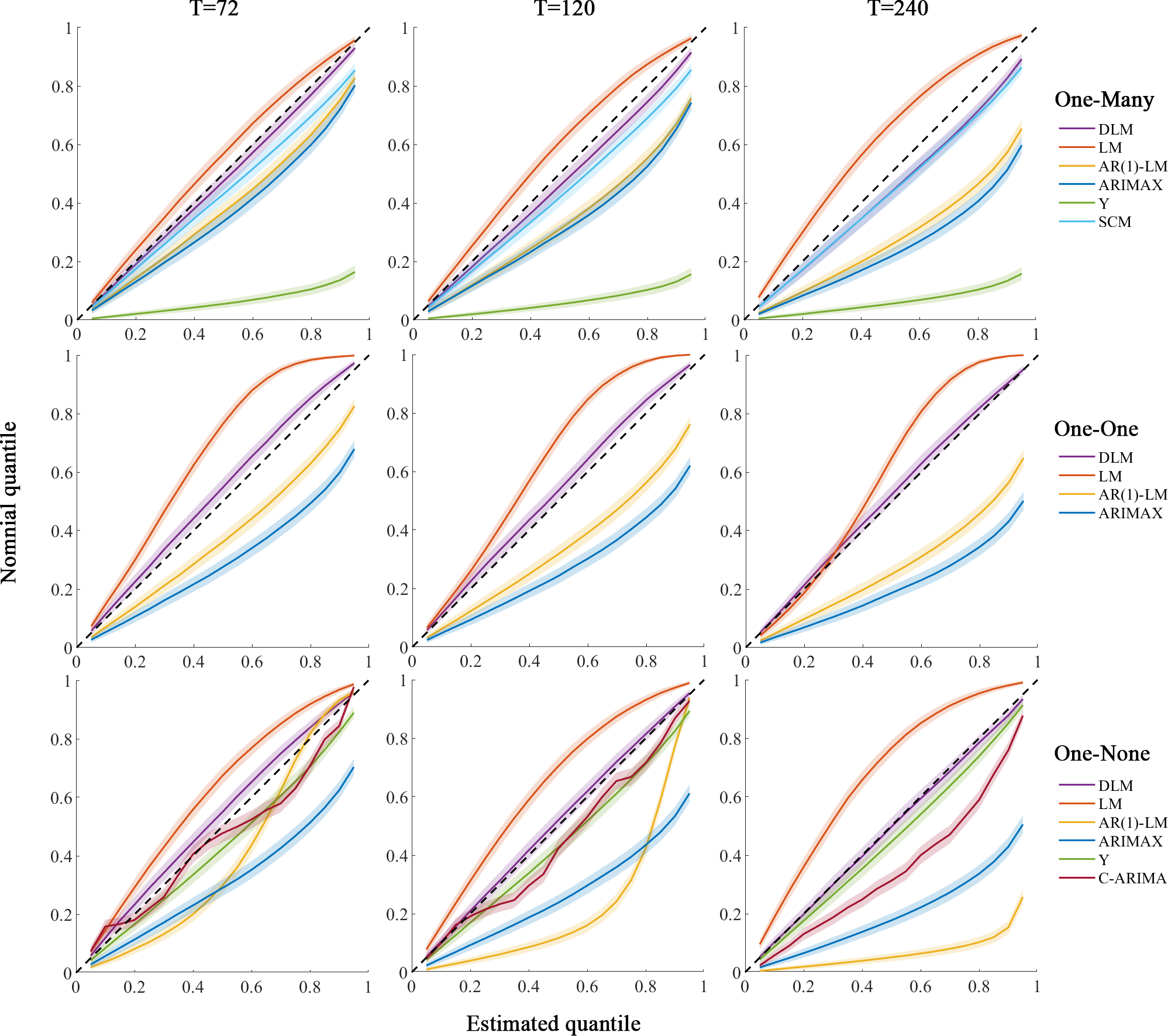}
\caption{Coverage probability for each quantile from 5\% to 95\%. The x-axis is the estimated quantile, while the y-axis is the nominal quantile. The best case is the 45-degree dotted line. Shaded areas represent 95\% intervals of the quantile estimates.
\label{fig:CPi}
}
\end{figure}

\begin{figure}[t!]
\centering
\includegraphics[width=0.9\textwidth,clip]{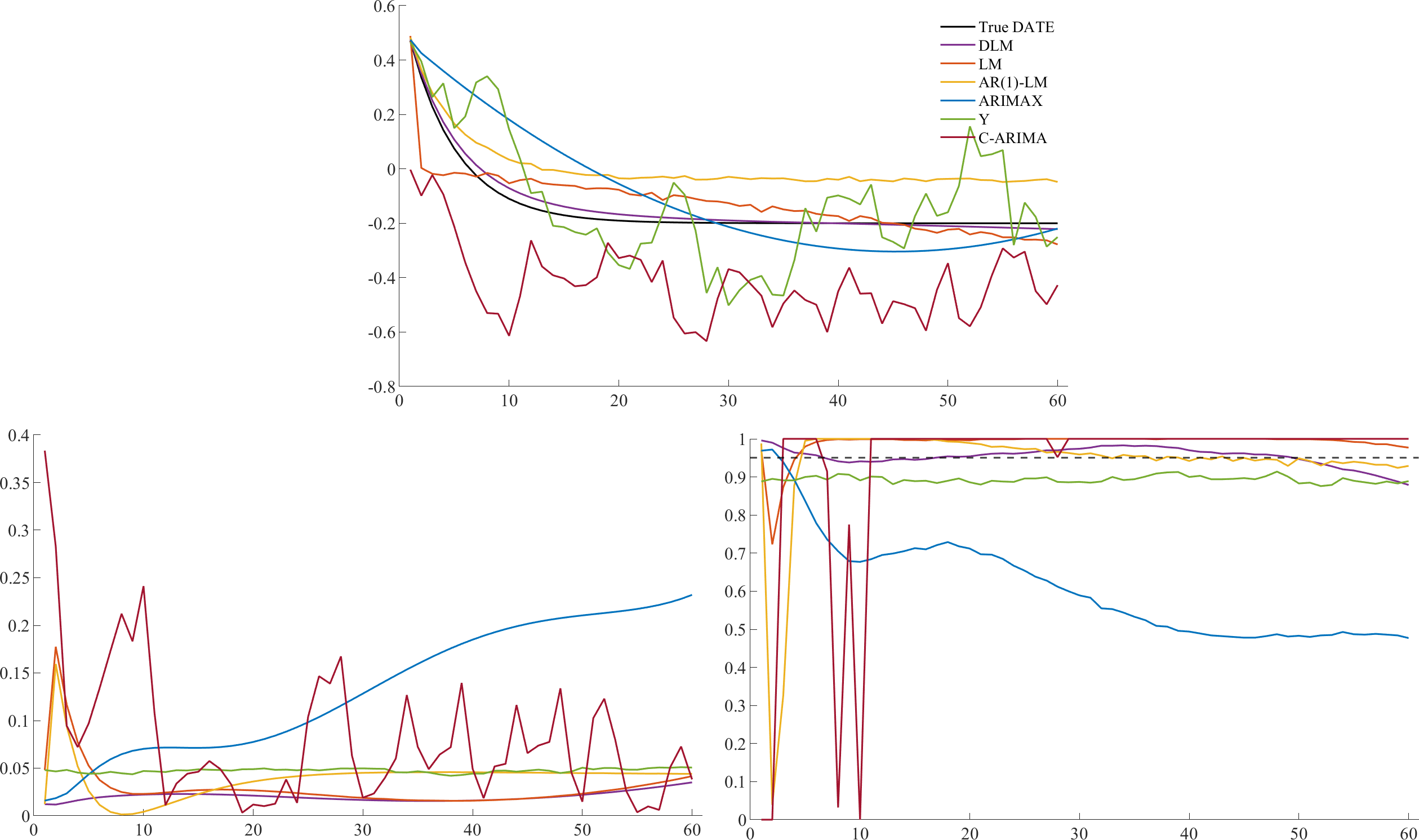}
\caption{Comparison of performance for One-None case with $T=120$.
Top panel: The estimated DATE for all models for one sample (black solid line is the true DATE).
Bottom left panel: Mean squared error for each $t=\treat :T$, averaged over 1,000 experiments.
Bottom right panel: Coverage probability (nominal 95\%) for each $t=\treat :T$, averaged over 1,000 experiments.
\label{fig:ON120}
}
\end{figure}

\begin{table}[t!]
\caption{Mean squared error (MSE) of point estimates and coverage probability (\%) for $T=72,120,240$,  averaged over 1,000 Monte Carlo replications.
The MSE values are standardized to the Many-Many case when $T=72$. \label{table:msecp}
}
\resizebox{\textwidth}{!}{\begin{tabular}{lrrrrrrrrrrrr}
\multicolumn{1}{l|}{Many-Many} & \multicolumn{1}{l|}{MSE}  & \multicolumn{1}{l}{CP (\%)} & \multicolumn{1}{l}{}      & \multicolumn{1}{l}{}       & \multicolumn{1}{l}{}  & \multicolumn{1}{l}{}         & \multicolumn{1}{l}{}    & \multicolumn{1}{l}{}   & \multicolumn{1}{l}{}      & \multicolumn{1}{l}{}       & \multicolumn{1}{l}{}  & \multicolumn{1}{l}{}        \\ \cline{1-3}
\multicolumn{1}{l|}{T=72}      & \multicolumn{1}{l|}{1.00} & \multicolumn{1}{l}{94.43}   & \multicolumn{1}{l}{}      & \multicolumn{1}{l}{}       & \multicolumn{1}{l}{}  & \multicolumn{1}{l}{}         & \multicolumn{1}{l}{}    & \multicolumn{1}{l}{}   & \multicolumn{1}{l}{}      & \multicolumn{1}{l}{}       & \multicolumn{1}{l}{}  & \multicolumn{1}{l}{}        \\
\multicolumn{1}{l|}{120}       & \multicolumn{1}{l|}{1.32} & \multicolumn{1}{l}{94.51}   & \multicolumn{1}{l}{}      & \multicolumn{1}{l}{}       & \multicolumn{1}{l}{}  & \multicolumn{1}{l}{}         & \multicolumn{1}{l}{}    & \multicolumn{1}{l}{}   & \multicolumn{1}{l}{}      & \multicolumn{1}{l}{}       & \multicolumn{1}{l}{}  & \multicolumn{1}{l}{}        \\
\multicolumn{1}{l|}{240}       & \multicolumn{1}{l|}{1.93} & \multicolumn{1}{l}{94.95}   & \multicolumn{1}{l}{}      & \multicolumn{1}{l}{}       & \multicolumn{1}{l}{}  & \multicolumn{1}{l}{}         & \multicolumn{1}{l}{}    & \multicolumn{1}{l}{}   & \multicolumn{1}{l}{}      & \multicolumn{1}{l}{}       & \multicolumn{1}{l}{}  & \multicolumn{1}{l}{}        \\
                               & \multicolumn{1}{l}{}      & \multicolumn{1}{l}{}        & \multicolumn{1}{l}{}      & \multicolumn{1}{l}{}       & \multicolumn{1}{l}{}  & \multicolumn{1}{l}{}         & \multicolumn{1}{l}{}    & \multicolumn{1}{l}{}   & \multicolumn{1}{l}{}      & \multicolumn{1}{l}{}       & \multicolumn{1}{l}{}  & \multicolumn{1}{l}{}        \\
\multicolumn{1}{l|}{}          & \multicolumn{6}{c|}{MSE}                                                                                                                                                & \multicolumn{6}{c}{Coverage Probability (\%)}                                                                                                                   \\
\multicolumn{1}{l|}{One-Many}  & \multicolumn{1}{c}{DLM}   & \multicolumn{1}{c}{LM}      & \multicolumn{1}{c}{AR(1)} & \multicolumn{1}{c}{ARIMAX} & \multicolumn{1}{c}{Y} & \multicolumn{1}{c|}{SCM}     & \multicolumn{1}{c}{DLM} & \multicolumn{1}{c}{LM} & \multicolumn{1}{c}{AR(1)} & \multicolumn{1}{c}{ARIMAX} & \multicolumn{1}{c}{Y} & \multicolumn{1}{c}{SCM}     \\ \hline
\multicolumn{1}{l|}{72}        & 20.01                     & 31.51                       & 36.27                     & 43.98                     & 49.22                 & \multicolumn{1}{r|}{56.39}   & 93.76                   & 95.43                  & 82.12                     & 80.23                      & 16.19                 & 84.89                       \\
\multicolumn{1}{l|}{120}       & 16.19                     & 27.40                       & 36.85                     & 41.20                     & 47.34                 & \multicolumn{1}{r|}{57.10}   & 92.81                   & 96.45                  & 75.91                     & 74.34                    & 15.90                 & 84.76                       \\
\multicolumn{1}{l|}{240}       & 13.32                     & 24.54                       & 41.11                     & 49.70                      & 49.13                 & \multicolumn{1}{r|}{55.27}   & 90.55                   & 97.28                  & 65.32                     & 59.75                      & 15.81                 & 86.11                       \\
                               & \multicolumn{6}{c}{}                                                                                                                                                    & \multicolumn{6}{c}{}                                                                                                                                            \\
\multicolumn{1}{l|}{One-One}   & \multicolumn{1}{c}{DLM}   & \multicolumn{1}{c}{LM}      & \multicolumn{1}{c}{AR(1)} & \multicolumn{1}{c}{ARIMAX}  & \multicolumn{1}{c}{Y} & \multicolumn{1}{c|}{DiD}     & \multicolumn{1}{c}{DLM} & \multicolumn{1}{c}{LM} & \multicolumn{1}{c}{AR(1)} & \multicolumn{1}{c}{ARIMAX}  & \multicolumn{1}{c}{Y} & \multicolumn{1}{c}{DiD}     \\ \hline
\multicolumn{1}{l|}{72}        & 26.25                     & 29.93                       & 65.66                     & 152.81                     & 95.01                 & \multicolumn{1}{r|}{113.58}  & 96.96                   & 99.81                  & 84.00                     & 67.90                      & -                     & -                           \\
\multicolumn{1}{l|}{120}       & 21.41                     & 33.35                       & 74.77                     & 148.72                     & 98.97                 & \multicolumn{1}{r|}{83.04}   & 96.47                   & 99.91                  & 75.12                     & 62.20                     & -                     & -                           \\
\multicolumn{1}{l|}{240}       & 17.52                     & 42.38                       & 79.96                     & 201.43                    & 96.95                 & \multicolumn{1}{r|}{53.78}   & 95.01                   & 99.96                  & 65.17                     & 50.16                      & -                     & -                           \\
                               & \multicolumn{6}{c}{}                                                                                                                                                    & \multicolumn{6}{c}{}                                                                                                                                            \\
\multicolumn{1}{l|}{One-None}  & \multicolumn{1}{c}{DLM}   & \multicolumn{1}{c}{LM}      & \multicolumn{1}{c}{AR(1)} & \multicolumn{1}{c}{ARIMAX}  & \multicolumn{1}{c}{Y} & \multicolumn{1}{c|}{C-ARIMA} & \multicolumn{1}{c}{DLM} & \multicolumn{1}{c}{LM} & \multicolumn{1}{c}{AR(1)} & \multicolumn{1}{c}{ARIMAX}  & \multicolumn{1}{c}{Y} & \multicolumn{1}{c}{C-ARIMA} \\ \hline
\multicolumn{1}{l|}{72}        & 27.52                     & 37.93                       & 34.59                     & 138.63                      & 57.01                 & \multicolumn{1}{r|}{114.68}  & 96.75                   & 98.51                  & 96.19                     & 70.38                      & 88.92                 & 97.73                       \\
\multicolumn{1}{l|}{120}       & 23.33                     & 34.20                       & 46.47                     & 158.26                      & 54.76                 & \multicolumn{1}{r|}{95.60}   & 95.62                   & 98.93                  & 93.88                     & 61.18                     & 89.87                 & 92.79                       \\
\multicolumn{1}{l|}{240}       & 18.58                     & 29.65                       & 83.82                     & 165.11                    & 53.56                 & \multicolumn{1}{r|}{60.29}   & 94.04                   & 99.16                  & 25.94                     & 50.56                      & 91.50                 & 87.75                      
\end{tabular}}
\end{table}

\section{Application: Effect of COVID-19 lockdown on unemployment in the UK}\label{sec:emp}
We illustrate the framework using UK unemployment data around the 2020 COVID-19 lockdown.
The UK government implemented a stay-at-home order on 3/23/2020; we analyze its time-varying effect on unemployment.
This is an interrupted time series (One-None case) with monthly data from 1/2015 to 12/2021 (pre-treatment: 1/2015--2/2020; post-treatment: 3/2020--12/2021).
In the application, we treat post-intervention covariates as unavailable and estimate the counterfactual via the maintained
state-space model; thus causal conclusions are conditional on correct dynamic specification rather than on reweighting-based identification.
We compare five models: DLM, LM, AR(1)-LM, ARIMAX, and C-ARIMA, with specifications identical to the simulation study.

Figure~\ref{fig:app:datecomp} shows how the methods produce different estimates.
LM estimates a flat effect because it cannot capture the arch-shaped trajectory of post-lockdown unemployment.
AR(1)-LM, ARIMAX, and C-ARIMA capture the arch but lag in adapting to the evolving effect due to their static-parameter dynamics.

\begin{figure}[t!]
\centering
\includegraphics[width=0.75\textwidth]{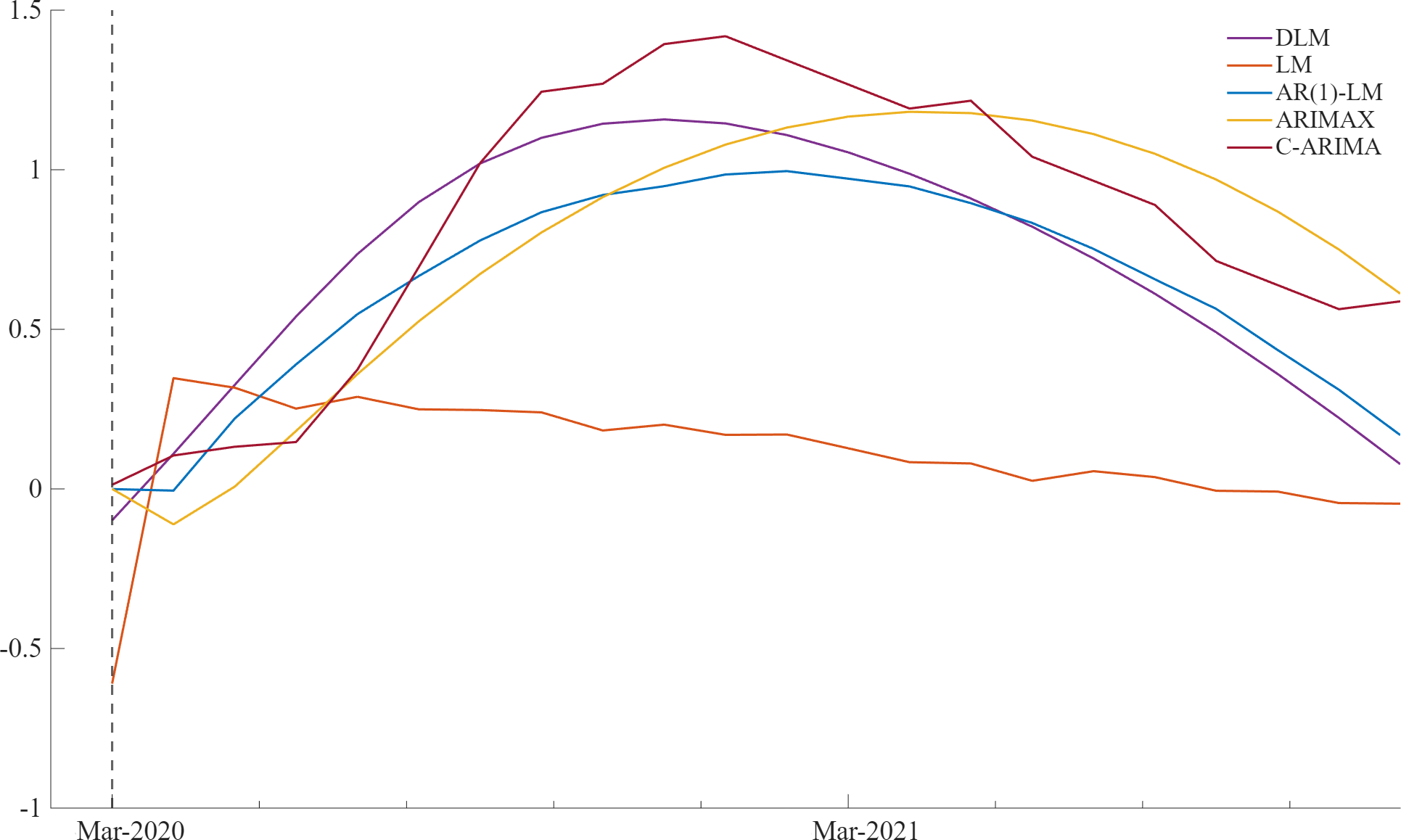}
\caption{Effect of lockdown on unemployment: The estimated mean DATE for the five models. \label{fig:app:datecomp}}
\end{figure}

Focusing on the DLM estimation of the DATE, Figure~\ref{fig:DATEDLM} shows a) top panel: estimation of the mean path of the treated and control, b) middle panel: estimated DATE with DLM over the treatment period, and c) bottom panel: cumulative DATE with DLM over the treatment period.
All results include 95\% credible intervals (shaded areas).
The estimated mean treated path fits the data well, capturing the sudden increase, and slow decrease over time.
The estimated DATE (middle panel) has a curved, arch-like trajectory, showing that the lockdown caused an initial, sudden increase after the lockdown, which took nearly a year to recover.
The cumulative DATE (bottom panel) echoes the above results, where the impact of the lockdown is front-heavy and tapering off at the end of the sampling period.

\begin{figure}[t!]
\centering
\includegraphics[width=0.8\textwidth]{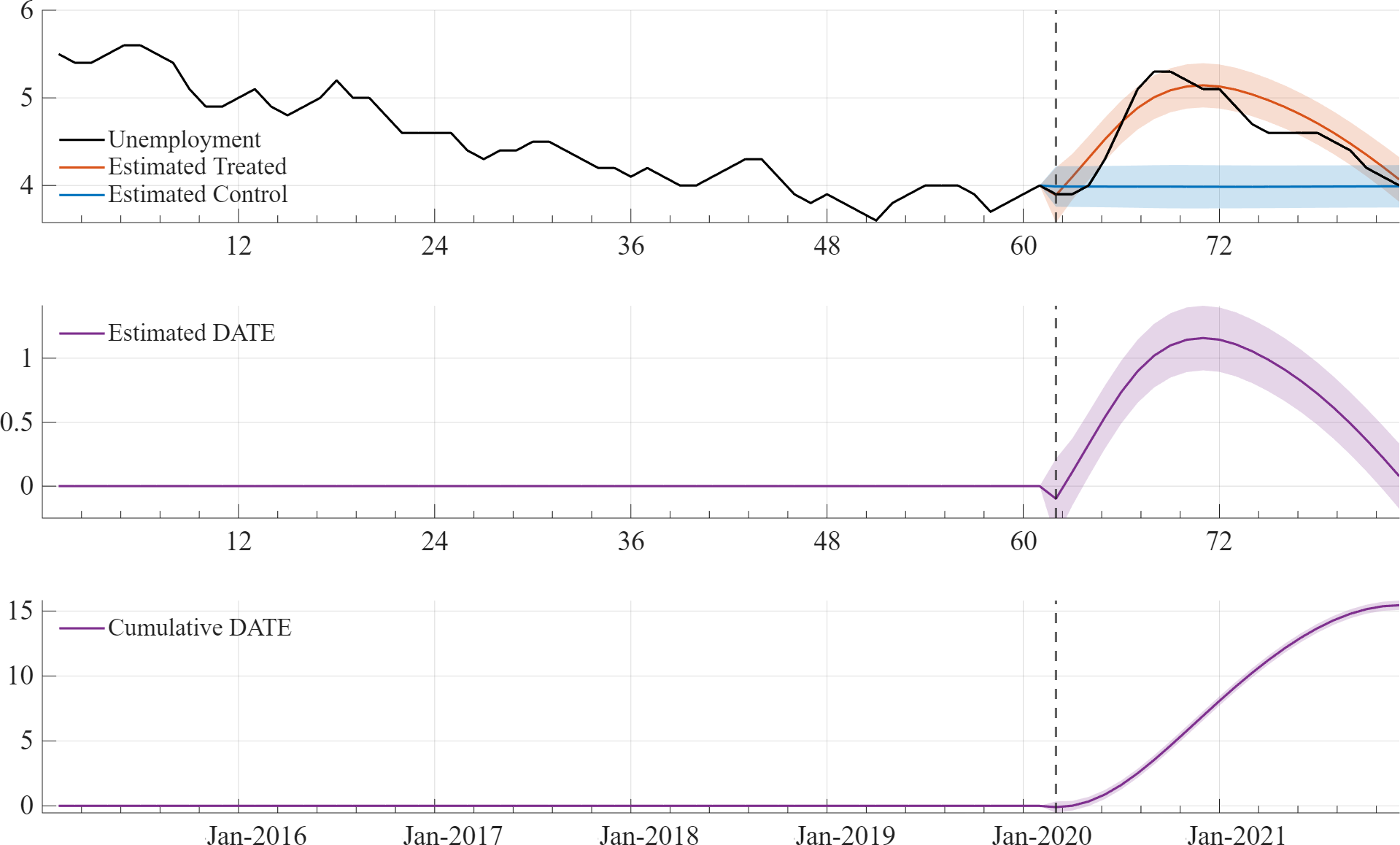}
\caption{Top Figure: The estimated mean path of the treated vs the estimated mean control path of the One-None case for the UK unemployment rate during the COVID-19 lockdown. Middle Figure: DATE of the lockdown intervention. Bottom Figure: The cumulative difference between estimated treated vs control paths.}
\label{fig:DATEDLM} 
\end{figure}

To assess whether the estimation procedure induces artificial dynamic effects, we conduct a placebo
intervention test in Supplementary Material~\ref{sm: robust}. When treatment is reassigned to an earlier, incorrect
date, the estimated treated/control trajectories remain overlapping and the DATE remains near zero,
indicating no spurious detection.

\subsection{Decomposition of the Dynamic Treatment Effect}

The framework permits decomposition of the DATE into interpretable components by activating one indicator at a time (Figure~\ref{fig:date-decomp}).
The spot effect is near zero, indicating no instantaneous jump.
The persistent component shows a pronounced positive effect in the early post-treatment window-- the dominant force driving the initially positive DATE.
The trend component is initially neutral but becomes increasingly negative, signifying deterioration in the post-intervention growth rate.
The cumulative trend impact eventually counterbalances the persistent level shift.
This decomposition distinguishes an absence of immediate impact (spot), a substantial but temporary level lift (persistent), and a structural trajectory change (trend)-- opposing forces that would be obscured in a single aggregated DATE curve.

The spot/persistent/trend breakdown is an interpretable basis for the intervention response within the model specification, not a uniquely
identified structural decomposition. The overall DATE trajectory is the sum of these components, but the allocation of that
trajectory across components depends on the chosen basis (e.g., delayed or nonlinear onsets may be absorbed differently).
We therefore interpret component estimates as descriptive summaries of the fitted intervention profile.

\begin{figure}[t!]
\centering
\includegraphics[width=0.75\textwidth]{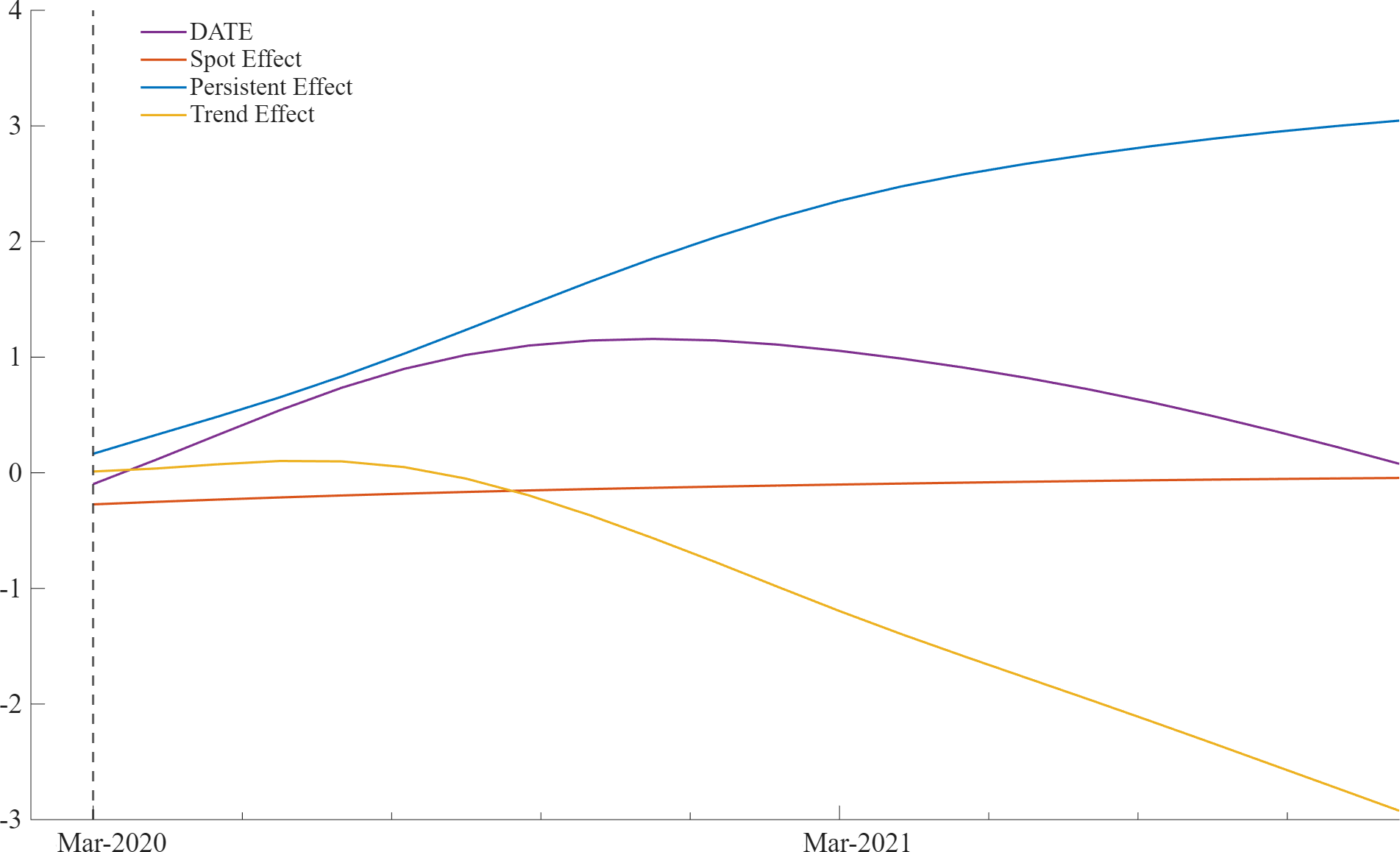}
\caption{Decomposition of the dynamic average treatment effect into the spot, persistent, and trend effect. Each effect corresponds to one row of indicators in the design matrix.}
\label{fig:date-decomp} 
\end{figure}

\section{Summary Comments \label{sec:summ}}

This paper develops a unified framework for causal inference with time series data by extending the i.i.d. potential outcomes model to stochastic processes. We define the dynamic estimand, dynamic average treatment effect (DATE), derive the dynamic inverse probability weighting (DIPW) estimator for observational settings, and show that a dynamic linear model (DLM) admits a state-space representation for the conditional mean dynamics underlying the DATE, yielding a practical DLM implementation when data are limited. Together, these components establish both the theoretical foundation and a practical toolkit for dynamic causal analysis.

Several limitations should be acknowledged.  
First, the framework relies on dynamic ignorability and the \emph{no feedback to
confounders} assumption.  
Violations of this condition, such as adaptive behavioral responses, would require
explicit modeling of the feedback mechanism.  
Second, estimation of dynamic propensity scores in continuous time remains
computationally demanding and sensitive to misspecification.  
These limitations suggest caution in applying the framework to systems with
strong endogeneity or rapid structural breaks.

\singlespacing
\bibliographystyle{chicago}
\bibliography{reference}

\newpage
\doublespacing
\setcounter{equation}{0}
\setcounter{section}{0}
\setcounter{table}{0}
\setcounter{page}{1}
\renewcommand{\thesection}{S\arabic{section}}
\renewcommand{\theequation}{S\arabic{equation}}
\renewcommand{\thetable}{S\arabic{table}}

\vspace{1cm}
\begin{center}
{\LARGE
{\bf Supplementary Material for ``Dynamic causal inference with time series data"}
}
\end{center}

This Supplementary Material provides additional discussions and comparisons, context and definitions of some important concepts in time series, the proofs of the theorems presented in the paper, and a robustness analysis through a placebo test. 

  \section{Comparison with other sequential potential-outcome frameworks}\label{sm:comp}

Several recent papers, including \cite{bojinov2019time,bojinov2021panel}, define trajectories by embedding time into a network structure:
past instances of a unit are treated as ``neighbors,'' and lagged treatments are
viewed as forms of temporal interference. Their formalism can be represented
schematically as
\begin{equation}
Y_t(z_{1:t})
= \Phi_t\!\left( U_{1:t},\, z_{1:t} \right),
\label{eq:schematic-bs-network}
\end{equation}
where $U_t$ denotes an innovation term and $z_{1:t}$ is a treatment path.
This is the standard potential outcome device of holding innovations fixed; unconditional on $U_{1:t}$, the process remains stochastic. Interference across
time arises because the potential outcome at time $t$ depends on the full
history of treatments $z_{1:t}$, analogous to interference across neighbors in a
network.

This network-style formulation treats dynamics as arising from functional
dependence on past shocks and past treatments; causal effects compare
\emph{shock-held-fixed} trajectories under different treatment sequences.

\paragraph{Our stochastic-process formulation.}
In contrast, we do not interpret time as a network nor treat temporal
dependence as interference. Instead, potential outcomes are full stochastic
processes governed by transition laws,
\[
\{Y_t(1)\}_{t>t_c} \sim \mathbb{P}_1, \qquad
\{Y_t(0)\}_{t>t_c} \sim \mathbb{P}_0,
\]
and causal effects compare these two probability laws. Because the transition
laws themselves may change under intervention, this framework allows treatment
to modify persistence, dependence structure, stability, and broader
distributional features; not only the deterministic mapping from shocks to
outcomes.

\paragraph{Conceptual distinctions.}
\begin{itemize}
\item \textit{Role of time.}  
  In shock-based/network formulations, time is encoded as a graph and treatment
  at one time interferes with outcomes at later times through deterministic
  functional dependence. In our formulation, time indexes the evolution of a
  stochastic process; no network reinterpretation is used.

\item \textit{Source of dynamics.}  
  In \eqref{eq:schematic-bs-network}, dynamics arise mechanically from
  recurrence through $\Phi_t$. In our approach, dynamics are properties of the
  underlying probability laws, which may differ under treatment.

\item \textit{Causal contrasts.}  
  Shock-based approaches define contrasts by fixing $U_{1:t}$ and varying the
  treatment sequence. Our DATE compares expectations under two full stochastic
  laws, integrating over population heterogeneity.

\item \textit{Intervention effects.}  
  Network-interference models assume the innovation process is invariant to
  treatment. Our approach allows interventions to change the innovation
  distribution, transition kernel, or other structural components.

\end{itemize}

\medskip
Thus, shock-based and network-interference dynamic potential-outcome frameworks
address causal questions defined over deterministic functions of shocks and
treatment paths, whereas our stochastic-process formulation targets causal
effects on probability laws governing dynamic systems whose transition behavior
may change under intervention.

\subsection*{Illustrating the difference as DAGs}

Figures~\ref{dag:boj} and \ref{dag_dlm} provide graphical illustrations of the dynamic potential
outcome framework in this paper. These diagrams are purely expository and are
not used for identification, which in our development is based on filtrations and
conditional independence in stochastic processes rather than structural
assumptions encoded by a DAG. Figure~\ref{dag:boj} depicts the sequential potential
outcome representation of \cite{bojinov2019time}: at each time $t$, the outcome process branches into
potential trajectories indexed by the treatment history, with the realized
trajectory shown by solid lines and counterfactual trajectories shown by dashed
lines. Figure~\ref{dag_dlm} illustrates the corresponding latent-state formulation of our proposed stochastic process potential outcomes framework, where the latent process $\theta_t$ evolves over time and, upon
intervention, splits into treated and control latent paths. Observations $Y_t$
are generated conditionally on the latent state and inherit its branching
structure. These diagrams complement the formal definitions in the main text by
visualizing how dynamic treatment regimes give rise to full process-level
potential outcomes and how the DATE compares the laws of the treated and control
trajectories.

\begin{figure}[t!]
\centering
\begin{tikzpicture}[
    obsY/.style={inner sep=2pt},
    potY/.style={inner sep=2pt},
    treat/.style={rectangle, draw, minimum size=6mm},
    >=Stealth,
    node distance=2.3cm
]

\node[obsY] (Y0) {$Y_0$};

\node[obsY, right=2.8cm of Y0, yshift=8mm] (Y11) {$Y_1(1)$};
\node[potY, right=2.8cm of Y0, yshift=-8mm] (Y10) {$Y_1(0)$};

\node[obsY, right=2.8cm of Y11, yshift=8mm] (Y211) {$Y_2(1,1)$};
\node[potY, right=2.8cm of Y11, yshift=-8mm] (Y210) {$Y_2(1,0)$};

\node[obsY, right=2.8cm of Y210, yshift=8mm] (Y3101) {$Y_3(1,0,1)$};
\node[potY, right=2.8cm of Y210, yshift=-8mm] (Y3100) {$Y_3(1,0,0)$};

\draw[->, thick] (Y0) -- (Y11);
\draw[->, dashed] (Y0) -- (Y10);

\draw[->, thick] (Y11) -- (Y210);
\draw[->, dashed] (Y11) -- (Y211);

\draw[->, thick] (Y210) -- (Y3101);
\draw[->, dashed] (Y210) -- (Y3100);

\node[above=1.8cm of Y0] {$t=0$};
\node[above=1cm of Y11] {$t=1$};
\node[above=0.2cm of Y211] {$t=2$};
\node[above=1cm of Y3101] {$t=3$};

\end{tikzpicture}
\caption{Sequential potential outcomes \citep{bojinov2019time} DAG for treatment history $Z=(1,0,1)$.}\label{dag:boj}
\end{figure}
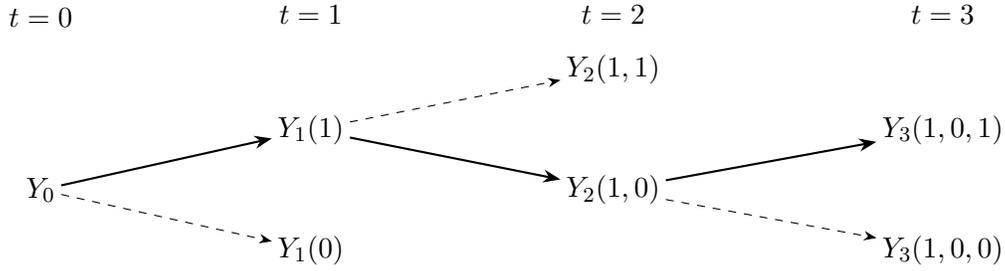

\begin{figure}[t!]
\centering
\begin{tikzpicture}[
    latentObs/.style={inner sep=2pt},
latentCF/.style={inner sep=2pt},
obsY/.style={inner sep=2pt},
potY/.style={inner sep=2pt},
    >=Stealth,
    node distance=2.1cm
]

\node[latentObs] (thm2) {$\theta_{t-2}$};
\node[latentObs, right=of thm2] (thm1) {$\theta_{t-1}$};

\node[latentObs, right=of thm1, yshift=7mm] (tht1) {$\theta_{t}(1)$};
\node[latentCF, right=of thm1, yshift=-7mm] (tht0) {$\theta_{t}(0)$};

\node[latentObs, right=of tht1] (thp1_1) {$\theta_{t+1}(1)$};
\node[latentCF, right=of tht0] (thp1_0) {$\theta_{t+1}(0)$};

\node[latentObs, right=of thp1_1] (thp2_1) {$\theta_{t+2}(1)$};
\node[latentCF, right=of thp1_0] (thp2_0) {$\theta_{t+2}(0)$};

\node[obsY, above=1.6cm of thm2] (ym2) {$Y_{t-2}$};
\node[obsY, above=1.6cm of thm1] (ym1) {$Y_{t-1}$};

\node[obsY, above=1.6cm of tht1] (yt1) {$Y_{t}(1)$};
\node[potY, below=1.6cm of tht0] (yt0) {$Y_{t}(0)$};

\node[obsY, above=1.6cm of thp1_1] (yp1_1) {$Y_{t+1}(1)$};
\node[potY, below=1.6cm of thp1_0] (yp1_0) {$Y_{t+1}(0)$};

\node[obsY, above=1.6cm of thp2_1] (yp2_1) {$Y_{t+2}(1)$};
\node[potY, below=1.6cm of thp2_0] (yp2_0) {$Y_{t+2}(0)$};

\draw[->, thick] (thm2) -- (thm1);

\draw[->, thick] (thm1) -- (tht1);
\draw[->, dashed] (thm1) -- (tht0);

\draw[->, thick] (tht1) -- (thp1_1);
\draw[->, dashed] (tht0) -- (thp1_0);

\draw[->, thick] (thp1_1) -- (thp2_1);
\draw[->, dashed] (thp1_0) -- (thp2_0);

\draw[->, thick] (thm2) -- (ym2);
\draw[->, thick] (thm1) -- (ym1);

\draw[->, thick] (tht1) -- (yt1);
\draw[->, dashed] (tht0) -- (yt0);

\draw[->, thick] (thp1_1) -- (yp1_1);
\draw[->, dashed] (thp1_0) -- (yp1_0);

\draw[->, thick] (thp2_1) -- (yp2_1);
\draw[->, dashed] (thp2_0) -- (yp2_0);

\node[above=0.9cm of ym2] {$t-2$};
\node[above=0.9cm of ym1] {$t-1$};
\node[above=0.2cm of yt1] {$t$};
\node[above=0.2cm of yp1_1] {$t+1$};
\node[above=0.2cm of yp2_1] {$t+2$};

\node[draw, below=1.4cm of yt1, left=0.8cm of yt1] (Zt) {$Z_t$};
\draw[->, thick] (Zt) -- (tht1);
\draw[->, dashed] (Zt) -- (tht0);

\draw[->, thick] ($(thm2)+(-1.4cm,0)$) -- (thm2);

\draw[->, thick] (thp2_1) -- ($(thp2_1)+(1.4cm,0)$);
\draw[->, dashed] (thp2_0) -- ($(thp2_0)+(1.4cm,0)$);

\end{tikzpicture}
\caption{Latent-state representation of dynamic potential outcomes for an
intervention at time $t$. The pre-treatment latent process $\theta_{t-2},
\theta_{t-1}$ is unique. At time $t$, the process branches into two potential
latent trajectories, treated $\theta_s(1)$ and control $\theta_s(0)$ for
$s \ge t$. Solid  edges denote the realized (treated) trajectory;
dashed edges denote the unobserved counterfactual path. Each latent
state realizes an observation $Y_s$ via a vertical link, illustrating how the DATE
compares the laws of the treated and control latent processes and their
associated observations.}\label{dag_dlm}
\end{figure}
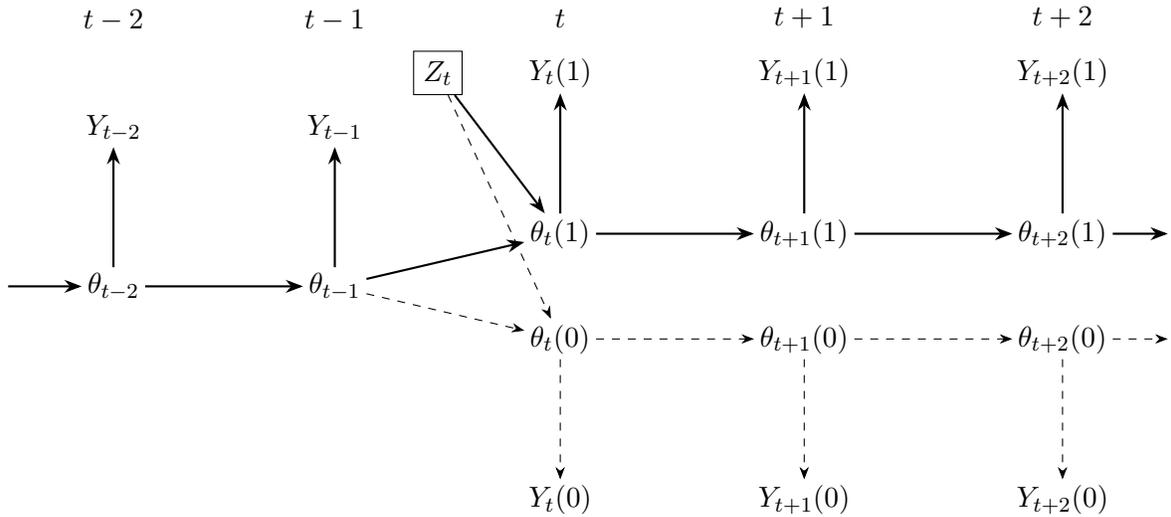

\subsection*{Relation to structural time series and CausalImpact}

Methods such as structural time series models and CausalImpact estimate the
counterfactual mean trajectory under modeling assumptions on trend, seasonality,
and regression structure. These approaches do not define process-level potential
outcomes or causal estimands in the sense of the DATE, and intervention effects
are interpreted through model-based counterfactual predictions rather than
contrasts of stochastic laws. Our framework can be combined with state-space
models when data are limited, but the interpretation of DATE
remains distinct from model-based forecasts.

\section{Additional background: time series probability concepts}\label{sm:comcepts}
\paragraph{Cylinder measure.}
To define and evaluate probabilities in time series, we must first introduce the concept of ``measure'' for stochastic processes.
Since uncertainty in stochastic processes is dependent on time (e.g., the range of outcomes is greater the further away from the current), the mathematical framework of measure must take this into account.
For this, we employ what is known as a ``cylinder measure."
The intuitive idea of a cylinder measure is this: consider the meandering path of a stochastic process over time, and imagine some hoops that the process must go through to succeed (or similarly imagine two pipes that an object has to go through, where the object fails if it touches either pipe).
As we vary the size of these hoops, the proportion of outcomes that succeed in going through changes.
Thus, by varying the size of these hoops, we can measure the probability of some process at a certain time.
If we consider multiple hoops, over time, and take the joint success/failure proportions over all hoops, then we successfully measure the stochastic process, and can define the probability of paths.
This is visualized in Figure~\ref{fig:cylinder}, where the ``T's" represent the upper and lower borders of the cylinder (or hoops in our analogy).

For a more rigorous definition, let $\left(E_{1},\cdots,E_{n}\right)$,
$E_{k}\in\mathscr{B}\left(\mathbb{R}\right),1\leqq k\leqq n$ be members of the Borel set, $\left(\mathbb{R}^{d}\right)^{n}$, and a subset of $\left(\mathbb{R}^{d}\right)^{T}$.
Then, the set, 
\[
A=\left\{ y\in\left(\mathbb{R}^{d}\right)^{T}\left|y\left(t_{1}\right)\in E_{1},\cdots,y\left(t_{n}\right)\in E_{n}\right.\right\} ,\ \ t_{1}<t_{2}<\cdots<t_{n},\ n=1,2,\cdots
\]
is called a {\it cylindrical set}.
Let $\mathscr{B}^{T}$ be the $\sigma$-algebra of the whole cylindrical set.
Fixing $t_{1},\cdots,t_{n}$, the whole cylindrical set obtained by changing $\left(E_{1},\cdots,E_{n}\right)$ is a sub-$\sigma$-field of $\mathscr{B}^{T}$.
This is denoted as $\mathscr{B}_{t_{1},t_{2},\cdots,t_{n}}$.
Given a probability measure, $P$, on the measurable space, $\mathscr{B}_{t_{1},t_{2},\cdots,t_{n}}$, and 
\[
\mu_{t_{1},t_{2},\cdots,t_{n}}\left(A\right)=P\left(\left\{ \omega\left|Y_{t_{1}}\left(\omega\right)\in E_{1},\cdots,Y_{t_{n}}\left(\omega\right)\in E_{n}\right.\right\} \right),
\]
then $\mu_{t_{1},t_{2},\cdots,t_{n}}$ is a probability measure on $\left(\left(\mathbb{R}^{d}\right)^{T},\mathscr{B}_{t_{1},t_{2},\cdots,t_{n}}\right)$.
Similarly for $t_{1},\cdots,t_{m},\left(n<m\right)$, if we construct 
$\mathscr{B}_{t_{1},t_{2},\cdots,t_{m}}$ and 
$\mu_{t_{1},t_{2},\cdots,t_{m}}$, then 
$\mu_{t_{1},t_{2},\cdots,t_{n}}$ and $\mu_{t_{1},t_{2},\cdots,t_{m}}$ satisfies the consistency condition:
\begin{equation}
\mu_{t_{1},t_{2},\cdots,t_{n}}\left(A\right)=\mu_{t_{1},t_{2},\cdots,t_{m}}\left(A\times\mathbb{R}^{m-n}\right).\label{eq:KolmogorovCC}
 \end{equation}
Conversely, if a family of finite dimension probability measure, $\left\{ \mu_{t_{1},t_{2},\cdots,t_{n}}\right\} _{n=1,2,\cdots}$, satisfies the above consistency condition, then there exists one probability measure, $\mu^{*}$, on $\left(\left(\mathbb{R}^{d}\right)^{T},\mathscr{B}^{T}\right)$, such that the restriction on $\mathscr{B}_{t_{1},t_{2},\cdots,t_{n}}$ matches that of $\mu_{t_{1},t_{2},\cdots,t_{n}}$ (Kolmogorov's extension theorem).
The probability measure, $\mu^{*}$, is called the {\it cylinder measure} of the stochastic process (or its distribution), $Y_{t}$, denoted $\mathbb{P}_{Y}$. 

\begin{figure}[t!]
\centering
\includegraphics[width=0.75\textwidth,clip]{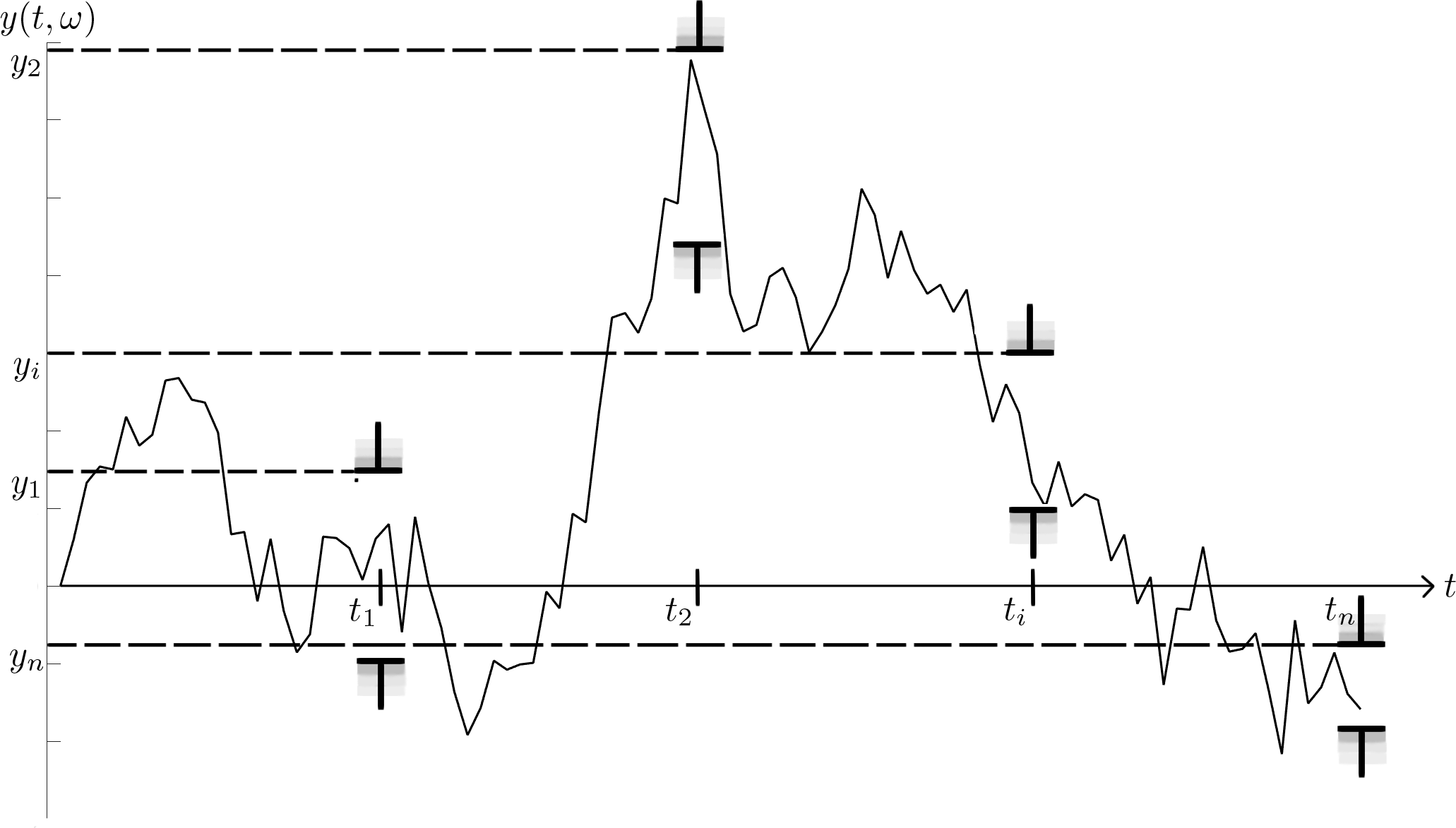}
\caption{Illustration of the cylinder measure. The ``T's" represent the upper and lower borders of the cylinder. The probability is measured as the success/failure proportion of passing within the cylinder, when the cylinder width is varied. The probability of a path is defined as the joint probability over cylinders for each time, $t$.
\label{fig:cylinder}
}
\end{figure}

\paragraph{Filtration.}
An important concept in stochastic process theory is filtration.
Given the framework of sampling in time series, each unit will have multiple paths that it can take.
Before observing anything, there is no information about the outcomes/paths it can take.
However, as time accrues, we are able to learn more about the characteristics of the process and possible paths it can take.
For example, if a process, halfway through, is a large positive number, it is very unlikely that the path will have negative outcomes, even if at the beginning this would be unknown.
The mathematical tools we employ, must take into account this feature of time series and stochastic processes.
Filtration does this by ``filtering out" some paths as time, and information, accrue.
This not only means that observations that have been realized will have filtered out all other paths in the past, but also filters out future paths based on the path taken.

More specifically, in the above discussion, we showed that, when $\omega$ is chosen, the function, $t\mapsto Y_{t}$, is uniquely determined for the stochastic process, $\left\{ Y_{t}\left(\omega\right)\right\} _{t\in\left[0,T\right]}$.
At $t=0$, $\omega$ is completely unknown, and is only fully determined when $t=T$.
For $s\left(<T\right)$, however, we can partially determine $\omega$.
Thus, if $A$ is the cylindrical set, we can partially determine the family of sets of $A$ where $Y_{\cdot}\left(\omega\right)\in A$ is true, and as $s$ increases, that family of sets becomes finer.
The finer $A$ is, the finer the family of subsets of $\Omega$, $Y_{\cdot}^{-1}\left(A\right)$, will be.
As the family of subsets, $Y_{\cdot}^{-1}\left(A\right)$, progresses in time, $\mathscr{F}$ denotes the family of sub-$\sigma$-algebra, $\left\{ \mathscr{F}_{t}\right\} _{t\in T}$.

In general, when $\left\{ \mathscr{F}_{t}\right\} _{t\in T}$ satisfies $\mathscr{F}_{s}\subset\mathscr{F}_{t}$ and $s<t$, it is called the {\it filtration} on $\left(\varOmega,\mathscr{F},P\right)$ (or increasing family of $\sigma$-algebras).
For the stochastic process, $\left\{ Y_{t}\right\} _{t\in T}$, if we set $\mathscr{F}_{t}^{Y}=\sigma\left\{ Y_{s};s\leqq t\right\} $ and $t\in T$, then $\left\{ \mathscr{F}_{t}^{Y}\right\} _{t\in T}$ is a filtration, and called a {\it filtration generated by a stochastic process}.
Filtration is an important concept in stochastic process theory because it describes how information increases over time.
The pair of probability space, $\left(\varOmega,\mathscr{F},P\right)$, and its filtration, $\left\{ \mathscr{F}_{t}\right\} _{t\in T}$, is called a {\it filtered probability space}, and is written as $\left(\varOmega,\mathscr{F},P,\left\{ \mathscr{F}_{t}\right\} _{t\in T}\right)$.

\paragraph{Martingales and predictable processes.}

For a real-valued stochastic process, $\left\{ Y_{t}\right\} _{t\in T}$, on $\left(\Omega,\mathscr{F},\mathbb{P}\right)$ to be a $\left\{ \mathscr{F}_{t}\right\}$-{\it martingale}, $Y_{t}$ must be i) $\mathscr{F}_{t}$-measurable, for each $t\in T$, and ii) $\mathbb{E}\left[Y_{t}\left|\mathscr{F}_{s}\right.\right]=Y_{s}\ \textrm{a.s.}$, for $s<t$.
Martingales may also be written as $\left(\left\{ Y_{t}\right\} ,\left\{ \mathscr{F}_{t}^{X}\right\} \right)$.

For a stochastic process, $\left\{ Y_{t}\right\} $, to be {\it predictable}, a function, $\left(t,\omega\right)\mapsto Y_{t}\left(\omega\right)$, on $\left[0,\infty\right)\times\Omega$ is predictable regarding the predictable $\sigma$-algebra.
A predictable $\sigma$-algebra, $\mathscr{P}$, is the smallest $\sigma$-algebra on $\left[0,\infty\right)\times\Omega$ that makes the projection from $\left[0,\infty\right)\times\Omega$ to $\mathbb{R}$ measurable, for all stochastic processes, $\left\{ Y_{t}\right\} $, that are $\left\{ \mathscr{F}_{t}\right\} $-adaptive and left continuous.
A predictable process for discrete time is when, for time, $s=1,2,\cdots,$ $Y_{n}$ is $\mathscr{F}_{n-1}$-measurable.

\paragraph{Exchangeability.}
One important concept/assumption for causal inference is exchangeability.
While this assumption is relatively simple in the i.i.d. case, it is considerably more complicated in the time series context, as we expound below.

For something to be exchangeable, it must satisfy permutation invariance.
Let $\rho\left(\cdot\right)$ be the permutation group that defines,
\[
P\left(\bigcap_{i=1}^{n}\left(Y_{i}=y_{i},Z_{i}=z_{i},X_{i}=x_{i}\right)\right)=P\left(\bigcap_{i=1}^{n}\left(Y_{i}=y_{\rho\left(i\right)},Z_{i}=z_{\rho\left(i\right)},X_{i}=x_{\rho\left(i\right)}\right)\right).
\]
Assume the model,
\[
Y_{i}=f\left(Z_{i},X_{i},U_{i}\right).
\]
Then, we can derive random assignment from conditional exchangeability:
\[
P\left(\left.\bigcap_{i=1}^{n}\left(f\left(z_{i},X_{i},U_{i}\right)=y_{i}\right)\right|\bigcap_{i=1}^{n}\left(Z_{i}=z_{i}\right)\right)=P\left(\left.\bigcap_{i=1}^{n}\left(f\left(z_{i},X_{\rho\left(i\right)},U_{\rho\left(i\right)}\right)=y_{i}\right)\right|\bigcap_{i=1}^{n}\left(Z_{i}=z_{i}\right)\right).
\]
From conditional exchangeability, we have
\[
\left.f\left(1,X_{j},U_{j}\right)\right|Z_{j}=1\overset{\textrm{d}}{=}\left.f\left(1,X_{k},U_{k}\right)\right|Z_{k}=1,
\]
since $P\left(\bigcap_{i=1}^{n}\left(Z_{i}=z_{i}\right)\right)=P\left(Z_{1}=z_{1},\cdots,Z_{n}=z_{n}\right)$.
Adding the no
interference between units assumption,
\[
\left.f\left(1,X_{j},U_{j}\right)\right|Z_{j}=1\overset{\textrm{d}}{=}\left.f\left(1,X_{k},U_{k}\right)\right|Z_{k}=0,
\]
we have
\[
\left.f\left(1,X_{k},U_{k}\right)\right|Z_{k}=1\overset{\textrm{d}}{=}\left.f\left(1,X_{k},U_{k}\right)\right|Z_{k}=0
\]
and thus obtain, $f\left(1,X_{i},U_{i}\right)\Perp Z_{i}$.
The conditional
exchangeability given $X$ is
\[
P\left(\left.\bigcap_{i=1}^{n}\left(f\left(z_{i},x_{i},U_{i}\right)=y_{i}\right)\right|\bigcap_{i=1}^{n}\left(X_{i}=x_{i},Z_{i}=z_{i}\right)\right)=P\left(\left.\bigcap_{i=1}^{n}\left(f\left(z_{i},x_{i},U_{\rho\left(i\right)}\right)=y_{i}\right)\right|\bigcap_{i=1}^{n}\left(X_{i}=x_{i},Z_{i}=z_{i}\right)\right).
\]

The stochastic process version is given as follows.
Denote the joint cylinder measure of the $n$-unit stochastic process as $\mathbb{P}$.
The exchangeability of a stochastic process is defined as
\[
\mathbb{P}\left(\bigotimes_{i=1}^{n}\left\{ Y_{t,i},Z_{\treat,i},X_{t,i}\right\} \right)=P\left(\bigotimes_{i=1}^{n}\left\{ Y_{t,\rho\left(i\right)},Z_{\treat,\rho\left(i\right)},X_{t,\rho\left(i\right)}\right\} \right).
\]
Let the model be
\[
Y_{t,i}=f\left(Z_{\treat,i},\left\{ X_{s,i}\right\} _{s\leqq t},\left\{ W_{s,i}\right\} _{s\leqq t}\right).
\]
Under random assignment, exchangeability is
\begin{alignat*}{1}
 & \mathbb{P}\left(\left.\bigotimes_{i=1}^{n}\left\{ Y_{t,i},z_{i},X_{t,i},W_{t,i}\right\} \right|\bigcap_{i=1}^{n}\left(Z_{\treat,i}=z_{i}\right)\right)\\
 & =\mathbb{P}\left(\left.\bigotimes_{i=1}^{n}\left\{ Y_{t,i},z_{i},X_{t,\rho\left(i\right)},W_{t,\rho\left(i\right)}\right\} \right|\bigcap_{i=1}^{n}\left(Z_{\treat,i}=z_{i}\right)\right).
 \end{alignat*}
and, for observation data, it is
\begin{alignat*}{1}
 & \mathbb{P}\left(\left.\bigotimes_{i=1}^{n}\left\{ Y_{t,i},z_{i},x_{t,i},W_{t,i}\right\} \right|\bigcap_{i=1}^{n}\left(\left\{ X_{s,i}\right\} =\left\{ x_{s,i}\right\} ,Z_{\treat,i}=z_{i}\right),0\leqq s\leqq T\right)\\
 & =\mathbb{P}\left(\left.\bigotimes_{i=1}^{n}\left\{ Y_{t,i},z_{i},x_{t,i},W_{t,\rho\left(i\right)}\right\} \right|\bigcap_{i=1}^{n}\left(\left\{ X_{t,i}\right\} =\left\{ x_{t,i}\right\} ,Z_{\treat,i}=z_{i}\right),,0\leqq s\leqq T\right).
 \end{alignat*}
This, however, is exchangeability under known, future (post-treatment) $\left\{ X_{t,i}\right\} $, and is the weakest assumption for observational data.
A more realistic definition of exchangeability is assuming only the pre-treatment $\left\{ X_{t,i}\right\} $ is known:
\begin{alignat*}{1}
 & \mathbb{P}\left(\left.\bigotimes_{i=1}^{n}\left\{ Y_{t,i},z_{i},x_{t,i},W_{t,i}\right\} \right|\bigcap_{i=1}^{n}\left(\left\{ X_{s,i}\right\} =\left\{ x_{s,i}\right\} ,Z_{\treat,i}=z_{i}\right),s\leqq\treat\right)\\
 & =\mathbb{P}\left(\left.\bigotimes_{i=1}^{n}\left\{ Y_{t,i},z_{i},x_{t,i},W_{t,\rho\left(i\right)}\right\} \right|\bigcap_{i=1}^{n}\left(\left\{ X_{s,i}\right\} =\left\{ x_{s,i}\right\} ,Z_{\treat,i}=z_{i}\right),s\leqq\treat\right).
 \end{alignat*}
This is a stronger, but more realistic assumption.

\section{Proofs}
\subsection{Preliminary}

We first summarize martingale theory, as it pertains to this paper.
For detailed explanations and proofs, see \cite{Ikeda-Watanabe_89}, \cite{Revuz-Yor_99}.

Assume that the filtration, $\left\{ \mathscr{F}_{t}\right\} $, is given for a complete probability space, $\left(\varOmega,\mathscr{F},P\right)$.
Let $X_{t}^{*}=\sup_{0\leqq s\leqq t}\left|X_{t}\right|$.
Let $\mathscr{M}^{2}=\mathscr{M}^{2}\left(\mathscr{F}_{t}\right)$ be the entirety of {\it square integrable martingales},
\[
M=\left\{ M_{t}\right\} \text{:}\sup_{t\geqq0}\mathbb{E}\left[M_{t}^{2}\right]<\infty,
\]
regarding $\left\{ \mathscr{F}_{t}\right\} $.
Then, $M_{\infty}=\lim_{t\rightarrow\infty}M_{t}$ exists.
The inner product of $\mathscr{M}^{2}$, $\mathbb{E}\left[M_{\infty}N_{\infty}\right]$, is a Hilbert space, and we can introduce $\left\Vert M-N\right\Vert =\left\Vert M_{\infty}-N_{\infty}\right\Vert _{2}$ as a distance (note that $\left\Vert \cdot\right\Vert _{2}$ is a $L^{2}\left(P\right)$-norm).

Regarding $M\in\mathscr{M}^{2}$, we define the {\it quadratic variation process}, $\left[M,M\right]$, as follows.
\[
\text{Partition}\ \varDelta:0=t_{1}<t_{1}<\cdots<t_{n}=t\ \text{satisfies}\ \max_{1\leqq i\leqq n}\left\{ \left|t_{i}-t_{i-1}\right|\right\} \rightarrow0,
\]
$M_{0}^{2}+\sum_{i=1}^{n}\left(M_{t_{i}}-M_{t_{i-1}}\right)^{2}$ converges in probability.
Denote this limit as $\left\langle M,M\right\rangle _{t}$.
Then, $\left\{ \left\langle M,M\right\rangle _{t}\right\} $ is a quadratic variation process of $M$.
Regarding $M,N\in\mathscr{M}^{2}$, we define
\[
\left\langle M,N\right\rangle _{t}=\left(1/4\right)\left\{ \left\langle M+N,M+N\right\rangle _{t}-\left\langle M-N,M-N\right\rangle _{t}\right\},
\]
then $M_{t}N_{t}-\left\langle M,N\right\rangle _{t}$ is a martingale.
This $\left\langle M,M\right\rangle $ is called the {\it predictable quadratic variation process} of $M$.
When $M,N\in\mathscr{M}^{2}$ satisfies $\left\langle M,N\right\rangle =0$, then $M$ and $N$ is said to be {\it orthogonal}, and denoted $M\perp N$.
Since $M_{t}N_{t}-\left\langle M,N\right\rangle _{t}$ is a martingale, $M\perp N$ is equivalent to the product, $M_{t}N_{t}$, being a martingale.

Let the Hilbert space, $L^{2}\left(\left\langle M\right\rangle \right)$, be
\[
L^{2}\left(\left\langle M\right\rangle \right)=\left\{ f_{s}\left|\textrm{predictable and for each }t,\ \mathbb{E}\left[\int_{0}^{t}\left|f_{s}\right|^{2}d\left\langle M\right\rangle _{s}\right]<+\infty\right.\right\},
\]
and let the semi-norm of $f\in L^{2}\left(\left\langle M\right\rangle \right)$ be
\[
\left\Vert f\right\Vert _{t}\left(\left\langle M\right\rangle \right)=\mathbb{E}\left[\int_{0}^{t}\left|f_{s}\right|^{2}d\left\langle M\right\rangle _{s}\right]^{1/2},\ t\in\mathbb{R}_{+}.
\]
This is a complete metric space.
Let $\mathfrak{L}_{0}$ denote the subset of simple functions, $L^{2}\left(\left\langle M\right\rangle \right)$, represented by bounded $\mathscr{F}_{t}$-adapted processes, $f_{t}$.
Thus, for the partition, $0=t_{0}<t_{1}<t_{2}<\cdots<t_{n+1}=\infty$, of $\mathbb{R}_{+}$, there are $\mathscr{F}_{t_{k}}$-measurable functions, $f_{k}$, and $\mathscr{F}_{0}$-measurable function, $f_{-1}$, and
\[
f_{t}=\begin{cases}
f_{-1} & \left(t=0\right),\\
f_{k} & \left(t_{k}<t\leqq t_{k+1}\right).
\end{cases}
\]

\begin{lem}
\cite{Ikeda-Watanabe_89} Chapter2. Lemma 2.1. 
{\it $\mathfrak{L}_{0}$ is a dense subset of $L^{2}\left(\left\langle M\right\rangle \right)$.}
\end{lem}

For the partition, $0=t_{0}<t_{1}<t_{2}<\cdots<t_{n+1}=\infty$, of $\mathbb{R}_{+}$ that determines $\mathfrak{L}_{0}$, the Stiltjes integral,
\[
X_{t}=f_{-1}M_{0}+\sum_{z=0}^{n}f_{t_{k}}\left(M_{t\wedge t_{k+1}}-M_{t\wedge t_{k}}\right)
\]
is called the probability integral of $f_{k}\in\mathfrak{L}_{0}$ by $M_{t}$ and denoted as $\int_{0}^{t}f_{s}dM_{s}$.
For general $f_{t}\in L^{2}\left(\left\langle M\right\rangle \right)$, the probability integral is defined as follows.
Denote the sequence of functions, $\left\{ f^{\left(n\right)}\right\} $, that belongs to $\mathfrak{L}_{0}$ that satisfy $\left\Vert f^{\left(n\right)}-f\right\Vert _{L^{2}\left(\left\langle M\right\rangle \right)}^{2}\rightarrow0$  as $X_{t}^{\left(n\right)}=\int_{0}^{t}f_{s}^{\left(n\right)}dM_{s}$.
There exists a subsequence of $X_{t}^{\left(n\right)}$ that uniform
convergence on compact sets almost surely to some limit $X_{t}$.
When $t$ is fixed, this convergence is also $L^{2}\left(\left\langle M\right\rangle \right)$-convergence.
This $X_{t}$ is called the probability integral of $f_{t}$ by $M$, and is denoted as $\int f_{s}dM_{s}$.
Let the entirety of the probability integral of $M$ be
\[
\mathfrak{L}\left(M\right)=\left\{ \int f_{s}dM_{s};\ f\in L^{2}\left(\left\langle M\right\rangle \right)\right\}.
\]
If $M\perp N$, then $\left\langle \int f\cdot dM,N\right\rangle =\int f\cdot d\left\langle M,N\right\rangle =0$, and $N$ is orthogonal to all elements of $\mathfrak{L}\left(M\right)$.
In general,
\begin{thm}
\cite{Ikeda-Watanabe_89} p.81. 
{\it For any square integral martingale, $M,N\in\mathscr{M}^{2}\left(\mathscr{F}_{t}\right)$, $N$ is expressed as the sum of elements in $\mathfrak{L}\left(M\right)$ and elements that are orthogonal to $\mathfrak{L}\left(M\right)$.}
\end{thm}

A {\it semimartingale} regarding the filtration, $\left\{ \mathscr{F}_{t}\right\} $, on $\left(\Omega,\mathscr{F},\mathbb{P}\right)$ is a stochastic process, $\left\{ X_{t}\right\} $, that has the following representation:
\begin{align*}
X_{t}=X_{0}+M_{t}+V_{t}, &  & M_{0}=V_{0}=0.
 \end{align*}
Here, $X_{0}$ is a $\mathscr{F}_{0}$-measurable random variable, $M\in\mathscr{M}^{2}$, $V$ are bounded variation processes.
Assume that the right continuous left limit exists for $M,V$.
This decomposition is called the {\it semimartingale decomposition}.
This decomposition is not unique.
When the bounded variation $V$ is locally integrable bounded variation and predictable, this stochastic process, $\left\{ X_{t}\right\} $, is a {\it special semimartingale}, which its decomposition is unique.
The decomposition of a continuous semimartingale is unique.

Almost all continuous path stochastic processes are special semimatingales.
For example, when the stochastic process, $X$, is $\mathscr{F}_{t}$-adapted, right continuous, $\mathbb{E}\left[\left|X_{t}\right|\right]<\infty,^{\forall}t$, and the variation,
\[
\sup_{0=s_{0}<s_{1}<\cdots<s_{n}=t}\sum_{i=1}^{n}\left|X_{s_{i}}-X_{s_{i-1}}\right|
\]
is integrable, the stochastic process, $X=\left\{ X_{t}\right\} _{t\leqq\infty}$, is called a {\it quasimartingale} on $\left[0,\infty\right]$.
A quasimartingale is a special semimartingale.
Therefore, if a stochastic process, $\left\{ X_{t}\right\} $, satisfies the conditions of quasisemimartingale, it is a semimartingale.

The necessary and sufficient conditions for a stochastic process, $X=\left\{ X_{t}\right\} _{t\leqq\infty}$, to be a quasimartingale (up to infinity) is for the positive right continuous supermatingale, $Y,Z$, to be decomposed to its difference, $X=Y-Z$ \citep{dellacherie2011probabilities}.
A right continuous supermartingale is a special semimartingale and if you take the difference of the two supermartingales, the quasimartingales are all special semimartingales.

The semimartingale version of the martingale orthogonal decomposition theorem is the following  F\"{o}llmer-Schweizer decomposition theorem.
\begin{thm}
\cite{monat1995follmer}
{\it Let $X=\left\{ X_{t}\right\} _{0\leq t\leq T}$ be a $\mathbb{R}^{d}$-valued semimartingale: $X_{t}=X_{0}+M_{t}+A_{t}$.
Assume $M\in\mathscr{M}^{2}\left(\mathscr{F}_{t}\right),M_{0}=0$.
Any $\mathscr{F}=\mathscr{F}_{T}$-measurable random variable, $H\in L^{2}\left(\Omega,\mathscr{F},P\right)$, can be decomposed as
\begin{equation}
H=H_{0}+\int_{0}^{T}\xi_{s}^{H}dX_{s}+L_{T}^{H},\ P\textrm{-a.s.}\label{eq:CausalTimeSeries_FS}
 \end{equation}
Here, $H_{0}$ is a $\mathscr{F}_{0}$-measurable random variable, and $\int_{0}^{T}\xi_{s}^{H}dX_{s}$ is a probability integral.
$L^{H}$ is a martingale orthogonal to $M$, and $L^{H}\in\mathscr{M}^{2}\left(\mathscr{F}_{t}\right),L_{0}^{H}=0$ and $\left\langle L^{H},M\right\rangle =0$.
$\left(H_{0},\xi^{H},L^{H}\right)$ is uniquely determined.
\Eqno{CausalTimeSeries_FS} is called the F\"{o}llmer-Schweizer decomposition.}
\end{thm}

\Eqno{CausalTimeSeries_FS} can be seen as a stochastic process that projects any martingale to a space spanned by a given semimartingale.
This is because, if $M_{t}^{Y}=\mathbb{E}\left[H\left|\mathscr{F}_{t}\right.\right]$, $M_{t}^{Y}$ is a martingale,and for any martingale, $M_{t}^{Y}$, we have $M_{t}^{Y}=\mathbb{E}\left[M_{T}^{Y}\left|\mathscr{F}_{t}\right.\right]$.
Therefore, the semimartingale, $Y_{t}=Y_{0}+M_{t}^{Y}+A_{t}^{Y}$, can be expressed as
\begin{equation}
Y_{t}=Y_{0}+A_{t}^{Y}+\int_{0}^{t}\xi_{s}^{Y}dX_{s}+L_{t}^{Y},\ \left\langle M,L^{Y}\right\rangle =0,\label{eq:CausalTimeSeries_FS2}
 \end{equation}

The error term in  \eqno{CausalTimeSeries_FS2}, $L_{t}^{Y}$, is a martingale stochastic process that is orthogonal, $\left\langle M,L^{Y}\right\rangle =0$.
Thus, $\mathbb{E}\left[\varDelta L_{t}\varDelta M_{t}\right]=0$ holds.
This does not mean that the conditional expectation is zero, $\mathbb{E}\left[\varDelta L_{t}\left|\mathscr{F}_{t}\right.\right]=0$.
However, if we assume the following assumption regarding the martingale portion, $M_{t}$, of the covariate, $X_{t}$, we have $\mathbb{E}\left[\varDelta L_{t}\left|\mathscr{F}_{t}\right.\right]=0$.
\begin{thm}
{\it $\left\langle \boldsymbol{\eta},\boldsymbol{M}_{t}\right\rangle $
are martingales with regard to $^{\forall}\boldsymbol{\eta}\in\mathbb{R}^{d}$,
and assume, 
\begin{equation}
\mathbb{E}\left[\left.\left\langle \boldsymbol{\eta},\varDelta\boldsymbol{M}_{t}\right\rangle ^{2}\right|\left\{ Y_{s},\boldsymbol{X}_{s}\right\} _{0\leqq s\leqq t}\right]=\left\langle \boldsymbol{\eta},\Sigma\left(t\right)\boldsymbol{\eta}\right\rangle ,\label{eq:CausalTimeSeries_BrownRep}
 \end{equation}
where $\Sigma\left(t\right)$ is a $\left\{ Y_{s},\boldsymbol{X}_{s}\right\} _{0\leqq s\leqq t}$-measurable
$\mathbb{R}^{d\times d}$ matrix value function. Then, 
\[
\mathbb{E}\left[\varDelta L_{t}\left|\boldsymbol{X}_{t+1},\left\{ Y_{s},\boldsymbol{X}_{s}\right\} _{s=1}^{t}\right.\right]=0,
\]
holds.}
\end{thm}

As a result, under the assumption of \eqno{CausalTimeSeries_BrownRep}, if there are no unobserved confounders, the right hand side of \eqno{CausalTimeSeries_FS2} is the correctly identified regression model.

\subsection{Proof of Theorem~\ref{thm:dipw}}\label{sm:dipw}
\subsubsection{Unbiasedness}
The DIPW estimator is defined as
\begin{align}
\widehat{\mu}^{\,1}_{t,\mathrm{DIPW}}
    &= \frac{1}{n}\sum_{i=1}^n \frac{Z_i}{p_i}\,Y_{t,i}, \label{eq:sDIPW}\\
\widehat{\mu}^{\,0}_{t,\mathrm{DIPW}}
    &= \frac{1}{n}\sum_{i=1}^n \frac{1-Z_i}{1-p_i}\,Y_{t,i}. \nonumber
\end{align}
Then, because
\begin{alignat*}{1}
\mathbb{E}\left[\frac{ZY_{t}}{p\left(\left.Z=1\right|\mathscr{F}_{\treat}^{X}\right)}\right] & =\mathbb{E}\left[\frac{Z^{2}Y_{t}(Z=1)+Z\left(1-Z\right)Y_{t}(Z=0)}{p\left(\left.Z=1\right|\mathscr{F}_{\treat}^{X}\right)}\right]=\mathbb{E}\left[\frac{ZY_{t}(Z=1)}{p\left(\left.Z=1\right|\mathscr{F}_{\treat}^{X}\right)}\right]\\
 & =\mathbb{E}\left[\mathbb{E}\left[\left.\frac{Z}{p\left(\left.Z=1\right|\mathscr{F}_{\treat}^{X}\right)}\right|\mathscr{F}_{\treat}^{X}\right]\mathbb{E}\left[\left.Y_{t}(Z=1)\right|\mathscr{F}_{\treat}^{X}\right]\right]\\
 & =\mathbb{E}\left[\mathbb{E}\left[\left.Y_{t}(Z=1)\right|\mathscr{F}_{\treat}^{X}\right]\right]=\mathbb{E}\left[Y_{t}(Z=1)\right].\\
\mathbb{E}\left[\frac{\left(1-Z\right)Y_{t}}{1-p\left(\left.Z=1\right|\mathscr{F}_{\treat}^{X}\right)}\right] & =\mathbb{E}\left[\frac{\left(1-Z\right)ZY_{t}(Z=1)+\left(1-Z\right)^{2}Y_{t}(Z=0)}{1-p\left(\left.Z=1\right|\mathscr{F}_{\treat}^{X}\right)}\right]=\mathbb{E}\left[\frac{\left(1-Z\right)Y_{t}(Z=0)}{1-p\left(\left.Z=1\right|\mathscr{F}_{\treat}^{X}\right)}\right]\\
 & =\mathbb{E}\left[\mathbb{E}\left[\left.\frac{\left(1-Z\right)}{1-p\left(\left.Z=1\right|\mathscr{F}_{\treat}^{X}\right)}\right|\mathscr{F}_{\treat}^{X}\right]\mathbb{E}\left[\left.Y^{0}\right|\mathscr{F}_{\treat}^{X}\right]\right]\\
 & =\mathbb{E}\left[\mathbb{E}\left[\left.Y_{t}(Z=0)\right|\mathscr{F}_{\treat}^{X}\right]\right]=\mathbb{E}\left[Y_{t}(Z=0)\right],
 \end{alignat*}
which is an unbiased estimator.

\subsubsection{Consistency}
\begin{proof}
First, fix each $t$ for the pre-treatment period, $0\leqq t<\treat$.
From the strong law of large numbers, we have,
\[
\frac{1}{n}{\displaystyle \sum_{i=1}^{n}}Y_{t}\overset{\textrm{a.s.}}{\rightarrow}\mathbb{E}\left[Y_{t}\right],\ \textrm{as }n\rightarrow\infty.
\]
Then, there exists a partition, $0\leqq t_{0}<\cdots<t_{k}<t_{k+1}=\treat$, on the interval, $\left[0,\treat\right]$, for any $\varepsilon>0$, such that $\left|\mathbb{E}\left[Y_{t_{i}}-Y_{t_{i-1}}\right]\right|<\varepsilon$,
$\left(0\leqq i\leqq k\right)$.
Since there is a jump at $t=\treat$, $t_{k+1}$ is set so that it overlaps with the delimiter of the partition $\treat$.
Then, for $t_{i-1}\leqq t<t_{i}$, we have
\begin{alignat*}{1}
\frac{1}{n}{\displaystyle \sum_{i=1}^{n}}Y_{t_{i-1}}-\mathbb{E}\left[Y_{t_{i-1}}\right]-\varepsilon & \leqq\frac{1}{n}{\displaystyle \sum_{i=1}^{n}}Y_{t_{i-1}}-\mathbb{E}\left[Y_{t_{i-1}}\right]\leqq\frac{1}{n}{\displaystyle \sum_{i=1}^{n}}Y_{t_{i}-}-\mathbb{E}\left[Y_{t_{i}-}\right]+\varepsilon.
 \end{alignat*}
As a result, we have $\limsup_{t\in\left[0,\treat\right)}\left|\widehat{\mathbb{E}\left[Y_{t}^{z}\right]}-\mathbb{E}\left[Y_{t}^{z}\right]\right|\leqq\varepsilon$,
$\left(z=0,1\right)$.

Next, fixing each $t$ for the post-treatment, $\treat\leqq t\leqq T$, from the strong law of large numbers, we have,
\[
\frac{1}{n}{\displaystyle \sum_{i=1}^{n}}\frac{Z_{i}}{p\left(\left.Z_{i}=1\right|\mathscr{F}_{\treat}^{X}\right)}Y_{t}\overset{\textrm{a.s.}}{\rightarrow}\mathbb{E}\left[Y_{t}\right],\ \textrm{as }n\rightarrow\infty.
\]
Then, there exists a partition, $\treat\leqq t_{0}<\cdots<t_{k}=T$, on the interval, $\left[\treat,T\right]$, for any $\varepsilon>0$, such that $\left|\mathbb{E}\left[Y_{t_{i}}-Y_{t_{i-1}}\right]\right|<\varepsilon$,
$\left(0\leqq i\leqq k\right)$.
Then, for $t_{i-1}\leqq t<t_{i}$, we have
\begin{alignat*}{1}
\frac{1}{n}{\displaystyle \sum_{i=1}^{n}}\frac{Z_{i}}{p\left(\left.Z_{i}=1\right|\mathscr{F}_{\treat}^{X}\right)}Y_{t_{i-1}}-\mathbb{E}\left[Y_{t_{i-1}}\right]-\varepsilon & \leqq\frac{1}{n}{\displaystyle \sum_{i=1}^{n}}\frac{Z_{i}}{p\left(\left.Z_{i}=1\right|\mathscr{F}_{\treat}^{X}\right)}Y_{t_{i-1}}-\mathbb{E}\left[Y_{t_{i-1}}\right]\\
 & \leqq\frac{1}{n}{\displaystyle \sum_{i=1}^{n}}\frac{Z_{i}}{p\left(\left.Z_{i}=1\right|\mathscr{F}_{\treat}^{X}\right)}Y_{t_{i}-}-\mathbb{E}\left[Y_{t_{i}-}\right]+\varepsilon.
 \end{alignat*}
As a result, we have $\limsup_{t\in\left[\treat,T\right]}\left|\widehat{\mathbb{E}\left[Y_{t}^{z}\right]}-\mathbb{E}\left[Y_{t}^{z}\right]\right|\leqq\varepsilon$,
$\left(z=0,1\right)$.
\end{proof}

\section{Proof of Theorem~\ref{thm:approx}}\label{sm:approx}

Note that a special semimartingale, $\left\{ Y_{t}\right\} _{t\in\left[0,T\right]}$, regarding filtration, $\left\{ \mathscr{F}_{t}\right\} $, on $\left(\Omega,\mathscr{F},\mathbb{P}\right)$ is a continuous time stochastic process that can be expressed as
\begin{equation}
Y_{t}=Y_{0}+V_{t}^{Y}+M_{t}^{Y}.\label{eq:TSCausal_semi}
 \end{equation}
Here, $Y_{0}$ is a $\mathscr{F}_{0}$-measurable random variable, $M^{Y}\in\mathscr{M}^{2}$ is a square integrable martingale,  and $V^{Y}$ is local integrable bounded variation and predictable.
Then, we have the following theorem:
\begin{thm}\label{thm:martingale}

{\it Let $\left\{ Y_{t}(Z=1),Y_{t}(Z=0)\right\} _{t\in\left[0,T\right]}$ satisfy the following assumptions: 

\noindent
$\left\{ Y_{t}(Z=1),Y_{t}(Z=0)\right\} _{t\in\left[0,T\right]}$ is a stochastic process with continuous paths, and $\left\{ Y_{t}(Z=1),Y_{t}(Z=0)\right\} _{t\in\left[0,T\right]}$ is $\mathscr{F}_{t}$-adapted and square integrable,
\[
\sup_{t\in\left[0,T\right]}\mathbb{E}\left[Y_{t}^{2}\right]<\infty.
\]
Further, let the mean variation be bounded,
\[
\sup_{\delta}\mathbb{E}\left[\sum_{i=0}^{n-1}\left|\mathbb{E}\left[Y_{t_{i+1}}-Y_{t_{i}}\left|\mathscr{F}_{t_{i}}\right.\right]\right|+\left|X_{t_{n}}\right|\right]<\infty.
\]
Here, the partition of $\sup_{\delta}$ is taken with regard to $\delta=\left(t_{0},t_{1},\cdots,t_{n}\right)$, $0\leqq t_{0}<t_{1}<\cdots<t_{n}$.
Then, $\left\{ Y_{t}(Z=1),Y_{t}(Z=0)\right\} _{t\in\left[0,T\right]}$ is a special semimartingale regarding the filtration, $\left\{ \mathscr{F}_{t}\right\} $.
Thus,
\[
\begin{aligned}Y_{t} & =Y_{0}+V_{t}^{Y}+M_{t}^{Y}.\end{aligned}
\]
Here, $Y_{0}$ is a $\mathscr{F}_{0}$-measurable random variable, $M_{t}^{Y}$ is a square integrable martingale, and $V^{Y}$ is local integrable bounded variation and predictable.}
\end{thm}

\begin{proof}
A right continuous bounded mean variation stochastic process is called a quasimartingale.
From Rao-Stricker theorem \citep[][Chapter VI.41 Theorem 41.3]{Rogers_Williams_Vol2}, the necessary and sufficient condition for a stochastic process to be a quasimartingale is for the stochastic process to be expressed as the difference of positive right continuous supermartingales, $Z,\tilde{Z}$, $Z-\tilde{Z}$.
A right continuous supermartingale can be expressed using the Doob-Meyer decomposition,
\[
Z=Z_{0}-A^{Z}+M^{Z},
\]
of the square integrable martingale, $M^{Z}$, and predictable monotone increasing process, $A^{Z}$.
Equivalent for $\tilde{Z}$,
\[
Y=Z_{0}-\tilde{Z}_{0}+\left(A^{\tilde{Z}}-A^{Z}\right)+M^{Z}-M^{\tilde{Z}},
\]
where $M^{Z}-M^{\tilde{Z}}$ is a square-integrable martingale, and the difference in predictable monotone increasing processes, $A^{\tilde{Z}}-A^{Z}$, is locally integrable bounded variation and a predictable process (equivalent to expressing a bounded variation function as a difference of monotone increasing functions in calculus).
Therefore, this is a special semimartingale, and if we take the difference between two supermartingales, quasimartingales are all supersemimartingales.
\end{proof}

The result follows from the Doob-Meyer decomposition for square-integrable
submartingales, which states that any (square-integrable) submartingale can be uniquely written as a linear combination of a predictable
finite-variation term and a martingale innovation.  In our context, $V_t^Y$ represents the
drift, or systematic component, of the causal process, while $M_t^Y$ captures random
innovations not explained by the past.  This orthogonal split underlies the later use of
the DLM, which models the predictable drift linearly and the innovation through Gaussian
noise.

Similarly, we have the following theorem for the covariate process.
\begin{thm}\label{thm:semimartingale}
{\it Let $\left\{ \boldsymbol{X}_{t}\right\} _{t\in\left[0,T\right]}$ be a stochastic process with continuous paths, and is $\mathscr{F}_{t}^{X}$-adapted and the vector elements are square integrable,
\[
\sup_{t\in\left[0,T\right]}\mathbb{E}\left[X_{j,t}^{2}\right]<\infty,\ 1\leqq j\leqq d.
\]
Further, let the mean variation be bounded,
\[
\sup_{\delta}\mathbb{E}\left[\sum_{i=0}^{n-1}\left|\mathbb{E}\left[X_{j,t_{i+1}}-X_{j,t_{i}}\left|\mathscr{F}_{t_{i}}\right.\right]\right|+\left|X_{j,t_{n}}\right|\right]<\infty,\ 1\leqq j\leqq d.
\]
Here, the partition of $\sup_{\delta}$ is with regard to $\delta=\left(t_{0},t_{1},\cdots,t_{n}\right)$, $0\leqq t_{0}<t_{1}<\cdots<t_{n}$.
Then, $\left\{ \boldsymbol{X}_{t}\right\} _{t\in\left[0,T\right]}$ is a special semimartingale regarding the filtration, $\left\{ \mathscr{F}_{t}^{X}\right\} $.
Thus,
\[
\begin{aligned}\boldsymbol{X}_{t} & =\boldsymbol{X}_{0}+\boldsymbol{V}_{t}^{X}+\boldsymbol{M}_{t}^{X}.\end{aligned}
\]
Here, $\boldsymbol{X}_{t}=\left[X_{1,t},\cdots,X_{d,t}\right]^{\top}$, $\boldsymbol{X}_{0}$ is a $\mathscr{F}_{0}$-measurable random variable, $\boldsymbol{M}_{t}^{X}$ is a square integrable martingale, and $\boldsymbol{V}_{t}^{X}$ is local integrable bounded variation and predictable.}
\end{thm}
Proof is the same as Theorem~\ref{thm:martingale}.
Applying the same bounded-variation and square-integrability conditions to the covariate
process, $\{X_t\}$, ensures that it is also a special semimartingale.   Establishing this parallel structure for both $Y_t$ and
$X_t$ allows us to express their conditional expectations using Itô calculus in the next theorem.

$\mathbb{E}\left[\left.Y_{t}\right|Z=1,\mathscr{F}_{t}^{X}\right]$ and $\mathbb{E}\left[\left.Y_{t}\right|Z=0,\mathscr{F}_{t}^{X}\right]$ can be expressed using the It\^{o} process as the following.
\begin{thm}\label{thm:ito}
{\it Let the post-treatment outcome processes, $\left\{ Y_{t}\left(=ZY_{t}(Z=1)+\left(1-Z\right)Y_{t}(Z=0)\right)\right\} _{t\in\left[\treat,T\right]}$, regarding the filtration, $\left\{ \mathscr{F}_{t}\right\} $, be a special semimartingale, and $\left\{ \boldsymbol{X}_{t}\right\} _{t\in\left[0,T\right]}$ be a special semimartingale regarding the filtration, $\left\{ \mathscr{F}_{t}^{X}\right\} $.
Then, $\mathbb{E}\left[\left.Y_{t}\right|Z=1,\mathscr{F}_{t}^{X}\right]$ and $\mathbb{E}\left[\left.Y_{t}\right|Z=0,\mathscr{F}_{t}^{X}\right]$ can be expressed iteratively as,
\begin{equation}
\begin{aligned}\mathbb{E}\left[\left.Y_{t}\right|Z=1,\mathscr{F}_{t}^{X}\right] & =\mathbb{E}\left[\left.Y_{t-\varDelta t}\right|Z=1,\mathscr{F}_{t-\varDelta t}^{X}\right]+\sum_{j=1}^{d}\beta_{j,t}^{1}dX_{j,t},\ \left(\treat\leqq t\leqq T\right)\\
\textrm{with initial value} & \mathbb{E}\left[\left.Y_{\treat}\right|Z=1,\mathscr{F}_{\treat}^{X}\right]=\mathbb{E}\left[\left.Y_{\treat}\right|Z=0,\mathscr{F}_{\treat}^{X}\right]+\alpha_{\treat},\\
\mathbb{E}\left[\left.Y_{t}\right|Z=0,\mathscr{F}_{t}^{X}\right] & =\mathbb{E}\left[\left.Y_{t-\varDelta t}\right|Z=0,\mathscr{F}_{t-\varDelta t}^{X}\right]+\sum_{j=1}^{d}\beta_{j,t}^{0}dX_{j,t},,\ \left(0\leqq t\leqq T\right)\\
\textrm{with initial value} & \mathbb{E}\left[\left.Y_{0}\right|Z=0,\mathscr{F}_{0}^{X}\right]=Y_{0}.
\end{aligned}
\label{eq:TSC-VCM}
 \end{equation}
Here, $\left\{ \beta_{j,t}^{z}\right\} _{1\leqq j\leqq d}^{z=0,1}$ is $\mathscr{F}_{t}$-predictable and the integral, $\int_{0}^{t}\beta_{j,s}^{1}dX_{j,s}$, is understood to be an It\^{o} integral.
Further, $\alpha_{\treat}$ is the conditional expectation  $\mathbb{E}\left[\left.Y_{\treat}(Z=1)-Y_{\treat-dt}(Z=1)\right|\mathscr{F}_{\treat}^{X}\right]$ of the jump in the outcome due to treatment, $Y_{\treat}(Z=1)-Y_{\treat-dt}(Z=1)$ (i.e., impulse).}
\end{thm}

(For $t<\treat$ and $Z=0$.) From the FS decomposition starting at $t=0$, we have 
\[
\begin{aligned}Y_{t}^{z} & =Y_{0}+\sum_{j=1}^{d}\int_{0}^{t}\beta_{j,s}^{z}dX_{j,s}+L_{t}^{z},\ \left\langle M_{j}^{X},L^{z}\right\rangle =0,\ \left(z=0,1\right),\left(1\leqq j\leqq d\right),\\
 & \left\{ L_{t}^{z}\right\} :\mathscr{F}_{t}\textrm{-martingale with }L_{0}^{z}=0.
\end{aligned}
\]
Since, $Y_{t}=Y_{t}(Z=0)$, we have
\[
Y_{t}=Y_{0}+\sum_{j=1}^{d}\int_{0}^{t}\beta_{j,s}^{0}dX_{j,s}+L_{t}^{0}.
\]
Right before $t$, $t-dt$, we have
\[
Y_{t-dt}=Y_{0}+\sum_{j=1}^{d}\int_{0}^{t-dt}\beta_{j,s}^{0}dX_{j,s}+L_{t-dt}^{0},
\]
thus,
\[
Y_{t}-Y_{t-dt}=\sum_{j=1}^{d}\beta_{j,t}^{0}dX_{j,t}+dL_{t}^{0}.
\]
If we take the expectation, $\mathbb{E}\left[\left.\cdot\right|Z=0,\mathscr{F}_{t}^{X}\right]$, on both sides, we have
\[
\mathbb{E}\left[\left.Y_{t}-Y_{t-dt}\right|Z=0,\mathscr{F}_{t}^{X}\right]=\sum_{j=1}^{d}\beta_{j,t}^{0}dX_{j,t}+\mathbb{E}\left[\left.dL_{t}^{0}\right|Z=0,\mathscr{F}_{t}^{X}\right].
\]
Since we are assuming no spillover, we have
\[
\mathbb{E}\left[\left.Y_{t-\varDelta t}\right|Z=0,\mathscr{F}_{t}^{X}\right]=\mathbb{E}\left[\left.Y_{t-\varDelta t}\right|Z=0,\mathscr{F}_{t-\varDelta t}^{X}\right].
\]
This leaves us to show, $\mathbb{E}\left[\left.dL_{t}^{0}\right|Z=0,\mathscr{F}_{t}^{X}\right]=0$, which we will show at the end.

\begin{proof}
    (For $t<\treat$ and $Z=0$.) From the FS decomposition starting at $t=0$, we have 
\[
\begin{aligned}Y_{t}^{z} & =Y_{0}+\sum_{j=1}^{d}\int_{0}^{t}\beta_{j,s}^{z}dX_{j,s}+L_{t}^{z},\ \left\langle M_{j}^{X},L^{z}\right\rangle =0,\ \left(z=0,1\right),\left(1\leqq j\leqq d\right),\\
 & \left\{ L_{t}^{z}\right\} :\mathscr{F}_{t}\textrm{-martingale with }L_{0}^{z}=0.
\end{aligned}
\]
Since, $Y_{t}=Y_{t}(Z=0)$, we have
\[
Y_{t}=Y_{0}+\sum_{j=1}^{d}\int_{0}^{t}\beta_{j,s}^{0}dX_{j,s}+L_{t}^{0}.
\]
Right before $t$, $t-dt$, we have
\[
Y_{t-dt}=Y_{0}+\sum_{j=1}^{d}\int_{0}^{t-dt}\beta_{j,s}^{0}dX_{j,s}+L_{t-dt}^{0},
\]
thus,
\[
Y_{t}-Y_{t-dt}=\sum_{j=1}^{d}\beta_{j,t}^{0}dX_{j,t}+dL_{t}^{0}.
\]
If we take the expectation, $\mathbb{E}\left[\left.\cdot\right|Z=0,\mathscr{F}_{t}^{X}\right]$, on both sides, we have
\[
\mathbb{E}\left[\left.Y_{t}-Y_{t-dt}\right|Z=0,\mathscr{F}_{t}^{X}\right]=\sum_{j=1}^{d}\beta_{j,t}^{0}dX_{j,t}+\mathbb{E}\left[\left.dL_{t}^{0}\right|Z=0,\mathscr{F}_{t}^{X}\right].
\]
Since we are assuming no spillover, we have
\[
\mathbb{E}\left[\left.Y_{t-\varDelta t}\right|Z=0,\mathscr{F}_{t}^{X}\right]=\mathbb{E}\left[\left.Y_{t-\varDelta t}\right|Z=0,\mathscr{F}_{t-\varDelta t}^{X}\right].
\]
This leaves us to show, $\mathbb{E}\left[\left.dL_{t}^{0}\right|Z=0,\mathscr{F}_{t}^{X}\right]=0$, which we will show at the end.

(For $t\geqq\treat$.) From the FS decomposition starting at $\treat$, we have
\[
\begin{aligned}Y_{t}^{z} & =Y_{\treat}+k\alpha_{\treat}+\sum_{j=1}^{d}\int_{\treat}^{t}\beta_{j,s}^{z}dX_{j,s}+L_{t}^{z},\ \left\langle M_{j}^{X},L^{z}\right\rangle =0,\ \left(z=0,1\right),\left(1\leqq j\leqq d\right),\\
 & \left\{ L_{t}^{z}\right\} :\mathscr{F}_{t}\textrm{-martingale with }L_{\treat}^{z}=0.
\end{aligned}
\]
Since $Y_{t}=ZY_{t}(Z=1)+\left(1-Z\right)Y_{t}(Z=0)$, we have
\[
Y_{t}=Y_{0}+Z\left(\sum_{j=1}^{d}\int_{0}^{t}\beta_{j,s}^{1}dX_{j,s}+L_{t}^{1}\right)+\left(1-Z\right)\left(\sum_{j=1}^{d}\int_{0}^{t}\beta_{j,s}^{0}dX_{j,s}+L_{t}^{0}\right).
\]
Right before $t$, $t-dt$, we have
\[
Y_{t-dt}=Y_{0}+Z\left(\sum_{j=1}^{d}\int_{0}^{t-dt}\beta_{j,s}^{1}dX_{j,s}+L_{t-dt}^{1}\right)+\left(1-Z\right)\left(\sum_{j=1}^{d}\int_{0}^{t-dt}\beta_{j,s}^{0}dX_{j,s}+L_{t-dt}^{0}\right),
\]
thus,
\[
Y_{t}-Y_{t-dt}=Z\left(\sum_{j=1}^{d}\beta_{j,t}^{1}dX_{j,t}+dL_{t}^{1}\right)+\left(1-Z\right)\left(\sum_{j=1}^{d}\beta_{j,t}^{0}dX_{j,t}+dL_{t}^{0}\right).
\]
If we take the expectation, $\mathbb{E}\left[\left.\cdot\right|Z=z,\mathscr{F}_{t}^{X}\right]$ on both sides, we have
\[
\mathbb{E}\left[\left.Y_{t}-Y_{t-dt}\right|Z=z,\mathscr{F}_{t}^{X}\right]=z\left(\sum_{j=1}^{d}\beta_{j,t}^{1}dX_{j,t}\right)+\left(1-z\right)\left(\sum_{j=1}^{d}\beta_{j,t}^{0}dX_{j,t}\right)+\mathbb{E}\left[\left.dL_{t}\right|Z=z,\mathscr{F}_{t}^{X}\right].
\]
Here, we let $dL_{t}=ZdL_{t}^{1}+\left(1-Z\right)dL_{t}^{0}$.
Since we assume no feedback to confounders, we have
\[
\mathbb{E}\left[\left.Y_{t-\varDelta t}\right|Z=1,\mathscr{F}_{t}^{X}\right]=\mathbb{E}\left[\left.Y_{t-\varDelta t}\right|Z=1,\mathscr{F}_{t-\varDelta t}^{X}\right].
\]
This leaves us to show $\mathbb{E}\left[\left.dL_{t}\right|\mathscr{F}_{t}^{X}\right]=0$.
$L_{t}^{z}$ is a $\mathscr{F}_{t}$-martingale and $\mathscr{F}_{t}^{X}\subset\mathscr{F}_{t}$, thus we have $\mathbb{E}\left[\left.dL_{t}\right|Z=z,\mathscr{F}_{t}\right]=0$.
From the law of iterations, we have
\[
\mathbb{E}\left[\left.dL_{t}\right|Z=z,\mathscr{F}_{t}^{X}\right]=\mathbb{E}\left[\left.\mathbb{E}\left[\left.dL_{t}\right|Z=z,\mathscr{F}_{t}\right]\right|Z=z,\mathscr{F}_{t}^{X}\right]=0.
\]
The above result $\mathbb{E}\left[\left.dL_{t}^{0}\right|Z=0,\mathscr{F}_{t}^{X}\right]=0$ can be shown in a similar way.
\end{proof}

Given the semimartingale structure of $\{Y_t\}$ and $\{X_t\}$, the conditional expectation
$E[Y_t| Z,\mathcal F^X_t]$ evolves as an Itô process driven by increments of the
covariate process.  The coefficients, $\{\beta^{z}_{j,t}\}$, are predictable because they depend only
on past information, and the jump $\alpha_{\treat}$ represents the instantaneous treatment effect at
the moment of intervention.  The iterative form in~\eqno{TSC-VCM} is obtained by applying the stochastic
integration-by-parts identity to $E[Y_t|\mathcal F^X_t]$.
Intuitively, for something to be $\mathscr{F}_{t}$-predictable is for the value to be known right before $t$.
Further, the impulse, $\alpha_{\treat}$, right after treatment is estimated using the regression model,
\[
Y_{\treat}(Z=1)=\mathbb{E}\left[\left.Y_{t}\right|Z=0,\mathscr{F}_{t}^{X}\right]+\alpha_{\treat}.
\]

For discrete observations, we have the following expression.
\begin{thm}\label{thm:smapprox}
{\it Let $\left\{ Y_{t}\left(=ZY_{t}(Z=1)+\left(1-Z\right)Y_{t}(Z=0)\right)\right\} _{t\in\left[0,T\right]}$ and $\left\{ \boldsymbol{X}_{t}\right\} _{t\in\left[0,T\right]}$ be square integrable discrete time stochastic processes.
Then, $\mathbb{E}\left[\left.Y_{t}\right|Z=1,\mathscr{F}_{t}^{X}\right]$ and $\mathbb{E}\left[\left.Y_{t}\right|Z=0,\mathscr{F}_{t}^{X}\right]$ can be expressed as
\begin{equation}
\begin{aligned}\mathbb{E}\left[\left.Y_{t}\right|Z=1,\mathscr{F}_{t}^{X}\right] & =\mathbb{E}\left[\left.Y_{t-1}\right|Z=1,\mathscr{F}_{t-1}^{X}\right]+\sum_{j=1}^{d}\beta_{j,t}^{1}\left(X_{j,t}-X_{j,t-1}\right),\ \left(\treat\leqq t\leqq T\right),\\
\textrm{with initial value } & \mathbb{E}\left[\left.Y_{\treat}\right|Z=1,\mathscr{F}_{\treat}^{X}\right]=\mathbb{E}\left[\left.Y_{\treat}\right|Z=0,\mathscr{F}_{\treat}^{X}\right]+\alpha_{\treat},\\
\mathbb{E}\left[\left.Y_{t}\right|Z=0,\mathscr{F}_{t}^{X}\right] & =\mathbb{E}\left[\left.Y_{t-1}\right|Z=0,\mathscr{F}_{t-1}^{X}\right]+\sum_{j=1}^{d}\beta_{j,t}^{0}\left(X_{j,t}-X_{j,t-1}\right).\ \left(0\leqq t\leqq T\right),\\
\textrm{with initial value } & \mathbb{E}\left[\left.Y_{0}\right|Z=0,\mathscr{F}_{0}^{X}\right]=Y_{0}.
\end{aligned}
\label{smeq:TSC-VCM-1}
 \end{equation}
Here, $\left\{ \beta_{j,t}^{z}\right\} _{1\leqq j\leqq d}^{z=0,1}$ is a $\mathscr{F}_{t-1}$-measurable stochastic process.}
\end{thm}
Proof is the same as Theorem~\ref{thm:ito}.
In discrete time, the Itô integral reduces to a finite difference, producing the recursive
representation in~\ref{smeq:TSC-VCM-1}.  Each step updates the expected outcome by the lagged mean plus a
weighted change in the covariates.  This shows that a time-varying linear model for the
increments of $X_t$ is sufficient to approximate the conditional expectation of $Y_t$,
justifying the use of dynamic regression models for DATE estimation.

The above result shows that it is sufficient to use a time-varying linear model for the increment of $X_{t}$ to obtain $\mathbb{E}\left[\left.Y_{t}\right|Z=1,\mathscr{F}_{t}^{X}\right]$.
To complete the specification, it is also necessary to estimate $\beta$.
From \ref{smeq:TSC-VCM-1}, $\left\{ \beta_{j,t}^{z}\right\} _{1\leqq j\leqq d}^{z=0,1}$ is a $\mathscr{F}_{t-1}$-measurable stochastic process.
This means that $\beta$ is a non-linear measurable function that depends on the variables, $\left(\left\{ Y_{s}(z)\right\} _{0\leqq s\leqq t-1},\left\{ \boldsymbol{X}_{s}\right\} _{0\leqq s\leqq t-1}\right)$.
We will first consider how to estimate this in discrete time, then remark on continuous time.

Since $\boldsymbol{\beta}_{t}^{z}=\left[\beta_{1,t}^{z},\cdots,\beta_{d,t}^{z}\right]^{\top},\left(z=0,1\right)$ is a $\mathscr{F}_{t-1}$-measurable stochastic process, we can Doob-Meyer decompose it to its martingale component and trend component (known at $t-2$):
\begin{alignat*}{1}
\boldsymbol{\beta}_{t}^{z} & =\boldsymbol{\beta}_{0}^{z}+\boldsymbol{M}_{t}^{z}+\boldsymbol{A}_{t}^{z}\\
\boldsymbol{\beta}_{t-1}^{z} & =\boldsymbol{\beta}_{0}^{z}+\boldsymbol{M}_{t-1}^{z}+\boldsymbol{A}_{t-1}^{z}.
 \end{alignat*}
Here, let $\boldsymbol{M}_{t}^{z}$ be the martingale component and $\boldsymbol{A}_{t}^{z}$ be the trend component.
Subtracting from both sides, we have
\[
\boldsymbol{\beta}_{t}^{z}-\boldsymbol{\beta}_{t-1}^{z}=\varDelta\boldsymbol{M}_{t-1}^{z}+\varDelta\boldsymbol{A}_{t-1}^{z},
\]
where the $\mathscr{F}_{t-1}$-conditional expectation of the difference of the martingale component, $\varDelta\boldsymbol{M}_{t-1}^{z}$, is zero: $\mathbb{E}\left[\left.\varDelta\boldsymbol{M}_{t-1}^{z}\right|\mathscr{F}_{t-1}\right]=0$.
Using some time-varying matrix, $G_{t}$, we can write $G_{t}\boldsymbol{\beta}_{t-1}^{z}=\boldsymbol{\beta}_{t-1}^{z}+\varDelta\boldsymbol{A}_{t-1}^{z}$, so the process of $\boldsymbol{\beta}^{z}$ is
\[
\boldsymbol{\beta}_{t}^{z}=G_{t}\boldsymbol{\beta}_{t-1}^{z}+\boldsymbol{\omega}_{t},\ \mathbb{E}\left[\boldsymbol{\omega}_{t}\right]=0.
\]
Note that the difference in the martingale component, $\varDelta\boldsymbol{M}_{t-1}^{z}$, is the transition innovation, $\boldsymbol{\omega}_{t}$.
Further, if we consider that we are discretely observing continuous time, we shown that $\boldsymbol{\omega}_{t}$ follows a Gaussian distribution with time-varying covariance matrix, under some assumptions:
\begin{lem}
\label{lem:TimeSeriesCausal_Normal}
{\it Let $\boldsymbol{\beta}_{t}^{z}=\left[\beta_{1,t}^{z},\cdots,\beta_{d,t}^{z}\right]^{\top},\left(z=0,1\right)$ be a stochastic process defined on continuous time.
When $\left\{ \boldsymbol{\beta}_{t}^{z}\right\} $ is $\mathscr{F}_{t}$-predictable, square-integrable, and continuous path, it is written as $\mathbb{E}\left[M_{i,t}^{z}M_{j,t}^{z}\right]=\int_{0}^{t}\left(\sigma_{s}^{z}\right)^{2}W_{ij,s}^{z}ds$
using $\sigma_{t}$, which is a function of the predictable, time-varying martingale component, $\boldsymbol{M}_{t}^{z}=\left[M_{1,t}^{z},\cdots,M_{d,t}^{z}\right]^{\top}$, of the Doob-Meyer decomposition, $\boldsymbol{\beta}_{t}^{z}=\boldsymbol{\beta}_{0}^{z}+\boldsymbol{M}_{t}^{z}+\boldsymbol{A}_{t}^{z}$,
and the matrix-valued function, $\boldsymbol{W}_{t}=\left(W_{ij,t}\right)_{1\leqq i,j\leqq d}$.
Here, the increment of the martingale component, $d\boldsymbol{M}_{t}^{z}\left(=\boldsymbol{M}_{t+\varDelta t}^{z}-\boldsymbol{M}_{t}^{z}\right)$, follows a Gaussian distribution, $N\left(0,\varDelta t\sigma_{t}^{2}W_{t}\right)$.}
\end{lem}

\begin{proof}
We use the following theorem that expresses continuous martingales with Brownian motion:
\begin{thm}
\cite{Ikeda-Watanabe_89} Chapter2. Theorem 7.1. 
Let $M_{i}\in\mathscr{M}_{2},i=1,2,\cdots,d$.
Further, assume
\begin{alignat*}{1}
\left\langle M_{i},M_{j}\right\rangle \left(t\right) & =\int_{0}^{t}\left(\sigma_{s}^{z}\right)^{2}W_{ij,s}^{z}ds\ i,j=1,2,\cdots,d\\
\det\left\{ \sigma_{s}\sqrt{W_{s}}\right\}  & \neq0\ ^{\forall}s,
 \end{alignat*}
holds.
Then, there exists a $d$-dimensional $\mathscr{F}_{t}$-Brownian motion, $\boldsymbol{B}_{t}=\left(B_{1,t},\cdots,B_{d,t}\right)$
$\boldsymbol{B}_{0}=0$, and we have
\[
M_{i,t}=\sum_{j=1}^{d}\int_{0}^{t}\sigma_{s}\sqrt{W_{ij,s}}dB_{j,s}.
\]
\end{thm}

Since the increment of the Brownian motion, $d\boldsymbol{B}_{s}$, follows a Gaussian distribution, $N\left(\boldsymbol{0},\varDelta s\boldsymbol{I}\right)$, we have $d\boldsymbol{M}_{t}^{z}=\sigma_{t}^{z2}W_{t}^{z}d\boldsymbol{B}_{t}$ and
\[
d\boldsymbol{M}_{t}^{z}\sim N\left(\boldsymbol{0},\varDelta t\sigma_{t}^{2}W_{t}\right).
\]
\end{proof}

Decomposing the coefficient process $\beta_t^{z}$ into its martingale and bounded-variation
components implies that the martingale increment $dM_t^{z}$ has zero mean and variance
proportional to $dt$.  Under continuous-path and square-integrability assumptions, this
increment is normally distributed with covariance $\sigma_t^2 W_t\,dt$, producing the
Gaussian transition law stated in the lemma.

\begin{thm}\label{thm:dlm}
{\it Assume all the assumptions in Lemma~\ref{lem:TimeSeriesCausal_Normal} hold.
The time series model for discrete observations of the continuous time stochastic process, $\left\{ \boldsymbol{\beta}_{t}^{z}\right\} $, is given as
\begin{equation}
\boldsymbol{\beta}_{t}^{z}=G_{t}\boldsymbol{\beta}_{t-1}^{z}+\boldsymbol{\omega}_{t},\ \boldsymbol{\omega}_{t}\sim N\left(0,\sigma_{t}^{2}\boldsymbol{W}_{t}\right),\ \ \ \left(t=1,2,\cdots,T\right),\left(z=0,1\right).\label{smeq:TimeSeriesCausal_DLMtransition}
 \end{equation}}
\end{thm}
Proof is shown above.
Equation~\ref{smeq:TimeSeriesCausal_DLMtransition} is the discrete-time representation of the continuous-time stochastic dynamics
for $\beta_t^{z}$.  The state evolution $\beta_t^{z} = G_t\beta_{t-1}^{z}+\omega_t$ captures slow
parameter drift through $G_t$ and innovation noise $\omega_t$.  Substituting this into the
observation equation for $Y_t$ yields the DLM form. 

\section{Estimation of DATE via DLM}\label{sm:estimation}

This appendix describes the estimation procedure used in the One-None case (or more broadly when the number of units are small).
We fit a single Dynamic Linear Model (DLM) with discount stochastic volatility
to the full observed series $\{Y_1,\dots,Y_T\}$, where the intervention at $t_c$
is encoded in the design matrices for $t \ge t_c$. Posterior inference on the
latent states and evolution variances is obtained via forward filtering backward
sampling (FFBS). At the intervention time $t_c$, we then branch the model
evolution into treated and control regimes and project both forward; their
difference defines the DATE trajectory.

\paragraph{Model specification.}
We assume the outcome admits the DLM representation
\begin{subequations}
 
\begin{align}
y_{t}^{z} & =y_{t-1}^{z}+\boldsymbol{F}_{t}^{\top}\boldsymbol{\theta}_{t}+\varepsilon_{t},\quad\varepsilon_{t}\sim N(0,\sigma_{t}^{2}),\\
\boldsymbol{\theta}_{t} & =\boldsymbol{G}_{t}\boldsymbol{\theta}_{t-1}+\boldsymbol{\omega}_{t},\quad\boldsymbol{\omega}_{t}\sim N\left(\left[\begin{array}{c}
0\\
\boldsymbol{0}
\end{array}\right],\left[\begin{array}{cc}
\sigma_{t}^{2} & \boldsymbol{0}^{\top}\\
\boldsymbol{0} & \sigma_{t}^{2}\W_{t}
\end{array}\right]\right),\\
\boldsymbol{F}_{t} & =\left[\begin{array}{c}
1\\
\boldsymbol{x}_{t}
\end{array}\right],\ \boldsymbol{\theta}_{t}=\left[\begin{array}{c}
\beta_{0,t}^{z}\\
\boldsymbol{\beta}_{t}^{z}
\end{array}\right]\\
\X & =[\x_{1},...,\x_{\treat},...,\x_{T}]\\
 & =\begin{bmatrix}y_{0} & ... & y_{\treat-1} & y_{\treat} & y_{\treat+1} & ... & y_{T-1}\\
0 & ... & 0 & \mathbbm{1}_{Z_{i}=1} & 0 & ... & 0\\
0 & ... & 0 & \mathbbm{1}_{Z_{i}=1} & \mathbbm{1}_{Z_{i}=1} & ... & \mathbbm{1}_{Z_{i}=1}\\
0 & ... & 0 & \mathbbm{1}_{Z_{i}=1} & 2\mathbbm{1}_{Z_{i}=1} & ... & (T-{\treat}+1)\mathbbm{1}_{Z_{i}=1}
\end{bmatrix},
 \end{align}
\end{subequations}
with $\sigma^2_t$ and $\W_t$ governed by discount stochastic volatility. The
intervention at $t_c$ is encoded through an intervention component in $F_t$ and/or
$G_t$ that is active for $t \ge t_c$ (e.g., a spot effect, level effect, or slope  effect). This single model is fitted to the entire series
$1{:}T$.

\paragraph{FFBS on the full series.}
Using the full data $\{Y_1,\dots,Y_T\}$ and the discount specification, we perform forward filtering followed by backward sampling to
obtain posterior draws
\[
    \{\boldsymbol{\theta}_t^{(s)}\}_{t=1}^T,\quad
    \{\W_t^{(s)},\sigma_t^{2(s)}\}_{t=1}^T,\qquad s=1,\dots,S.
\]
In particular, this yields a posterior sample of the state at the intervention:
\[
    \boldsymbol{\theta}_{t_c}^{(s)} \sim p(\boldsymbol{\theta}_{t_c} | Y_{1:T}), \qquad s=1,\dots,S,
\]
together with the corresponding evolution variance sequences implied by the
discount specification.

\paragraph{Branching at the intervention time.}
For each posterior draw $s$, we then construct two post-intervention trajectories
by forward simulation from $\boldsymbol{\theta}_{t_c}^{(s)}$:

\begin{itemize}
\item \textbf{Treated path.}  
  Using the fitted intervention-augmented system matrices
  $\{G_t^{(s)},F_t^{(s)},W_t^{(s)},\sigma_t^{2(s)}\}_{t>t_c}$, we propagate
  $\boldsymbol{\theta}_{t_c}^{(s)}$ forward for $h$ steps to obtain
  \[
      \boldsymbol{\theta}_{t_c+1:t_c+h}^{(s)}, \quad
      Y_{t_c+1:t_c+h}^{(s)}(1).
  \]

\item \textbf{Control path.}  
  Starting from the same $\boldsymbol{\theta}_{t_c}^{(s)}$, we propagate under a
  counterfactual system in which the intervention component is turned off for
  $t>t_c$ (i.e., by removing the intervention covariates or setting its effect to
  zero), using the same posterior draw of the remaining parameters and
  variances. This counterfactual $\boldsymbol{\theta}_{t_c}^{(s)}$, where the design matrix is turned off, is denoted as $\boldsymbol{\theta}_{t_c+1:t_c+h}^{*(s)}$. This yields
  \[
      \boldsymbol{\theta}_{t_c+1:t_c+h}^{*(s)}, \quad
      Y_{t_c+1:t_c+h}^{(s)}(0).
  \]
\end{itemize}

Crucially, both forward projections condition on the same posterior information
up to $t_c$ (which has been estimated using the full series).

\paragraph{Estimated DATE.}
For each posterior draw $s$ and horizon $h \ge 0$, we define
\[
    \mathrm{DATE}^{(s)}(h)
        =
        Y_{t_c+h}^{(s)}(1) - Y_{t_c+h}^{(s)}(0).
\]
Posterior means and credible intervals for the DATE trajectory are obtained from
the empirical distribution of $\{\mathrm{DATE}^{(s)}(h)\}_{s=1}^S$.

\section{Additional discussions on estimating the DATE}\label{sm: date}
\subsection{Discussion: Retrospective and forward-looking causality using DATE}
In this paper, we only consider the causal effect in terms of retrospective analysis; defining and estimating the effect of the treatment after the fact.
This perspective is relevant to the contexts that motivate our development.
For example, most analyses of policy interventions (e.g., COVID lockdowns) are done after some time has passed from the intervention, and the interest is with regard to what the effect {\it was}.
This is useful in many contexts because 1) it evaluates the intervention in terms of whether it was effective or not, and 2) it is informative regarding what the effect might be if it were implemented in the future or in different units.
Returning to the COVID example, a retrospective causal analysis of the effect of lockdowns can be used to evaluate the efficacy of the lockdown (whether it was effective) and inform the relevant decision makers on whether to implement them in other locations.

Another perspective, though not explicitly considered in this paper, is forward-looking, predictive causality.
In this context, what is relevant is how the estimated treatment effect affects the decision making moving forward, post-intervention.
For example, consider a marketer interested in the effect of advertisements on search queries.
While the retrospective perspective is important, what is more relevant to the marketer is how the intervention affects the outcomes in the future, in order to inform whether they should continue/discontinue the advertisement.
This is the context of \cite{KevinLiEtAlCausalMVTS2024}, where they develop what they call outcome adaptive modeling (OAM: Section 3.4) for causal predictions in the setting where there is only one treated unit.
Consistent with our theoretical results, they apply a (multivariate) dynamic linear model to generate forecasts of the treated and counterfactual control series, the difference in the mean being the forward-looking DATE.

The retrospective and forward-looking perspectives complement each other in a holistic view of causal inference and its interface with decision making.
Our proposed foundations, and theoretical results provide justification for both perspectives.
While we only consider retrospective analysis, the forward-looking analysis is trivially done by extending the estimation strategies proposed in this paper \citep[which is consistent with][albeit with more elaborate models]{KevinLiEtAlCausalMVTS2024}.
In practice, both should be done to evaluate, monitor, and intervene in real-time for improved sequential decision making.

\subsection{Discussion: Dynamic outcome treatment effect}
The goal of DATE is to estimate the underlying average effect of a treatment, which is why, even in the one treated, one control case, we estimate the mean path of the treated and control.
In some contexts, the estimand is the effect of the observed outcomes, rather than the underlying mean processes.
For example, when there is only one series that gets treated (i.e., interrupted time series), some contexts warrant the estimand to be the difference between the observed treated and the estimated mean of the control.
To delineate this from DATE, we call this the dynamic outcome treatment effect (DOTE), where the average is replaced by the observed, path-wise outcome.
Although the focus of this paper is on DATE, the theoretical results above and the simulation results below are consistent when the estimand is DOTE.
Specific characteristics of DOTE and its estimation strategies will be left for another paper.

\section{Additional discussion on assumptions}\label{sm:assumption}

\subsection{Dynamic ignorability and positivity}
Dynamic ignorability and positivity are conceptually simple but substantively
strong assumptions. They require that, conditional on the pre-treatment
confounder history $\mathscr{F}^X_{t_c}$, the assignment of the intervention at
$t_c$ is as good as random with respect to the post-treatment potential outcome
processes, $\{Y_t(1), Y_t(0)\}_{t > t_c}$. Here, we briefly discuss settings in
which these assumptions are reasonable, situations in which they are likely to
fail, and practical diagnostics for assessing plausibility.

\paragraph{Plausible settings.}
Dynamic ignorability is most credible when confounders are genuinely
exogenous and not affected by latent behavioral or economic responses.
Examples include interventions triggered by external shocks, such as extreme
weather, natural disasters, supply-chain disruptions, commodity price spikes, or
regulatory changes determined by exogenous thresholds. In these settings,
$\mathscr{F}^X_{t_c}$ captures the information available to the decision maker,
and the treatment decision is not based on unobserved forecasts of the outcome
process itself.

\paragraph{Common violations.}
Many economic, epidemiological, and policy interventions adjust dynamically
to the evolving state of the system. Mobility restrictions respond to infection
rates; fiscal policy responds to expected unemployment; individuals change their
behavior in anticipation of policy changes. In such settings, variables in
$\mathscr{F}^X_{t_c}$ may themselves be influenced by latent features of the
outcome process, and ignorability may be violated. Positivity can also fail when
certain histories deterministically lead to treatment or non-treatment, making
counterfactual contrasts unidentified.

\paragraph{Consequences and alternatives.}
When ignorability is questionable, identifying the DATE requires either
design-based assumptions (e.g., discontinuities, staggered adoption with
plausibly exogenous timing) or explicit structural modeling of the decision
process. Alternative causal frameworks include g-computation, marginal
structural models, structural time series models, and fully specified
state-space representations, in which the treatment assignment mechanism is
modeled jointly with the outcome process. These approaches impose stronger
modeling assumptions, but allow identification under weaker ignorability
conditions.
This is consistent with existing, non-time series causal inference approaches.

\paragraph{Practical diagnostics.}
Although ignorability is inherently untestable, researchers can evaluate its
plausibility through (i) domain knowledge regarding the decision process,
(ii) balance checks on pre-treatment trends and covariates, (iii) placebo tests
using pre-intervention pseudo-treatments, and (iv) sensitivity analyses that
quantify how large a violation of ignorability would be required to overturn the
estimated DATE. These diagnostics do not guarantee validity but help assess how
strongly conclusions depend on the identifying assumptions.

\subsection{Feedback Between Outcomes and Confounders}

Our identification strategy conditions on the pre-treatment confounder history
$\mathcal F^X_{t_c}$, which is assumed to include all covariates that jointly
influence treatment assignment at $t_c$ and the post-treatment potential outcome
processes, $\{Y_t(1), Y_t(0)\}_{t>t_c}$. This formulation does not require
structural restrictions on the evolution of the outcome process prior to $t_c$,
but it implicitly assumes the absence of feedback from latent components of
$Y_t$ to the confounder process $X_t$ in the periods immediately preceding
treatment.

Such feedback is common in economic and epidemiological systems: expectations
about future outcomes can influence covariates, such as mobility, investment,
or precautionary behavior. When this occurs, the observed confounder history
$\mathcal F^X_{t_c}$ may itself be a function of unobserved features of the
potential outcomes, in which case dynamic ignorability may fail.

When feedback is a concern, alternative strategies are available. These include
explicit structural modeling of the joint $(X_t, Y_t)$ system, sequential g-computation,
marginal structural models with treatment-dependent covariates, or state-space
representations, in which both the confounder and outcome processes share latent
dynamic components.

In practice, the plausibility of the no-feedback condition is best assessed using
domain knowledge, examination of the temporal ordering of confounders, and
sensitivity analyses that quantify the robustness of the DATE to violations of
conditional independence.

\subsection*{Stationarity and stability considerations}

Our theoretical development does not require the outcome process to be
stationary, and the DATE remains well-defined for any square-integrable
stochastic process with a transition law before and after the intervention.
Nonstationarity (e.g., trends or seasonality) is accommodated because the DATE
contrasts the laws of the treated and control processes starting from a common
initial condition $Y_{t_c}$. The intervention may itself induce changes in
stability or persistence; such changes are part of the causal effect rather than
a violation of assumptions. When researchers expect strong nonstationary
components, these may be included in $\mathcal F^X_{t_c}$ or modeled
explicitly, without altering identification.

\subsection*{Extensions to multiple or staggered interventions}

The paper focuses on a single intervention at time $t_c$, but the framework
extends naturally to settings with multiple or staggered interventions. For
multiple interventions, one may index treatment episodes by their initiation
times and define process-level potential outcomes for each regime. For staggered
interventions across units, the DATE can be defined relative to each unit’s
treatment time and aggregated across cohorts. Identification requires the same
dynamic ignorability conditions applied at each intervention point, and the
estimation strategies (DIPW or DLM-based representation) may be adapted without
changing the core theory. We leave detailed development to future work.

\section{Additional material for simulation study}\label{sm: sim}

\subsection{Models used in simulation}\label{sm: model desc}
\begin{itemize}
\itemsep-0.5em
\item[-] 
{\bf Dynamic linear model} \citep[DLM:][]{WestHarrison1997book2,Prado2010}: Utilizing the model specified in \eqno{DLM}, the DATE is estimated with priors $\btheta_{0}|v_{0}\sim N(\m_{0},(v_{0}/s_{0})\C_0)$ with $\m_{0}=(0.95,\zero)'$, $\C_{0}=I$,
and $1/v_{0}\sim G(n_{0}/2,n_{0}s_{0}/2)$ with $n_{0}=20,s_{0}=0.01$. 
The discount factors for stochastic volatility, $(\beta,\delta)$, are estimated using a grid search over $[0.95, 0.99, 0.999]$.
The parameters are estimated using the forward filter backward sampling algorithm \citep[FFBS:][]{Schnatter1994}.
Once the parameters are learnt, the mean processes are generated from the posterior parameters.

\item[-]
{\bf Linear model} (LM): Fitting a simple linear regression model with the indicator functions specified in \eqno{DLMe}.
Once the parameters are estimated, the mean processes are generated from the estimated model.

\item[-]
{\bf Linear model with AR(1) term} (LM-AR(1)): Fitting a simple linear regression model with the indicator functions and AR(1) term specified in \eqno{DLMe}.
Once the parameters are estimated, the mean processes are generated from the estimated model.

\item[-]
{\bf Autoregressive integrated moving average with exogenous variable} (ARIMAX): Fitting an ARIMA(1,1,1) model with indicator functions specified in \eqno{DLMe}.
Once the parameters are estimated, the mean processes are generated from the estimated model.

\item[-]
{\bf Observed treated/control series} (Y): Instead of fitting a model for the mean path, utilize the observed series as the expected mean path.
For One-Many, the treatment effect is calculated as the difference between the observed treated series and the mean of the control series.
In the case of One-One, the treatment effect is the difference between the observed treated and control series.
For One-None, the control path is estimated using LM above, and the treatment effect is calculated as the difference between the observed treated series and the estimated control mean path.
This is equivalent to ITS with linear regression.

\end{itemize}

Apart from the above estimation strategies, we also consider specific estimation strategies for specific scenarios.

\begin{itemize}
\itemsep-0.5em
\item [-] {\bf Synthetic control} (SCM): For the One-Many case, we apply the synthetic control method of \cite{abadie2010synthetic}.
Estimation is standard, fitting a weighted average of many control series to the treated series, pre-treatment, and estimate of the synthetic control. The treatment effect is estimated by taking the difference between the treated series and the synthetic control, post-treatment.

\item [-] {\bf Difference-in-differences} (DiD): For the One-One case, a basic DiD estimator is estimated using the difference in pre-treatment averages and the differences in post-treatment series, per time period.

\item [-]
{\bf Causal-ARIMA} (C-ARIMA): For the One-None case, we fit the C-ARIMA model in \cite{menchetti2021estimating}, using the R code they provide. 
The model generates sample paths of the counterfactual control, from which the difference between the observed treated series is the treatment effect.
\end{itemize}

\

\section{Placebo intervention test for robustness}\label{sm: robust}

We propose a robustness check for dynamic causal analyses to verify that the
estimation procedure does not uncover spurious effects when treatment is assigned at
incorrect times. We implement a placebo intervention test by randomly selecting an
intervention time that occurs strictly before the true intervention date.
All estimation steps-- fitting the dynamic linear model, constructing the treated and
control trajectories, and computing the DATE curve-- are repeated with this placebo
intervention.

Figure~\ref{fig:DATE_DLM_placebo} displays the resulting placebo-estimated treated and control trajectories
(top panel) and the corresponding DATE trajectory (bottom panel). As expected, the
estimated DATE remains centered around zero with uncertainty bands that contain
zero throughout the post-placebo period. No statistically or substantively meaningful
dynamic effect is detected. This confirms that the dynamic linear model and DATE
estimation procedure do not artificially generate causal structure when no
intervention has occurred.

Taken together with the main analysis, this placebo result provides additional
assurance that the observed dynamic treatment effect is attributable to the actual
policy intervention rather than to spurious time series features or model
misspecification.

\begin{figure}[t!]
\centering
\includegraphics[width=1\textwidth]{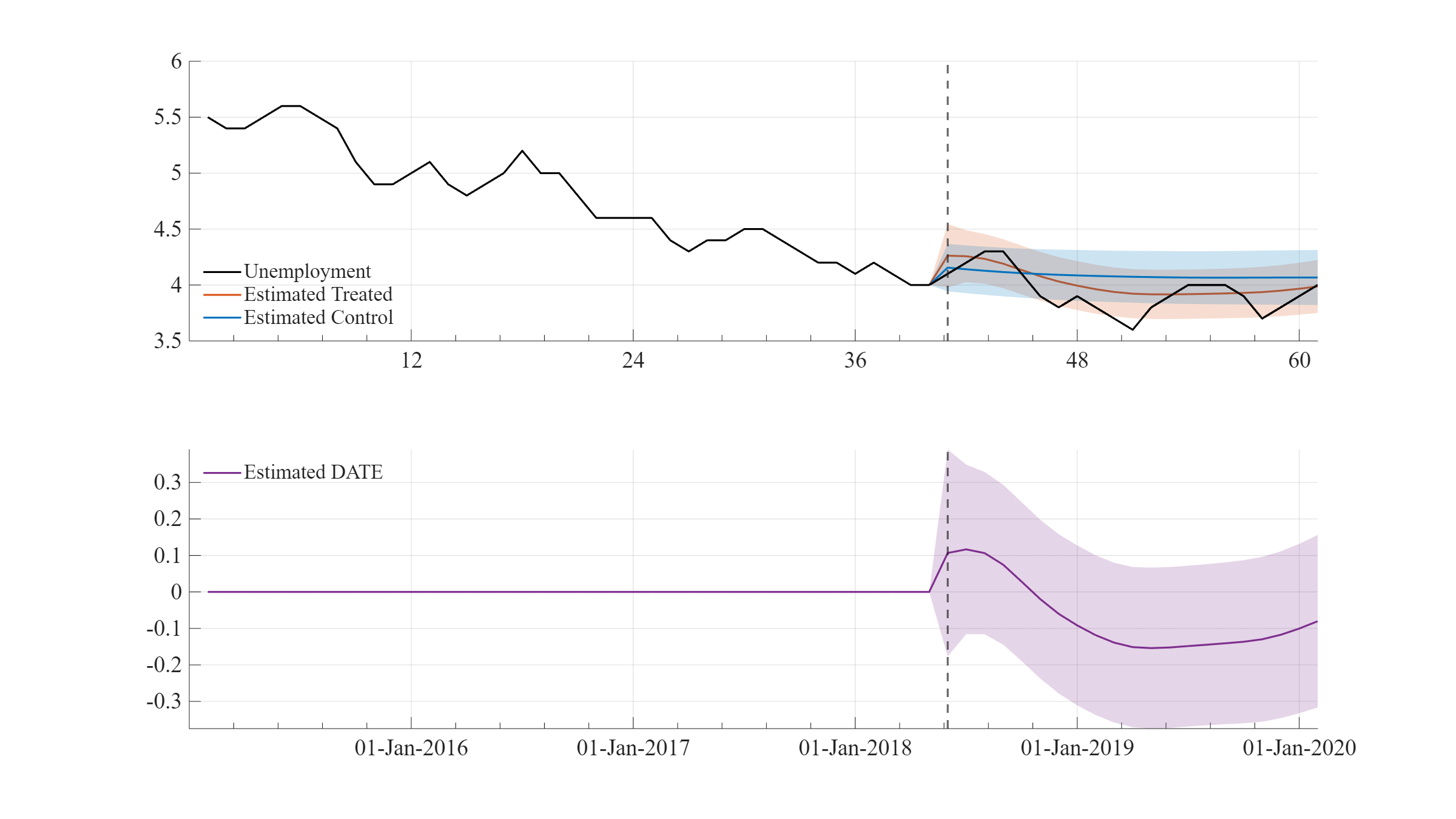}
\caption{Placebo intervention test: Top Figure: The estimated mean path of treated vs estimated mean control path of the One-None case for UK unemployment, where the random quasi-intervention occurs at a time strictly before the actual COVID-19 intervention. Bottom Figure: DATE of the quasi-intervention.}
\label{fig:DATE_DLM_placebo} 
\end{figure}

\end{document}